\documentclass[12pt]{article}
\usepackage{anyfontsize}
\usepackage{authblk}
\usepackage{amstext}
\usepackage{amssymb}
\usepackage{mathtools}
\usepackage{amsthm}
\usepackage{bibentry}
\usepackage{comment}
\usepackage{array}
\usepackage{booktabs}
\usepackage{multirow}
\usepackage{xfrac}
\usepackage{subfig}
\usepackage{graphicx}
\usepackage{tabularx}
\usepackage{float}
\usepackage[compact]{titlesec}
\usepackage{makecell}

\usepackage{natbib}
\usepackage{caption}
\usepackage{enumerate}
\usepackage{enumitem}

\usepackage{amsmath}
\usepackage{mathrsfs}
\usepackage{bm}
\usepackage{dsfont}
\usepackage{setspace}
\usepackage{color,xcolor}                          

\usepackage[colorlinks=true,allcolors=blue]{hyperref}

\usepackage[capitalize,noabbrev]{cleveref}
\crefname{assumption}{assumption}{assumptions}
\Crefname{assumption}{Assumption}{Assumptions}
\crefname{equation}{Eqn}{Eqns}
\Crefname{equation}{Equation}{Equations}

\DeclareMathOperator*{\argmax}{arg\  max}

\newtheorem{theorem}{Theorem}
\newtheorem{proposition}{Proposition}

\newtheorem{lemma}{Lemma}
\newtheorem{assumption}{Assumption}
\newtheorem{definition}{Definition}

\newtheorem{remark}{Remark}

\usepackage[toc,page,header]{appendix}
\usepackage{minitoc}
\noptcrule 

\usepackage{xparse}

\newcommand{\ppi}{\bm\pi}

\newcommand{\varub}{\widehat{\mathrm{Var}}^{\mathrm{UB}}}
\newcommand{\mat}[1]{{\bm{#1}}}

\newcommand{\lWt}[1][]{\overline{\bm W}_{t#1}}

\newcommand{\lHt}[1][]{\overline{H}_{t#1}}

\newcommand{\lct}[2]{\overline{#1}_{#2}}

\newcommand{\dto}{\overset{d}{\to}}
\newcommand{\pto}{\overset{p}{\to}}
\newcommand{\shiftmean}{\frac{1}{T-M+1}\sum_{t=M}^T}
\DeclareMathOperator*{\plim}{plim}
\newcommand{\ttheta}{\bm \theta}
\newcommand{\ggamma}{\bm \gamma}
\newcommand{\bbeta}{\bm \beta}
\newcommand{\mmu}{\bm \mu}

\newcommand{\aat}{\bm A_t}

\newcommand{\ft}{\mathcal{F}_{t-1}}

\newcommand{\Var}{\mathbb{V}}
\newcommand{\Cov}{\mathrm{Cov}}
\newcommand{\E}{\mathbb{E}}
\renewcommand{\P}{\mathbb{P}}

\DeclareDocumentCommand{\hatV}{ O{} O{\bbeta} O{\ppi^{(M)}} }{%
  \hat{V}^{#1(#2)} (#3)
}
\DeclareDocumentCommand{\hatVPS}{ O{} O{\bbeta} O{\ppi^{(M)}} O{\ttheta} O{\hat}}{%
  \hat{V}^{#1(#2)} (#3;#5 #4)
}
\DeclareDocumentCommand{\hatVTruePS}{ O{} O{\bbeta} O{\ppi^{(M)}} O{\ttheta^\ast} }{%
  \hat{V}^{#1(#2)} (#3;#4)
}

\DeclareDocumentCommand{\tildeV}{ O{\bbeta} O{\ppi^{(M)}} O{}}{%
  \widetilde{V}^{(#1)}#3 (#2)
}
\DeclareDocumentCommand{\tildeVPS}{ O{\bbeta} O{\ppi^{(M)}} O{\ttheta} O{} O{\hat}}{%
  \widetilde{V}^{(#1)}#4(#2;#5 #3)
}
\DeclareDocumentCommand{\tildeVTruePS}{ O{\bbeta} O{\ppi^{(M)}} O{\ttheta^\ast} O{} }{%
  \widetilde{V}^{(#1)}#4(#2;#3)
}

\DeclareDocumentCommand{\projVt}{ O{\bbeta} O{\ppi^{(M)}} O{-M} }{%
  V_t^{#2,(#1)}(\lHt[#3])
}
\DeclareDocumentCommand{\projV}{ O{\bbeta} O{\ppi^{(M)}} O{-M} }{%
  V^{#2,(#1)}(\lHt[#3])
}

\DeclareMathOperator{\expit}{expit}

\newcommand{\rhobar}{\bar\rho}
\DeclareDocumentCommand{\rhot}{ O{\bbeta} O{\ppi^{(M)}} O {} }{%
  \rho_{t#3}^{(#1)}(#2)
}
\DeclareDocumentCommand{\rhotPS}{ O{\bbeta} O{\ppi^{(M)}} O{\ttheta} O{\hat}}{%
  \rho_t^{(#1)}(#2;#4 #3)
}

\DeclareDocumentCommand{\rhotTruePS}{ O{\bbeta} O{\ppi^{(M)}} O{\ttheta^\ast}}{%
  \rho_t^{(#1)}(#2;#3)
}

\usepackage{xparse}

\renewenvironment{proof}[1][\proofname]{{\noindent\bfseries #1. }}{\qed}
\begin{document}

\def\spacingset#1{\renewcommand{\baselinestretch}%
{#1}\small\normalsize} \spacingset{1}

%%%%%%%%%%%%%%%%%%%%%%%%%%%%%%%%%%%%%%%%%%%%%%%%%%%%%%%%%%%%%%%%%%%%%%%%%%%%%%

\title{\bf Dynamic Policy Evaluation and Learning with Spatio-temporal Data}

  \author{Lingxiao Zhou\thanks{Department of Biostatistics, Yale University} \quad
    Kosuke Imai\thanks{Department of Government and Department of Statistics, Harvard University} \quad
    Jason Lyall\thanks{Department of Government, Dartmouth College} \quad
    Georgia Papadogeorgou\thanks{Department of Statistics, University of Florida}}
  \maketitle
 
\date{}

\maketitle

\begin{abstract}
Although sequential decision-making is ubiquitous across domains, policy evaluation and learning with spatio-temporal data remain challenging due to spatial spillover and temporal carryover effects. We develop methods for evaluating and learning individualized dynamic policies under spatio-temporal interference.  Under a semiparametric additive outcome model that allows for complex spillover and carryover effects, we consider a family of stabilized estimators for evaluating the performance of a given individualized policy.  From this family, we select a data-adaptive optimal estimator that minimizes the asymptotic variance.  We then derive the asymptotic distribution of the proposed policy evaluation estimator, and establish the finite-sample regret bounds of our policy learning estimator. We further propose a statistical test to select the complexity of the semiparametric additive model by determining the appropriate order of interactions. Through simulations, we assess the finite-sample performance of our estimators and the validity of the proposed test. Our motivating application examines the optimal allocation of economic aid in Iraq from February 2007 to July 2008. Drawing on declassified conflict data, we study the weekly assignment of aid projects across districts and find that reallocating aid away from regions with a persistently high level of violence can substantially reduce insurgent attacks.

\end{abstract}

\noindent%
\vfill

\newpage
\spacingset{1.9} % DON'T change the spacing!
\section{Introduction}

A number of scholars across a wide range of disciplines have long studied how to design optimal policies to maximize an expected outcome \citep[e.g.,][]{imai2011estimation,zhao2012estimating,swaminathan2015counterfactual,kitagawa2018should,kallus2021minimax,athey2021policy,zhang2022safe,cui2025policy}. 

Building on this extensive literature, a growing body of work considers dynamic policy learning, where decision rules are updated adaptively as new information becomes available \citep[e.g.,][]{Murphy2003,moodie2007demystifying,kasy2021adaptive,weltz2022reinforcement,kock2022functional,mcdonald2024sequential,sakaguchi2025estimation,cesa2025adaptive}. 
%Such problems naturally arise when decisions are made repeatedly over time. For instance, designing dynamic treatment regimes is common in healthcare \citep{chakraborty2011dynamic} and education \citep{almirall2018developing}, where treatment is assigned at each time point according to evolving covariates and past outcomes.  
Although existing studies address methodological challenges arising from carryover effects, where past treatments have lingering impacts, they often do not account for spillover effects across units \citep[e.g.,][]{kitagawa2022policy,moodie2012q}. 

A small number of recent works develop policy evaluation and learning methods when the outcome of one unit is affected by another unit's treatment \citep[e.g.][]{viviano2025policy,park2024minimum,zhang2023individualized}. However, they impose structural assumptions, such as anonymous interference, where spillover effects depend only on the number of treated units, or partial interference, where no spillover effects occur across clusters.
Importantly, existing approaches neither accommodate dynamic policies nor account for carryover effects.

These limitations highlight the need for methods that handle {\it both} temporal carryover and spillover effects at the same time under flexible structural assumptions.
In this paper, we develop a methodological framework for individualized dynamic policy evaluation and learning under spatio-temporal interference (\cref{sec:setting}). Our approach builds on a semiparametric additive outcome model \citep{yu2022estimating,zhang2023individualized} to reduce the high-dimensionality of the treatment assignment in the presence of interference across units. In its basic form, the potential outcomes at each time period are assumed to be additive in the current treatments of all units, given the information observed up to that time period, while allowing outcomes to depend arbitrarily on past treatments and covariates. Under our framework, the treatment of a unit can affect the outcome of any unit at any time period. 

Leveraging the semiparametric additivity structure, we construct estimators for evaluating the empirical performance of a given policy using linearized inverse probability weighting (IPW) (\cref{sec:estimation}). However, even under the semiparametric additive model, the variance of the weighting estimator can remain high, particularly when evaluating policies that span multiple time periods. 
To further improve estimation efficiency, we introduce a family of mixing estimators that combine two types of normalized weights, and select the optimal mixing weight by minimizing the asymptotic variance of the resulting estimator. We then employ this evaluation estimator for policy learning (\cref{sec:policy_learning}). The proposed mixing construction yields stable performance in policy evaluation and improves efficiency in policy learning. This new estimator also mitigates the propensity overfitting problem discussed in \cite{swaminathan2015self}, where if unnormalized weights are used directly in the policy learning step, the optimization algorithm tends to favor larger weights by selecting policies that closely mimic the observed treatment assignments.

We establish the consistency of the estimator and derive its asymptotic distribution.
We find the optimal policy by maximizing the estimated value of policy over a prespecified policy class \citep{kitagawa2018should,zhou2023offline}.  We establish that the average outcome under the learned policy converges to that of the optimal policy at the rate of $1/\sqrt{T}$, matching the optimal regret rate for policy learning under the i.i.d. setting \citep{kitagawa2018should,athey2021policy}.
We also extend the semiparametric additive model to include interactions and develop a statistical test to determine the order of interactions (\cref{sec:interaction}). 

Finally, we evaluate the methodology through simulations (\cref{sec:simulation}), which show that policies learned from the proposed estimator outperform those based on the standard IPW estimator. We further apply our approach to wartime data from Iraq (\cref{sec:application}), and derive an optimal allocation of aid over multiple time periods between districts to reduce insurgent attacks. The resulting optimal policies shift aid away from regions with persistently high levels of violence and are estimated to perform significantly better than the realized aid allocation.

\paragraph{Related Literature}

The methods for policy learning in the spatio-temporal setting remain underdeveloped. To the best of our knowledge, the only existing work in the literature is that of \citet{laber2018optimal}. The authors focus on maximizing a cumulative outcome across a long time horizon, using a Bayesian online method that continuously updates the policy as new data become available. Their method relies on a parametric, Markovian state-space model, in which spatial interference and temporal carryover effects are explicitly encoded through low-dimensional transition dynamics. 

In contrast, our approach learns the optimal policy over relatively short time periods from an offline dataset, and does not require specifying a parametric model for the entire data-generating process. Instead, spatial interference and temporal carryover effects are accommodated through an additive semiparametric outcome model.

Most existing work on policy learning considers static or cross-sectional settings with independent units and no interference \citep[e.g.,][]{manski2004statistical,dehejia2005program,qian2011performance,zhang2012estimating,zhao2012estimating,imai2023experimental,kitagawa2018should,jin2022policy,kallus2018balanced,mbakop2021model,athey2021policy,zhou2023offline}. A smaller body of research addresses cross-sectional policy learning with interference across units, typically imposing structural assumptions on the interference structure \citep[e.g.,][]{ananth2020optimal,zhang2023individualized,park2024minimum,viviano2025policy}. Among these, \citet{zhang2023individualized} is particularly relevant. They use a similar additive structure on the potential outcomes to stabilize IPW weights. However, like other works, their method is not applicable to temporal data with carryover effects.

There is growing interest in policy learning for time series settings. \citet{adusumilli2019dynamically} studies optimal allocation with budget constraints when subjects arrive sequentially, and \citet{kitagawa2022policy} considers optimal one-period policies based on information observed in prior periods. Closely related is the literature on dynamic treatment regimes (DTRs), which examines treatment strategies that unfold over multiple time points. At each time point, the treatment assigned to an individual may depend on the observed history up to that time, including past treatments, covariates, and outcomes. Most existing works, however, assume independent units and do not allow for spillover effects \citep[e.g.,][]{laber2014dynamic,barrett2014doubly,tao2018tree,simoneau2020estimating,sakaguchi2025estimation}. 

Recent work extends DTR methods to accommodate interference, allowing outcomes to depend on the treatments received by neighboring units in a known network \citep{su2019modelling,jiang2023dynamic,rizi2024dynamic}. These approaches typically rely on a fixed and correctly specified network structure applied to data observed over a short time horizon. In contrast, we consider a single, long time series with arbitrarily connected units and do not require knowledge of the underlying network. 

Another related area is multi-armed bandit algorithms, which optimize cumulative outcomes over many rounds by balancing exploration and exploitation \citep{lattimore2020bandit,kock2022functional,adusumilli2025risk}. Importantly, standard bandit models do not consider interference or temporal carryover effects, assuming each action only affects its immediate reward and has no impact on other units or future rewards.

There also exists the broader causal inference literature on interference and temporal dependence.
In the general interference setting, many studies impose structural assumptions \citep[e.g.,][]{liu2016inverse,aronow2017estimating,basse2018analyzing,athey2018exact,wallace2019model,wang2020design,li2022random,giffin2023generalized,gao2023causal}. Relatedly, \citet{harshaw2023design} introduce a linear exposure-response model for bipartite experiments to address interference in two-sided marketplaces.
In time series settings, some studies explicitly incorporate temporal carryover effects in the outcome model \citep{almirall2010structural,brodersen2015inferring,liang2025randomization}, while others avoid modeling outcome dynamics \citep{sun2021estimating,callaway2021difference,kitagawa2022policy}.
In spatio-temporal settings, \citet{papadogeorgou2022causal} proposed a framework that accommodates both spillover effects across units and temporal carryover effects.
However, all of these works focus solely on estimating causal effects of potential interventions, rather than deriving optimal policy rules.

In addition, our work relates to the literature on treatment effect heterogeneity, which is often used to inform policy learning. In the standard setting with i.i.d. data, a large body of work studies heterogeneous treatment effects using flexible modeling and machine learning approaches \citep[e.g.,][]{athey2016recursive,wager2018estimation,chernozhukov2018double,kunzel2019metalearners}. 
Recently, these frameworks have been adapted to handle complex dependencies by imposing specific structural constraints. In interference and temporal settings, many studies rely on pre-defined structures, such as clustered interference \citep{bargagli2025heterogeneous}, social network ties \citep{park2024groupwise}, fixed temporal lead-lag patterns \citep{sun2021estimating}, or group-level fixed effects \citep{kattenberg2023causal}. In spatio-temporal contexts, \citet{zhang2023spatiotemporal} utilizes a structural outcome model to incorporate interference, whereas \citet{zhou2024estimating} allow for arbitrary spatial spillover and temporal carryover effects. While these works provide tools for identifying where treatment effects vary, they do not directly target the learning of optimal policies.

\section{The Problem Formulation}
\label{sec:setting}

We begin by describing the setup and formalizing the problem of policy evaluation and learning with spatio-temporal data.

\subsection{Spatio-Temporal Data Structure}

We consider data over $T$ time periods and $n$ arbitrarily connected units. For each unit $i \in \{1,\ldots,n\}$ and time $t \in \mathcal{T}=\{1,\ldots,T\}$, let $W_{ti}$ be the binary treatment assigned to unit $i$ at time $t$, and let $\bm W_t = (W_{t1},\dots,W_{tn})^\top$ be the treatment vector for all units at that time. We use $\overline{\bm W}_t$ to denote the treatment path from time $1$ to $t$ for all units, where $\overline{\bm W}_t = \left(\bm W_1,\dots,\bm W_t\right)\in\mathcal{W}_t=\{0,1\}^{n\times t}$. A particular realization of $\overline{\bm W}_t$ is denoted by $\overline{\bm w}_t$.

We define $Y_{ti}(\overline{\bm w}_t)$ as the potential outcome for unit $i$ at time $t$ under treatment history $\overline{\bm w}_t$. Therefore, potential outcomes may depend on the entire past treatment path for all units, and we do not make any structural assumptions about spatial spillover or temporal carryover effects.  The observed outcome for unit $i$ at time $t$ is $Y_{ti} = Y_{ti}(\overline{\bm W}_t)$, and $\bm Y_t = (Y_{t1},\dots,Y_{tn})^\top$ denotes the observed outcome vector across all units at that time. Let $\overline{\mathcal{Y}}_T = \{\bm Y_{t}(\overline{\bm w}_t): \overline{\bm w}_t \in \mathcal{W}_t, t \in \mathcal{T}\}$ denote the collection of all potential outcomes across all time periods and treatment paths.

Similarly, let $\bm X_{ti}(\overline{\bm w}_{t-1})$ be a $p$-dimensional covariate vector for unit $i$ at time $t$, with observed value $\bm X_{ti} = \bm X_{ti}(\overline{\bm W}_{t-1})$. Then, a unit's covariates are allowed to depend on past treatments of any unit. Define $\bm X_t = (\bm X_{t1},\dots,\bm X_{tn})$ as the $p \times n$ covariate matrix at time $t$ across the $n$ units. The covariate matrix $\bm X_t$ is realized after $\bm W_{t-1}$ but prior to $\bm W_{t}$.  The covariate history up to time $t$ is denoted by $\overline{\bm X}_t$, and $\overline{\mathcal{X}}_T = \{\bm X_t(\overline{\bm w}_{t-1}): \overline{\bm w}_{t-1} \in \mathcal{W}_{t-1}, t \in \mathcal{T}\}$ is the collection of potential covariate values for all time periods and treatment paths. 

Finally, let $\lHt = \{\overline{\bm W}_t, \overline{\bm Y}_t, \overline{\bm X}_{t+1}\}$ denote the observed history preceding treatment assignment at time $t+1$, and the complete history $\lHt^\ast = \{\lWt,\overline{\mathcal{Y}}_T,\overline{\mathcal{X}}_T\}$ which includes all counterfactual values of the outcome and the covariates. Our setting involves a time series of $n$ connected units, for which potential outcomes and potential covariate values are viewed as fixed quantities. Therefore, in our framework, randomness arises only from the treatment assignment over time.

% {\bf [[GP: 1) Discuss that the notation we use for pot out implies that we allow arbitrary spatial spillover and temporal carryover effects. 2) Last paragraph: you defined observed history but then discuss the history including counterfactual quantities. Also the full history should include past and future outcomes and covariates.]]}

\subsection{Policies and their Value}\label{subsec: policyvalue}

A one-time-period policy $\ppi$ specifies the assignment of treatments across the units for a single time period given what has been observed. It is a vector-valued function that maps the observed history $\lHt[-1]$ to a treatment decision $\bm W_t^{\ppi}$ over all the units, i.e.,
\[
\bm W_t^{\ppi} = \ppi\left(\lHt[-1]\right) = \left(\pi_{1}\left(\lHt[-1]\right),\dots,\pi_{n}\left(\lHt[-1]\right)\right) \in \{0,1\}^n.
\] 
We use $\pi_{i}\left(\lHt[-1]\right)\in\{0,1\}$ to denote the policy decision for unit $i.$ Since the outcome for each unit may depend on its own treatment and the treatments assigned to others, one may use a different treatment rule for each unit.

The policy rule $\ppi$ is assumed to be fixed across time. To ensure that it is well defined across all time periods, we consider policies that depend only on information from a fixed number of most recent time periods. Instead, if a policy was allowed to depend on the entire history, the dimension of inputs to the treatment rule $\ppi$ would grow over time. 
%, preventing a single fixed mapping $\ppi$ from being used throughout. 
Importantly, this restriction applies only to the policies we will learn and does not imply any restriction on the underlying potential outcomes and realized treatments, which may depend on the full treatment history.

% We define the value of policy $\ppi$ for unit $i$ at time $t$ as the potential outcome that would be realized if the policy were to be implemented at time $t$ conditional on the history,
% \begin{equation}\label{eq: policy_value_single}
% Y_{ti}^{\ppi}\left(\lHt[-1]\right) = Y_{ti}\left(\overline{\bm W}_{t-1}, \bm w_t = \ppi\left(\lHt[-1]\right)\right).    
% \end{equation}
We define the potential outcome for unit $i$ at time $t$ under policy $\ppi$ as the outcome that would be realized if the policy were implemented at time $t$ conditional on the history,
\begin{equation}\label{eq: policy_value_single}
Y_{ti}^{\ppi}\left(\lHt[-1]\right) 
= Y_{ti}\left(\overline{\bm W}_{t-1}, \bm w_t = \ppi\left(\lHt[-1]\right)\right).    
\end{equation}
Averaging across all units and time periods, we define the policy value $V^{\ppi}$ as:
\begin{equation}\label{eq:policy_value}
 V^{\ppi}=\frac{1}{Tn} \sum_{t=1}^T\sum_{i=1}^n Y_{ti}^{\ppi} \left(\lHt[-1]\right) .
\end{equation}
Thus, the policy value measures the overall performance of a policy across all units and over all time periods. 

We next consider dynamic policies that specify individualized treatment rules over multiple time periods. 
Let 
$
\ppi^{(M)} = (\ppi_1,\dots,\ppi_M)
$
denote a policy that determines the treatment decision over $M$ consecutive periods.
 Under the  policy $\ppi^{(M)}$, the treatment for unit $i$ at time $t$ is
assigned according to $\pi_{1i}\left(\lHt[-1]\right)$, at time $t-1$ according to $\pi_{2i}(\lHt[-2])$, continuing until time period $t-M+1$, for
which treatment is assigned according to $\pi_{Mi}(\lHt[-M])$.

To define the value of policy $\ppi^{(M)}$, we first consider a two-time-period policy $(\ppi_1,\ppi_2)$ for simplicity. We define the potential outcome for unit $i$ at time $t$ under this policy as
\begin{align*}
    Y_{ti}^{\ppi_1,\ppi_2}\left(\lHt[-2]\right)  = Y_{ti}^{\ppi_1}\left(\lHt[-1]^{\ppi_2}\right)
     = Y_{ti}(\lWt[-2],\bm w_{t-1}=\ppi_2\left(\lHt[-1]\right),\bm w_t =\ppi_1(\lHt[-1]^{\ppi_2}))
\end{align*}
where $\lHt[-1]^{\ppi_2} = \left\{\lHt[-2],\ppi_2(\lHt[-2]),\overline{\bm Y}_{t-1}(\lWt[-2],\bm w_{t-1} = \ppi_2(\lHt[-2])),\overline{\bm X}_{t}\left(\lWt[-2],\bm w_{t-1} = \ppi_2\left(\lHt[-2]\right)\right)\right\}$ denotes the history that would have been observed up to time $t-1$ under the scenario, in which
$W_{t-1}$ were to be assigned according to $\ppi_2$.
Therefore, the quantity $Y_{ti}^{\ppi_1,\ppi_2}\left(\lHt[-2]\right)$ represents the potential outcome at time $t$ for unit $i$ when the treatment at time $t-1$ is assigned according to the policy $\ppi_2$ based on $\lHt[-2]$, and the treatment at time $t$ is assigned according to the policy $\ppi_1$ based on $\lHt[-1]^{\ppi_2}$. Continuing this process, we can define the value of a policy $\ppi$ that is assigned over $M$ time periods, $\ppi^{(M)} = (\ppi_1, \ppi_2, \dots, \ppi_M)$, on unit $i$ at time $t$, $Y_{ti}^{\ppi^{(M)}}\left(\lHt[-M]\right)$. Then, the overall value of this policy is defined as
$$
V^{\ppi^{(M)}} = \frac{1}{(T-M+1)n}\sum_{t=M}^T\sum_{i=1}^n Y_{ti}^{\ppi^{(M)}}\left(\lHt[-M]\right).
$$
%where $Y_{ti}^{\ppi}\left(\lHt[-M]\right)$ is the policy value for each unit $i$ and time period $t$.  

Therefore, our formulation of the policy value for a multiple-time-period policy follows a backward sequential manner in agreement with the dynamic treatment regime literature \citep[e.g.][]{moodie2012q,chakraborty2014dynamic,sakaguchi2024policy}.
This formulation considers the outcome only in the last of the $M$ time periods. Alternatively, one can define the policy value to incorporate outcomes from all $M$ periods (see \cref{a:sec:alternative_policy_value}).

For policy learning, we define a class of policies, which determine its structure. For one-time-period policies, we denote this class by $\Pi$. For policies over $M$ time periods, one may consider a different policy class at a different time period. For notational simplicity, however, we focus on the case where the same policy class $\Pi$ is used at each period, and we denote it  by $\Pi^{M}$.  Given a policy class, an optimal policy, denoted by $\ppi^{(M)\ast}$, is given by
\begin{equation}
\ppi^{{(M)}\ast} = \argmax_{\ppi^{(M)} \in \Pi^M} V^{\ppi^{(M)}}.
\end{equation}
Finally, the regret of a candidate policy $\ppi^{(M)}$ is defined as the difference between its value and that of optimal policy $\ppi^{{(M)}\ast}$, i.e.,
$
R(\ppi) = V^{\ppi^{{(M)}\ast}} - V^{\hat\ppi^{(M)}}.
$

There are many possible choices of policy classes in the literature. In this paper, we focus on a linear policy class, which has been widely used due to its simplicity and interpretability \citep{zhang2013robust,zhang2023individualized,kitagawa2022policy}. Specifically, we consider linear policies that depend only on information from the previous time period:
\begin{equation}\label{eq: linearclass}
 \Pi = \left\{\ppi:\pi_{i}\left(\lHt[-1]\right) = \mathds{1}(\bm X_{ti}^P\ggamma> 0),\ggamma\in\mathbb{R}^p\right\},   
\end{equation}
where $\bm X_{ti}^P$ is constructed from $\lHt[-1]$. 

% \[
% \Pi = \Big\{ \ppi : \pi_{i}\left(\lHt[-1]\right) = \mathds{1}\big(\gamma_0 + \bm \gamma_1^\top \bm X_{ti} + \gamma_2 Y_{t-1,i} + \gamma_3 W_{t-1,i} \ge 0\big), \; i=1,\dots,n \Big\}.
% \]
Other commonly used policy classes include decision trees and decision tables \citep{athey2021policy,ben2025safe,jia2023bayesian,zhou2023offline}.

% {[[GP: 1) Include the examples of policy classes in the end of the section. 2) we need a little more fluff. For example, you define the policy value, but you should give some intuition about what that means, and discuss how interference plays a role. 3) we say that we restrict the time horizon for optimal policy. A careless reader might think that we say that we restrict temporal carryover effects. We need to be clear that we do not restrict the dependence of potential outcomes on all previous treatments.]]}

\section{Policy Evaluation}
\label{sec:estimation}

Before discussing policy learning, we consider policy evaluation, that is, the estimation of the policy value $V^{\ppi^{(M)}}$ for any candidate policy $\ppi^{(M)}$. 
% We next introduce an estimator that leverages semiparametric outcome model assumptions to improve the efficiency of the standard IPW estimator. 

\subsection{Assumptions}

For any given unit and time period, we observe only one potential outcome under the realized treatment. Thus, learning causal estimands from observed data requires non-verifiable assumptions. We consider the unconfounded treatment assignment mechanism.
\begin{assumption}[Unconfoundedness]\label{assump: unconfoundedness} \singlespacing
Conditional on the observed history $\lHt[-1]$, the assignment of treatments at time $t$ is independent of all counterfactual outcomes and covariates across all time periods, i.e.,
  $\P(\bm W_t\mid\lHt[-1],\lct{\mathcal{Y}}{T},\lct{\mathcal{X}}{T}) = \P(\bm W_t\mid \lHt[-1])$.
\end{assumption}
Assumption~\ref{assump: unconfoundedness} is more restrictive than the standard sequential ignorability assumption, which only requires that the treatment assignment at time $t$ is independent of future potential outcomes \citep[e.g.,][]{robins2000marginal}.
This stronger restriction is necessary since we only observe a single time series of arbitrarily connected units.

We refer to $e_t(\bm w_t) := \P(\bm W_t = \bm w_t\mid \lHt[-1])$ in Assumption~\ref{assump: unconfoundedness} as the {\it propensity score}. 
In addition, the unit-level propensity score at time $t$ is defined as $e_{ti}(W_{ti}) := \P(W_{ti} = w_{ti}\mid \lHt[-1]),$ where $\lHt[-1]$ includes the observed history of all units. We make the following factorization assumption on the propensity score.
\begin{assumption}[Factored time-specific propensity score]\label{assump: factor_ps} \singlespacing At time period $t$, the propensity score satisfies 
$e_t(\bm w_t) = \prod_{i=1}^ne_{ti}(w_{ti})$. In addition, there exists $\eta>0$ such that $\eta<e_{ti}(w_{ti})<1-\eta$ for any $w_{ti}\in\{0,1\}.$
\end{assumption}

\Cref{assump: factor_ps} imposes two conditions. 
First, given the observed history, the treatment assignment at each time period is independent across units.
We do not assume marginal independence, allowing the treatment assignment of one unit to depend on the entire observed history involving other units $\lHt[-1]$.
% The assumption therefore restricts the assignment mechanism given the available information rather than asserting that units act independently in general.
Second, the positivity condition applies to each unit separately rather than the entire treatment vector across all units at time $t$, making the assumption realistic even if the number of units is large.

In the current setting, an additional assumption about the structure of interference is required for efficient estimation. We consider a semiparametric additive outcome model \citep{zhang2023individualized}. For a policy over a single time period, we assume that, conditional on the observed history, the outcome for each time period $t$ and unit $i$ depends on the treatment vector only through the interactions of treatments up to degree $\beta$.

\begin{assumption}[Semiparametric additive outcome model]\label{assump: additive_model} \singlespacing
    For $1\le t\le T$, $1\le i\le n$ and $\bm w_t\in\{0,1\}^n$, the potential outcomes satisfy
    \begin{equation}\label{eq: additive_model}
      Y_{ti}(\lWt[-1],\bm w_t) = \bm g_{ti}\left(\lHt[-1]\right)^\top\phi^{(\beta)}(\bm w_t),
   % \E\left[ Y_{ti}(\lWt)\mid \lHt[-1],\bm W_t = \bm w_t\right] = \bm g_{ti}\left(\lHt[-1]\right)^\top\phi^{(\beta)}(\bm w_t),
\end{equation}
where $\bm g_{ti}\left(\lHt[-1]\right)$ is a vector of unknown functions and $\phi^{(\beta)}(\bm w_t)$ is an augmented treatment vector at time $t$ that includes interactions up to $\beta$-order between treatments of different units with $\beta\le n$, i.e.,
$\phi^{(\beta)}(\bm w_t) = \left(1,\{w_{ti}\}_i,\{w_{ti_1}w_{ti_2}\}_{i_1\neq i_2},\dots,\{w_{ti_1}w_{ti_2}\dots w_{ti_\beta}\}_{i_1\neq\dots \neq i_\beta}\right)^\top$.
\end{assumption}
\Cref{assump: additive_model} allows for flexible outcome models.  Each function $g_{ti}\left(\lHt[-1]\right)$, which characterizes treatment effect heterogeneity, is allowed to vary across time periods and units, while no restriction is imposed on its dependence on the observed history.  The order of interaction can be adjusted based on the desired levels of model complexity and interpretability (see \cref{sec:interaction} for a data-driven order selection methodology). When $\beta=n$, the assumption does not impose any restriction on the outcome model.

For policies over multiple time periods, we impose a sequence of additive outcome
models, % since the policy value is defined iteratively through the expected value of shorter policy sequences and the additive structure may not be preserved in this process. 
similarly to the policy value which is defined iteratively.
\begin{assumption}[Sequential additive outcome models]\label{assump: additive_model_multi} \singlespacing
Let $\Pi$ be a policy class. For all $M\le t\le T$, $1\le i\le n,$ $\bm w_{t-M+1},\dots, \bm w_t\in\{0,1\}^n$ and $(\ppi_1,\dots,\ppi_{M-1})\in\Pi^{M-1}$, we assume 
\begin{align*}
       Y_{ti}(\lWt[-1],\bm w_t) &= \bm g_{1ti}\left(\lHt[-1]\right)^\top\phi^{(\beta_1)}(\bm w_t)\\
       Y_{ti}^{\ppi_1}(\lHt[-1]^{\bm w_{t-1}}) &= \bm g_{2ti}\left(\lHt[-2]\right)^\top\phi^{(\beta_{2})}(\bm w_{t-1})\\
       \vdots \\
       Y_{ti}^{\ppi_1,\dots,\ppi_{M-1}}(\lHt[-M+1]^{\bm w_{t-M+1}})
       &= \bm g_{Mti}(\lHt[-M])^\top\phi^{(\beta_{M})}(\bm w_{t-M+1}),
\end{align*}
where $\overline{H}_j^{\bm w_j} = \{\overline{H}_{j-1},\bm w_{j},\overline{\bm Y}_j(\overline{\bm W}_{j-1},\bm w_j),\overline{\bm X}_{j+1}(\overline{\bm W}_{j-1},\bm w_j)\}$
denotes the history that would have been observed up to time $j$ under the scenario, in which the treatment assignment at $j$ is $\bm w_j$, $\{\bm g_{mji}(\overline{H}_{j-m})\}$ are vectors of unknown heterogeneous effect functions, and $\{\phi^{(\beta_m)}(\bm w_{j-m+1})\}$ are augmented treatment vectors at time $j-m+1$ that contain up to $\beta_m$-order interactions between treatments of different units with $\beta_m\le n$, defined analogously to \Cref{assump: additive_model}.
\end{assumption}

% \begin{align*}
%        Y_{ti}(\lWt[-1],\bm w_t) &= \bm g_{1ti}\left(\lHt[-1]\right)^\top\phi^{(\beta_1)}(\bm w_t)\\
%        V_{ti}^{\ppi_1}(\lHt[-1]^{\bm w_{t-1}}) &= \bm g_{2ti}\left(\lHt[-2]\right)^\top\phi^{(\beta_{2})}(\bm w_{t-1})\\
%        \vdots \\
%        V_{ti}^{\ppi_1,\dots,\ppi_{M-1}}(\lHt[-M+1]^{\bm w_{t-M+1}})
%        &= \bm g_{Mti}(\lHt[-M])^\top\phi^{(\beta_{M})}(\bm w_{t-M+1}),
% \end{align*}

\Cref{assump: additive_model_multi} represents the potential outcome at time $t$ as a function of the observed history and hypothetical treatment assignments. In particular, the $m$-th model characterizes the dependence of the potential outcome at time $t$ on the treatment vector at time $t-m$. For $Y_{ti}^{\ppi_1,\dots,\ppi_m}(\lHt[-m]^{\bm w_{t-m}})$, only the treatment vector at time $t-m$ is set to $\bm w_{t-m}$ in the counterfactual history, while the subsequent treatments from time $t-m+1$ to $t$ follow the policies $\ppi_m,\dots,\ppi_1$.
Thus, the assumption allows $Y_{ti}^{\ppi_1,\dots,\ppi_m}(\lHt[-m]^{\bm w_{t-m}})$ to depend on intermediate covariates and outcomes between time periods $t - m$ and $t$ that are functions of history including past treatments. 

However, \Cref{assump: additive_model_multi} restricts the way in which intermediate covariates, outcomes, and treatments between $t-m$ and $t$ can affect the outcome at time $t$ by placing a low-order interaction structure. In general, we can use different orders of interaction.  The choice of $\beta_m$ involves a bias–variance tradeoff \cite[e.g.][]{cortez2023exploiting,zhang2023individualized,sengupta2025low}.
In \cref{sec:interaction}, we introduce a  hypothesis test to sequentially select the values of  $\beta_1,\dots,\beta_M$.

Similar assumptions have appeared in the DTR literature \citep{chakraborty2014dynamic,jiang2023dynamic}, where researchers often fit sequential parametric models. 
In contrast, Assumptions~\ref{assump: additive_model}~and~\ref{assump: additive_model_multi} specify a semiparametric additive model on the potential outcomes. In the following section, we leverage this interference structure to construct an efficient estimator for policy values. 
% In contrast, we assume a more flexible structural additive representation and do not require sequential parametric estimation. 

In \cref{subsec:network_restriction}, we consider alternative outcome models of neighborhood interference where the parameter $\beta$ represents the neighborhood size of interference.
Furthermore, in \cref{subsec:block_dependence}, we allow for the dependence of treatment assignments across units.

% Additionally, we explore alternative assumptions where we consider a slightly more restrictive model based on the network structure, and relax the conditional unit-level independence assumption presented in \Cref{assump: factor_ps}. Details are in Appendix \ref{sec:alternative_assumptions}.

\subsection{Efficient estimator of the policy value}\label{subsec:single}

For a single-time-period policy $\ppi\in\Pi$, we define weights $\rho_t^{(\beta)}(\ppi)$ as
\begin{align}\label{eq: addIPW_weight}
\begin{split}
 \rhot[\beta][\ppi] = 1 + \sum_{i=1}^n (\rho_{ti}(\ppi) - 1) + \sum_{1 \leq i_1 < i_2 \leq n} (\rho_{ti_1}(\ppi) - 1)(\rho_{ti_2}(\ppi) - 1) + \cdots \\
 + \sum_{1 \le i_1 < \dots < i_\beta \le n} (\rho_{ti_1}(\ppi) - 1)\cdots(\rho_{ti_\beta}(\ppi) - 1), 
 \end{split}
\end{align} 
where $\rho_{ti}(\ppi) = \mathds{1}(W_{ti} = \pi_{i}\left(\lHt[-1]\right))/e_{ti}(W_{ti})$. Therefore, these weights are additive across groups of units that are allowed to interact in the potential outcome model given in \Cref{assump: additive_model}. 
For policies $\ppi^{(M)}$ defined over multiple time-periods, the weights are given by the following product of $M$ single-period weights,
$$
\rhot = \prod_{m=1}^M \rhot[\beta_m][\ppi_m][-m+1].
$$
% The additive formulation of weights for the single time period policy does not extend to policies over multiple periods. This multiplicative structure across time periods of additive time-specific weights reflects the sequential nature of the dynamic setting, where treatment decisions at each period depend on past history. 
%
We refer to these weights as the additive inverse propensity weights (additive IPW).

We consider a class of estimators that are convex combinations of the following two additive IPW estimators,
$$\hatV[H,]
= \frac{1}{T}\sum_{t=1}^T \frac{\rhot}{\rhobar}\ \overline{Y}_{t} \quad \text{and} \quad
\hatV[S,]
  = \frac{1}{T}\sum_{t=1}^T \left(\rhot
      + (1 - \rhobar)\right)\,\overline{Y}_t,
$$
where $\overline{Y}_t = \frac{1}{n}\sum_{i=1}^nY_{ti}$ is the average outcome across units at time $t$, and
$ \rhobar = \frac{1}{T}\sum_{t=1}^T\rhot
$ is the average additive IPW weight across time.
The first estimator $\hatV[H,]$ uses normalized H\'ajek weights and is expected to yield a low finite-sample variance than its un-normalized counterpart \citep{hajek1971comment,swaminathan2015self}. However, the additive weights $\rho_t^{(\bbeta)}(\ppi^{(M)})$ can be negative, making $\rhobar$ close to 0 and leading to instability. In contrast, the second estimator $\hatV[S,]$ normalizes the weights by adding $1 - \rhobar$ in their expression, making it linear in the weights and potentially more stable. However, this does not have the variance-stabilizing effect associated with weight normalization. 

Since neither estimator uniformly dominates the other, we consider the class of estimators that are their weighted averages, with weights that sum to 1,
\begin{equation}\label{eq: estimator_family}
   \mathcal{C}(\ppi^{(M)}) = \left\{\lambda\hatV[H,]+\left(1-\lambda\right)\hatV[S,]: \lambda\in[0,1]\right\}.
\end{equation}
This normalized weighted structure plays an important role in policy learning (see \cref{sec:policy_learning}).
The following theorem establishes that each addIPW estimator in $\mathcal{C}(\ppi^{(M)})$ is consistent and asymptotically normal.

% It straightforward to show that each estimator in \cref{eq: estimator_family} is shift-invariant, that is, if a constant is added to all outcome values, the estimator shifts by the same
% constant.
% This shift-invariance property helps mitigate propensity overfitting during policy learning \citep{swaminathan2015self}.

% To reduce estimation variance, we assume additive outcome models with optional interaction terms. We first introduce the assumption for estimating policy value over a single time period and then extend it to multiple time periods in \cref{subsec:multi}. We make the following assumption.
% which is expected to have a smaller variance than the standard IPW estimator. 
% Combining \Cref{assump: unconfoundedness,assump: factor_ps,assump: additive_model}, 

% In practice, employing normalized weights often improves the variance of the estimator and  An appealing feature of normalized estimators is that they are shift-invariant: if a constant is added to all outcome values, the estimator shifts by the same constant. We therefore focus on normalized, shift-invariant estimators of the value $V^{\ppi}$.

\begin{theorem}\label{lemma:addIPW_true} \singlespacing
Suppose that \Cref{assump: unconfoundedness,assump: additive_model_multi,assump: factor_ps,assump:Policy_est1} hold. Let
$\hat V^{(\beta)}(\ppi^{(M)},\lambda)$ be an arbitrary estimator in $\mathcal{C}(\ppi^{(M)}).$  Then
$$\sqrt{T-M+1}\left(\hat V^{(\bbeta)}(\ppi^{(M)},\lambda)-V^{\ppi^{(M)}}\right)\overset{d}{\to}N(0,\sigma^2_\lambda),$$ where $\sigma^2_\lambda=\plim_{T\to\infty}\frac{1}{T-M+1}\sum_{t=M}^T\Var\left(\rhot(\overline{Y}_t-\lambda \hatV-(1-\lambda)\overline{Y})\mid \lHt[-M]\right)$, and $\overline{Y} = \shiftmean \overline{Y}_t$.
\end{theorem}

% The estimator $\hatV[H,][\beta][\ppi]$ normalizes the weights $\rho_t^{(\beta)}(\ppi)$ by dividing with their mean. This normalization often leads to lower variance but may be unstable when the mean of the weights $\rhobar$ is close to zero. In contrast, $\hatV[S,][\beta][\ppi]$ shifts the weights by $1 - \rhobar$, leading to improved numerical stability compared to $\hatV[H,][\beta][\ppi]$, though potentially offering less variance reduction compared to $\hatV[H,][\beta][\ppi]$. Since neither estimator uniformly dominates the other, it is natural to consider combining them. 

% Based on \cref{lemma:addIPW_true}, all estimators in $\mathcal{C}(\ppi^{(M)})$ are consistent and asymptotically normal. 
We choose the optimal estimator by minimizing the asymptotic variance $\sigma^2_\lambda$ over $\lambda\in[0,1]$, leading to the following unique minimizer,
\begin{equation*}
 \lambda^\ast=\text{Proj}_{[0,1]}\left(\frac{\plim_{T\to\infty}\shiftmean \E[\rhot(\overline{Y}_t-\overline{Y})(\rhot-1)\mid\lHt[-M]]}{\plim_{T\to\infty}\shiftmean(V^{\ppi^{(M)}}-\overline{Y})\E(\rhot-1)^2\mid\lHt[-M]]}\right), 
\end{equation*}
where $\text{Proj}_{[0,1]}(x) = \min(\max(0,x),1).$ 
% See \cref{a:sec:lambda} for the proof.
Since the optimal value $\lambda^\ast$ depends on unknown population quantities, we use its consistent estimator $\hat\lambda$ obtained by substituting sample analogues of the conditional expectations. The resulting mixing estimator is
\begin{equation}
    \hatV = \hat\lambda\hatV[H,] + (1-\hat\lambda)\hatV[S,],
\label{eq:hatV_lambdahat}
\end{equation}
% where $\hat\lambda$ is a consistent estimator of $\lambda^\ast$. 

\cref{a:sec:lambda} provides a formal derivation of $\lambda^*$ and establishes the consistency of $\hat{\lambda}$.
Importantly, \cref{thm:addIPW_true} shows that replacing $\lambda^\ast$ with $\hat\lambda$ does not affect the asymptotic distribution of the resulting estimator. Therefore, using $\widehat \lambda$ leads to the most efficient estimator within the class $\mathcal{C}(\ppi^{(M)})$, and
%Since, by construction, the asymptotic variance at $\lambda^\ast$ is no larger than that of either endpoint estimator $\hatV[H,]$ or $\hatV[S,]$, 
the proposed mixing estimator $\hatV$ in \cref{eq:hatV_lambdahat} is asymptotically at least as efficient as $\hatV[H,]$ and $\hatV[S,]$.

% Even though replacing $\lambda^\ast$ with $\hat\lambda$ does not change its asymptotic distribution, using the estimated value $\lambda^\ast$ introduces finite-sample variability. Rather than allowing unrestricted linear combinations, we restrict $\lambda$ to $[0,1]$, so that the resulting estimator remains a convex combination of $\hatV[H,]$ and $\hatV[S,].$ This ensures that the mixing estimator remains within the interval determined by $\hatV[H,]$ and $\hatV[S,].$

A classic estimator of policy value is the standard IPW estimator, which has been used for policy evaluation and learning in a wide range of settings \citep{swaminathan2015counterfactual,zhang2012estimating,kitagawa2018should,papadogeorgou2022causal}. 
The additive IPW weights for the single-period policy in \cref{eq: addIPW_weight} correspond to a polynomial approximation to the standard IPW weights, where interaction terms are included up to order $\beta$. This construction exploits the additive 
outcome model (Assumption~\ref{assump: additive_model}). The additive IPW weights have variance no greater than that of the standard IPW weights unless $\beta=n$, in which case they are equivalent (see \cref{a:sec:IPW estimator} for details). Moreover, standard IPW weights rely on overlap for the full joint treatment vector. When the treatment pattern favored by the target policy is rarely or never observed, the standard IPW weights can be zero or nearly zero. The additive IPW weights do not require observing the entire target treatment pattern exactly, and can therefore remain nonzero in such settings.

\subsection{Extension to estimated propensity scores}

If the true propensity score is unknown, we use an estimated propensity score $e_{ti}(W_{ti};\hat\ttheta)$ to obtain $\rhotPS$ and $\hatVPS$
where $e_t(W_t;\ttheta)$ is a parametric propensity score model with $\hat{\ttheta}$ being the MLE of parameter vector $\ttheta \in \mathbb{R}^K$. Under the correct specification of propensity score model, the estimator remains consistent and asymptotically normal, as proved in the next theorem.
The regularity conditions  (\Cref{assump:Policy_est1,assump: ps_model} in \cref{subsec: regularity}) assume the boundedness of outcomes and covariates, moment and stability conditions for score functions, and the smoothness of the propensity score model.

\begin{theorem}[Asymptotic normality of the normalized estimator using the estimated propensity score]\label{thm:addIPW_est} \singlespacing
Suppose that \Cref{assump: unconfoundedness,assump: additive_model_multi,assump: factor_ps,assump:Policy_est1,assump: ps_model} hold. Then, we have,
$$
\sqrt{T-M+1}\Big(\hatVPS-V^{\ppi^{(M)}}\Big) \dto N( 0,\widetilde{\sigma}^2_{\lambda^\ast}),
$$
where $\widetilde{\sigma}^2_{\lambda^\ast} = \sigma^2_{\lambda^\ast}-\mat{U}^\top \mat{\Sigma}_{ps}^{-1}\mat{U}$, with $\sigma^2_{\lambda^\ast}$ defined in \Cref{lemma:addIPW_true}, and $\mat{\Sigma}_{ps}$ and $\mat{U}$ defined in \Cref{assump: ps_model}.
\end{theorem}

The following theorem establishes that the use of an estimated propensity score yields a more efficient estimator, which is analogous to well-known results in the standard i.i.d. setting \citep{hirano2003efficient} and related findings in settings with interference \citep{liu2016inverse,papadogeorgou2022causal}.

\begin{theorem}[Estimated vs. known propensity scores]\label{thm: Hajek efficiency} \singlespacing
If the propensity score model is correctly specified, the estimator $\hatVPS$ based on the estimated propensity score has asymptotic variance that is no greater than that of the estimator $\hatV$ based on the known propensity score.
\end{theorem}

The asymptotic variance of the proposed estimator cannot be consistently estimated because we only observe one time series per unit. This limitation holds regardless of whether we use the true or estimated propensity scores. Following \cite{zhou2024estimating}, we consider an upper bound on the variance.
\begin{proposition}[Variance upper bound and its consistent estimation] \label{prop: bound} \singlespacing
    The asymptotic variance of $\hatVPS$ is bounded by $$\mathrm{Var}^{\text{UB}}(\hatVPS) = \frac{1}{(T-M+1)^2}\sum_{t=M}^T\E\left[\rhotPS^2\left( \overline{Y}_t-\left(\lambda^\ast V^{\ppi^{(M)}}+(1-\lambda^\ast)\overline{Y}\right)\right)^2\mid \lHt[-M]\right],$$ which can be consistently estimated by
\begin{equation*}\label{eq: var_bound}
\varub(\hatVPS) = \frac{1}{(T-M+1)^2}\sum_{t=M}^T\tilde\rho_t^{(\bbeta)}(\ppi^{(M)};\hat\ttheta)^2\left(\overline{Y}_{t}-\left(\hat\lambda\hatVPS+(1-\hat\lambda)\overline{Y}\right)\right)^2,
\end{equation*}
where $\tilde\rho_t^{(\bbeta)}(\ppi^{(M)};\hat\ttheta) = \hat\lambda\dfrac{\rhotPS}{\rhobar}+(1-\hat\lambda)(\rhotPS+(1-\rhobar))$.
\end{proposition}
The proofs of all the theoretical results for this section appear in \Cref{a: Asymptotic_property}.

% In most cases outcome models with two-way interactions are
% sufficient for effective policy evaluation and learning. If the true model includes
% higher-order interactions, $\hatV[H,][\beta][\ppi]PS$ estimate the projection of the true policy value $V^{\ppi}$ into the linear treatment vector space, i.e.,
% \begin{equation}
%     \frac{1}{T}\sum_{t=1}^T\bm g^{\text{Proj.}}_{ti}\left(\lHt[-1]\right)^\top\phi^{(\beta)}(\ppi\left(\lHt[-1]\right)),
% \end{equation}
% where $g^{\text{Proj.}}_{ti}\left(\lHt[-1]\right) = \arginf_g \frac{1}{n}\sum_{i=1}^n\E\left[\left(\E[Y_{ti}\mid W_t,\lHt[-1]]-g_{ti}\left(\lHt[-1]\right)\right)^2\mid \lHt[-1]\right].$ 

\section{Policy learning}\label{sec:policy_learning}

In this section, we consider the problem of finding optimal policies over multiple time periods from time series data with arbitrary interference. We present our theoretical results in the case of estimated propensity scores while leaving the corresponding results for known propensity scores to \cref{a:subsec:proof-regret}.

\subsection{Finding an optimal policy} \label{subsec: find_optimal_policy}

We obtain an optimal policy by maximizing the estimated policy value function. Specifically, the estimated optimal policy $\hat\ppi^{(M)}$ is defined as a solution to the following multi-dimensional optimization problem,
\begin{equation}\label{eq: optimize}
\hat\ppi^{(M)}=(\hat\ppi_1,\dots,\hat\ppi_M) = \argmax_{\ppi^{(M)} \in \Pi^M} \hatVPS.
\end{equation}
Since the objective is nonsmooth and nonconvex, we approximate the deterministic policy class with a smooth stochastic policy class, which allows for gradient-based optimization.

For illustration, we focus on the linear policy class defined in \cref{eq: linearclass}, though optimization over a different policy class would proceed similarly. 
We approximate the deterministic linear policy class $\Pi$ by the following compact stochastic policy class,
\[
\Pi^{\text{stoch}}(k) = \left\{\ppi^{\text{stoch}}(k):
\pi_{i}^{\text{stoch}}\left(k; \lHt[-1]\right) = \expit(k\bm X_{ti}^P\ggamma),\,
||\ggamma||_\infty\le 1\right\},
\]
where $\bm X_{ti}^P$ is the covariate matrix in the linear policy class of \cref{eq: linearclass}. 
Under a stochastic policy $\ppi^{\text{stoch}}(k) \in \Pi^{\text{stoch}}(k)$, the quantity 
$\pi_{i}^{\text{stoch}}(k; \lHt[-1])$ represents the probability that unit $i$ is assigned treatment based on the observed history $\lHt[-1]$.

The value of a stochastic policy is defined analogously to the value of a deterministic policy except that it represents a weighted average of potential outcomes with weights equal to the probabilities induced by the policy (see \cref{a:subsec:stochastic_policy_value} for the exact definition).
Therefore, instead of solving \cref{eq: optimize} directly, % over the policy class $\Pi$, 
we optimize an estimated policy value function for stochastic policies over the stochastic policy class $\Pi^{\text{stoch}}(k)$.

\paragraph{The choice of tuning parameter $k$.}
For a fixed policy in $\ppi^{\text{stoch}}(k) \in \Pi^{\text{stoch}}(k)$, the tuning parameter $k$ controls how closely this policy approximates the corresponding deterministic rule. Specifically, as $k$ increases, $\ppi^{\text{stoch}}(k)$ will tend to assign treatment to units with $\bm X_{ti}^P\ggamma > 0$, and control otherwise. Meanwhile, the class $\Pi^{\text{stoch}}(k)$ expands as $k$ increases. As a result, for sufficiently large $k$, $\Pi^{\text{stoch}}(k)$ contains a policy close to the optimal deterministic policy. %, whereas for small $k$, it may not.

%In practice, $k$ should be set to be a sufficiently large constant to approximate the deterministic rule, but not excessively large to avoid vanishing gradients during optimization. 
% It can be chosen based on the scale of $\bm X_{ti}^P$. For example, 
If all covariates in $\bm X_{ti}^P$ are standardized, we have found that when $k=5$ many functions in the policy class often yield probabilities close to $0$ or $1$ over the observed range of $\bm X_{ti}^P$. If the learned $\ggamma$ has norm close to its constraint boundary and yet the resulting policy is far from being deterministic, this indicates that a larger $k$ may be needed.

\paragraph{Estimation of the optimal policy and theoretical guarantees.}
To estimate the optimal stochastic policy, we decompose the multi-period optimization problem into a series of single-period optimization problems over the stochastic policy class. Specifically, we first compute the first-lag optimal policy
$
\hat\ppi_1 = \argmax_{\ppi\in\Pi^{\text{stoch}}(k)} \hat V^{\ppi}.
$
Then, we compute the second-lag optimal policy while fixing the first-lag policy at $\hat\ppi_1,$ i.e.,
$
\hat\ppi_2 = \argmax_{\ppi\in\Pi^{\text{stoch}}(k)} \hat V^{\hat\ppi_1,\ppi}.
$
We continue in this manner until we obtain the $M$-th lag optimal policy $\hat\ppi_M$.
Approximation based on stochastic policies allows us to solve the sequential optimization problem using gradient-based methods. %In particular, we use the limited-memory Broyden-Fletcher-Goldfarb-Shanno (L-BFGS) algorithm, a widely used quasi-Newton optimization method \citep[e.g.][]{liu1989limited}. %that approximates the Hessian matrix using a small number of past gradient evaluations.  

The next theorem shows that our approach yields policies whose value is asymptotically no less than that of the optimal deterministic policy, up to an arbitrarily small error.
\begin{theorem}\label{thm: approx_optimize} \singlespacing
Suppose that \Cref{assump: unconfoundedness,assump: additive_model_multi,assump: factor_ps,assump:Policy_est1,assump: ps_model} hold. 
Let $\hat\ppi^{\text{stoch},(M)}(k)$ denote the policy obtained by optimizing over the stochastic policy class $\Pi^{\text{stoch}}(k)$. Then there exists $K\in\mathbb{N}$ such that for any $\epsilon>0$ and $k>K$,
\[
\P(V^{\hat\ppi^{\text{stoch},(M)}}(k)>V^{\hat\ppi^{(M)}}-\epsilon)\to 1,\quad \text{as  } T\to\infty.
\]
\end{theorem}
The proof is given in  \Cref{a:subsec:stochastic_proof}. While similar stochastic approximations are used in the presence of interference with cross-sectional data \citep{zhang2023individualized} and dynamic policy learning without interference among units \citep{fang2023fairness}, we establish the theoretical guarantees for dynamic policy learning under semiparametric additive interference.
The stochastic approximation framework and the result of \cref{thm: approx_optimize} extend, under analogous regularity conditions, to other commonly used policy classes, including fixed-depth decision trees and treatment sets with piecewise-linear boundaries \citep{kitagawa2018should,zhang2023individualized,zhou2023offline,viviano2025policy}.

%In such cases, a stochastic approximation can be obtained by replacing the hard decision with a smooth probability function, allowing the policy to assign treatment with probability close to 0 or 1.

In general, the learned stochastic policy is not guaranteed to converge to a deterministic policy as $T \to \infty$.  This is because $\Pi^{\text{stoch}}(k)$ always contains policies that differ from deterministic ones, irrespective of the choice of $k$. Such stochastic policies might lead to higher policy value than deterministic ones. However, \cref{thm: approx_optimize} shows that the learned stochastic policy will be at least as good as the optimal deterministic one. Even though the learned stochastic policy could differ from a deterministic rule, in practice, we have found that the learned stochastic policy often produces treatment assignment probabilities close to $0$ or $1$  in all of our simulation and empirical studies (\cref{sec:application}). %with probabilities either smaller than $0.01$ or larger than $0.99$. 
In the case where a deterministic policy is required, one can add a penalty term to the optimization objective function (see \cref{a:subsec:deterministic_penalty}).

%There are cases where the optimization problem can be solved directly without invoking the stochastic policy class approximation. 
For certain policy classes, if the mixing parameter $\lambda$ is fixed to a specific value and the outcome model contains no interaction term (i.e., $\beta=1$), the problem in \cref{eq: optimize} admits a mixed integer programming (MIP) formulation based on auxiliary variables \citep{kitagawa2018should,zhang2023individualized}. 
However, when $\lambda$ is estimated or the outcome model has interaction terms, the resulting objective cannot in general be formulated as an MIP, and a stochastic policy class approximation needs to be employed.

\paragraph{Addressing propensity overfitting in policy learning with weighting estimators.}
The policy learning step in \cref{eq: optimize} requires the estimation of the policy value under different policies.
In principle, different estimators of the policy value can be used. When such estimator employs inverse propensity weights, estimation of an optimal policy can suffer from propensity overfitting, a phenomenon in which the algorithm increases the objective function by selecting policies that align with the propensity score rather than improving the true policy value \citep{swaminathan2015self}.

To address this issue, the proposed estimators for the policy value in $\mathcal{C}(\ppi^{(M)})$ introduced in \cref{sec:estimation} employ standardized weights, making themselves shift-invariant so that adding a constant to the outcome variable shifts the policy value estimator by the same constant. Since artificial outcome shifts cannot be exploited to inflate the objective of the optimization problem without improving the true policy value, shift-invariant estimators are not susceptible to propensity overfitting.  We do not consider the broader class of estimators
$\left\{\frac{1}{T}\sum_{t=1}^T\rhot\overline{Y}_t+\Lambda(1-\rhobar): \Lambda\in\mathbb{R}\right\}$,
which includes $\mathcal{C}(\ppi^{(M)})$ as a special case \citep{khan2023adaptive}, because the estimators in this class are generally not shift-invariant.

%%%%%%%%%%%%%%%%%%%%%%%%%%%%%

% We estimate the optimal policy by maximizing the estimated policy value function,
% \begin{equation}\label{eq: optimize}
% \hat\ppi^{(M)}=(\hat\ppi_1,\dots,\hat\ppi_M) = \argmax_{\ppi^{(M)} \in \Pi^M} \hatVPS.
% \end{equation}

% We first present the details for solving the optimization problem in \cref{eq: optimize} for policy $\hat\ppi$ over one time period. This optimization is challenging because the objective is nonconvex. Moreover, the object involves the policy indicator variable $\mathds{1}\left\{W_{ti}=\pi_{i}\left(\lHt\right)\right\}$, which is non-smooth and non-differentiable, making gradient-based optimization infeasible.

% We propose an approach that approximates the deterministic policy by a stochastic policy. For illustration, we present the details for a linear policy class defined in \cref{eq: linearclass}.

% Let $k \ge 1$. We approximate the linear policy class $\Pi$ by a compact stochastic policy class $\Pi_k^{\text{stoch}} = \left\{\ppi^{\text{stoch}}(k):\pi_{ki}^{\text{stoch}}\left(\lHt[-1]\right) = \expit(kX_{ti}^P\gamma), ||\ggamma||_\infty\le 1\right\}$, where each 

\subsection{Regret analysis}\label{subsec: regret}
To assess the performance of a learned policy relative to the optimal one, we measure its regret, defined in \cref{subsec: policyvalue}.
The following theorem establishes a regret bound for the learned policy with
estimated propensity scores.
The result is derived under a regularity condition on the coordinate policy class
$
\mathcal G=\big\{\overline{H}_{t-1}\to \pi_{i}(\overline{H}_{t-1}):\ \ppi\in\Pi,\ i=1,\ldots,n\big\}.
$ has bounded sequential Rademacher complexity \citep{rakhlin2015sequential}, as formalized in
\Cref{assump:coord-capacity}. This regularity condition limits how flexibly
unit-level treatment decisions can depend on past histories. For the deterministic linear policy class, this condition holds if the variables in the history have bounded support, together
with a separation condition that prevents the linear combination used to assign treatment from being
arbitrarily close to the threshold defining the treatment decision. We provide further discussion in
\Cref{a:subsec:regret-regularity}.

\begin{theorem}[Regret bound with estimated propensity scores]\label{thm: regret_estimate_ps} \singlespacing
Suppose that 
\Cref{assump: unconfoundedness,assump: additive_model_multi,assump: factor_ps,assump:Policy_est1,assump:coord-capacity,assump: ps_model} hold.
% Let $\hat\ttheta$ be the maximum likelihood estimator whose score function $\psi\left(W_t, \lHt[-1]; \ttheta\right)$ satisfies Assumption~\ref{assump: ps_model}, and and 
Let the learned policy under the estimated propensity score be $\hat\ppi^{(M)}_{\hat\ttheta} = \argmax_{\ppi^{(M)}\in\Pi^M}\hatVPS$, 
then $R(\hat\ppi_{\hat\ttheta}^{(M)})=O_p\left(\sqrt{1/T}\right)$.
\end{theorem}

The proof is in \cref{a: regret-bound-proof}. \Cref{thm: regret_estimate_ps} shows that when the propensity scores are estimated, the regret of the learned policy decreases at the rate $\sqrt{1/T}$ in probability. A smaller regret bound means that the performance of the learned policy is closer to that of the true optimal policy. Notably, this rate matches the benchmark in policy learning with independent and identically distributed observations \citep{kitagawa2018should,athey2021policy}. Despite the presence of temporal dependence and interference among units in our setting, we do not incur any loss in the regret rate relative to the i.i.d.\ case.

When the true propensity scores are known, we establish a finite-sample regret bound that is also of order $\sqrt{1/T}$. This bound depends on the complexity of the policy class, the magnitude of the outcome, and the order of interaction $\bbeta$ (see \Cref{thm:regret} in \cref{a: regret-bound-proof}).

\section{Selecting the order of interactions}\label{sec:interaction}

The estimator $\hatVPS$ depends on $\bbeta = (\beta_1,\dots,\beta_M)$, where $\beta_m$ denotes the highest order of interaction in the $m$-th model in \Cref{assump: additive_model_multi}. We introduce a statistical test for choosing the order of interaction. % Each $\beta_m$ can in principle take any value up to $n$. 
Lower values for each $\beta_m$ impose stronger restrictions on the potential outcomes, and are thus expected to lead to estimators with smaller variance. However, such low-rank models run the risk of model misspecification if the true model includes interactions of higher order. % At the same time, higher-order models can capture complex dependencies at the cost of higher variability.
We focus on a practical setting with $\bbeta \in \{1,2\}^M$.

\subsection{Testing Procedure}

Our test proceeds sequentially by performing one hypothesis test for each of the $M$ models in \Cref{assump: additive_model_multi}. At the $m$-th step, given the selected order of interaction for the first $m - 1$ models, $\beta_1,\dots,\beta_{m-1}$ and the optimal policy for the first $m - 1$ time periods, $\hat\ppi^{(m-1)}_{\hat\ttheta} = (\hat\ppi_{1,\hat\ttheta},\dots,\hat\ppi_{m-1,\hat\ttheta})$, we test whether the lower-order $m$-th model with $\beta_m = 1$ suffices for the potential outcomes relative to the higher-order model with $\beta_m=2$. Let $\bbeta_{mk} = (\beta_1,\dots,\beta_{m-1},k)$ for $k=1,2$ denote the two candidate interaction-order sequences at step $m$. The proposed hypothesis test compares these two candidates. 

For a policy $\ppi_m \in \Pi$ in the policy class, let $\projVt[\bbeta_{mk}][\hat\ppi^{(m-1)}_{\hat\ttheta},\ppi_m][-m]$
be the projected policy value of $(\hat\ppi^{(m-1)}_{\hat\ttheta},\ppi_m)$ based on the outcome model with the order of interaction specified in $\bbeta_{mk}$. For a correctly-specified model, the projected policy value reduces to the true policy value $V_t^{\hat\ppi^{(m-1)}_{\hat\ttheta},\ppi_m}$ (see \Cref{a:projection} for the definitions of the projected and true policy values and details on the projection).
We use
\[
\Delta^\ast_{mt}(\ppi_m)
=
\projVt[\bbeta_{m1}][\hat\ppi^{(m-1)}_{\hat\ttheta},\ppi_m][-m]
-
\projVt[\bbeta_{m2}][\hat\ppi^{(m-1)}_{\hat\ttheta},\ppi_m][-m]
\]
to denote the difference between the projected policy values of
$(\hat\ppi^{(m-1)}_{\hat\ttheta},\ppi_m)$ at time $t$ under the two model specifications indexed by $\bbeta_{m2}$ and $\bbeta_{m1}$.

Formally, our null hypothesis is given by,
\[
H_{m,0}: \Delta_{mt}^\ast(\ppi_m)=0
\quad \forall\, (t,\ppi_m)\in\{m,\ldots,T\}\times\Pi.
\]
If higher order interactions are not present, both models indexed by $\bbeta_{m1}$ and $\bbeta_{m2}$ are correctly specified and their corresponding projected policy values 
equal the true policy value at time $t$, leading to $\Delta^\ast_{mt}(\ppi_m) = 0$. In this case, the null hypothesis holds, and the lower order $\bbeta_{m1}$ can be used for estimation.
In contrast, support for the alternative hypothesis implies that the lower-interaction model does not adequately capture interactions in the true potential outcome model.

\subsection{Theoretical Properties}
We use the asymptotic distribution of the estimators derived in \Cref{thm:addIPW_est} to perform this hypothesis test.
The null hypothesis requires the equality to hold for every policy in $\Pi$, but evaluating it over the entire policy class is generally infeasible. We therefore randomly select $L$ candidate policies $\ppi_{m1},\dots,\ppi_{mL} \in \Pi$ to cover a diverse set of decision rules. For the linear policy class in \cref{eq: linearclass}, these candidate policies can be selected by drawing random coefficient vectors $\ggamma$.
% First, we randomly select $L$ candidate policies $\ppi_{m1},\dots,\ppi_{mL} \in \Pi$. 
For each single-period candidate policy $\ppi_{ml}$, we define the compounded $m-$time period policy
\(
\ppi_l^{(m)}
=
(\hat\ppi_{1,\hat\ttheta},\dots,\hat\ppi_{m-1,\hat\ttheta},\ppi_{ml}),
\) for $l=1,\dots,L$,
which % denotes the $m$-period policy formed by 
combines the previously selected policies for periods $1,\dots,m-1$ with the $l$-th candidate policy for period $m$. Then, we estimate the policy value for each one of them using the mixing estimator in \cref{eq:hatV_lambdahat}, yielding $\hatVPS[][\bbeta_{mk}][\ppi_l^{(m)}]$. An empirical analogue of $\Delta_{mt}^\ast(\ppi_l^{(m)})$ is $\hat\Delta_{mt}(\ppi_l^{(m)}) = \hat V_t^{(\bbeta_{m1})}(\ppi_l;\hat\ttheta)-\hat V_t^{(\bbeta_{m2})}(\ppi_l;\hat\ttheta),$ where $\hat V_t^{(\bbeta_{mk})}(\ppi_l^{(m)};\hat\ttheta)$ obtained from  \cref{a:eq:hatVt}, is the time-$t$ estimator averaged over in $\hatVPS[][\bbeta_{mk}][\ppi_l^{(m)}]$. 
We collect these estimated differences across the $L$ candidate policies in  
\(
\hat{\bm\Delta}_{mt} = \left(\hat\Delta_{mt}(\ppi_1^{(m)}),\dots,\hat\Delta_{mt}(\ppi_L^{(m)})\right)^\top
\).

Under the null hypothesis, 
% We aggregate these $\hat{\bm\Delta}_{mt}$ over time by defining
\( \displaystyle 
\hat{\bm\Delta}_m
=
\frac{1}{\sqrt{T-m+1}}
\sum_{t=m}^T
\hat{\bm\Delta}_{mt}.
\)
% Under $H_{m,0}$, $\hat{\bm\Delta}_m$ 
is asymptotically normally distributed with mean zero. Therefore, we form the test statistic based on
\[
T_m
=
\left\lVert\hat{\bm\Delta}_m\right\rVert_2^2
=
\frac{1}{T-m+1}
\left\lVert
\sum_{t=m}^T
\hat{\bm\Delta}_{mt}
\right\rVert_2^2,
\]
where a large value indicates a systematic difference between the projected policy values under different model specification.

We approximate the null distribution of the test statistic using a multiplier bootstrap. 
Although one could estimate the asymptotic covariance matrix of $\hat{\bm\Delta}_m$ and its eigenvalues, the resulting estimates tend to be unstable in finite samples.
For each $b=1,\dots,B$, let $\{Z_t^{(b)}\}_{t=m}^T$ be independent standard normal random variables. Define
\[
T_m^{(b)}
=
\frac{1}{T-m+1}
\left\lVert
\sum_{t=m}^T
Z_t^{(b)}\hat{\bm\Delta}_{mt}
\right\rVert_2^2.
\]
We compute the $p$-value as
\[
\frac{1}{B}
\sum_{b=1}^B
\mathds{1}\left\{
T_m^{(b)}\geq T_m
\right\}.
\]
We conduct the sequential testing procedure at an overall significance level of $\alpha$. We reject $H_{m,0}$ and select $\beta_m = 2$ if the $p$-value is smaller than $\alpha/M$. Otherwise, we set $\beta_m=1$. We then compute the corresponding policy $\hat\ppi_{m,\hat\ttheta}$ by maximizing $\hatVPS[][\beta_1,\dots,\beta_m][\ppi]$ over $\Pi$.

The following theorem establishes that, under the null hypothesis, the proposed test asymptotically controls the Type I error rate as the number of time periods increases.

\begin{theorem}[Asymptotic validity of the test for multiple time periods]\label{thm:test_multi} \spacingset{1}
Let $1\le m\le M$ and $p\text{-value}_m$ denote the $p$-value derived from the conditional distribution of the bootstrap statistic given the data. Then, under \Cref{assump: unconfoundedness,assump: additive_model_multi,assump: factor_ps,assump:Policy_est1,assump: test_regularity,assump: ps_model}, and under the null hypothesis $H_{m,0}$,
\[
\limsup_{T\to\infty} \P\left( p\text{-value}_m < \frac{\alpha}{M} \right) \leq \frac{\alpha}{M}.
\]
Consequently, the sequential procedure controls the overall type I error at level $\alpha$.
% via the Bonferroni correction.
\end{theorem}

The proof is provided in \cref{a:sec:test_multi}.
In the simulations and application, we use $L=1000$ candidate policies and $B=1000$ bootstrap realizations. The value of $L$ should be large enough so that the sampled candidate policies provide adequate coverage of the policy class $\Pi$. Once $L$ is reasonably large, we have found that the performance of the procedure is not sensitive to its exact value (see \cref{a:subsec:L_sensitivity}).
Although our procedure is developed for policy learning, the same procedure can be used to choose the order of interactions in the context of policy evaluation with the test defined based on the single policy under consideration ($L = 1$).

%, without choosing $L$ policies as in policy learning, using the same logic to select the vector $\bbeta$. 
In many applications, models with interactions up to second order are typically considered sufficient, as higher-order interactions are expected to be weak or nonexistent. This is supported by our simulation results in \Cref{sec:simulation}, where estimators based on low-order interaction models perform well even when moderate higher-order interactions are present. 

Nevertheless, the test introduced above can be extended to allow for general values of $\bbeta$ and higher-order models.
For example, in some settings, we might want to allow for complex interactions among treatments of recent time periods (larger values of $\beta_m$ for small $m$), while restricting the interactions of older treatments for their effect on current outcomes (smaller values of $\beta_m$ for large $m$).

Lastly, simultaneous testing of all $\beta$ values is not straightforward under our sequential estimation procedure. We can employ the Bonferroni correction, which controls the overall Type I error rate but may be too conservative in practice. We leave the development of less conservative testing approaches for the order of interaction to future work.

\section{Simulation studies}
\label{sec:simulation}

We conduct simulation studies to evaluate the performance of the proposed methods. Due to space constraints, we focus on simulations for policy learning here, while we summarize conclusions for policy evaluation below and in detail in \cref{a:subsec:additional_policy_eval}.

\subsection{Study design}

We consider time series of length $T \in \{50, 100, 300, 500, 1000\}$ with $n \in \{10,50,100\}$ units. At each time period $t$ and for each unit $i$, we independently generate two covariates $\bm X_{ti} = (X_{1ti}, X_{2ti})$ from the standard normal distribution. The treatment $W_{ti}$ is generated from a Bernoulli distribution with success probability depending on the observed history:
\begin{align*}
\P(W_{ti} \mid \lHt)
= \expit\Big(
    0.8 
    - 0.3 X_{1ti} 
    + 0.5 X_{2ti} 
    - 0.8 W_{t-1,i} 
    + \frac{1}{n-1} \sum_{j\neq i} W_{t-1,j}
    - 0.2 Y_{t-1,i}
    \Big),
\end{align*}
where $\expit(x) = 1/(1+e^{-x})$ is the logistic function. The outcomes are generated as $Y_{ti}\mid \lHt[-1] \sim N(\mu_{ti}, 1)$, where
\begin{align*}
    & \mu_{ti} = (-0.5 X_{1ti} + X_{2ti}) W_{ti}
    \,+\, \sum_{j\neq i} X_{1tj} W_{tj}
    \,-\, \frac{0.3}{n-1} \sum_{j\neq i} W_{t-1,j}
    \,+\, X_{2ti} \\[2pt]
    &\hspace{1.9cm}
    +\, \frac{c}{n-1}\, X_{2ti} \left(\sum_{j\neq i}W_{tj}\right) W_{ti}
    \,+\, 0.2 Y_{t-1,i}.
\end{align*}
The parameter $c \in \{0, -2, -10\}$ represents three scenarios with different strengths of interactions between a unit's treatment and the treatments of other units in the outcome model. When $c = 0$, the linear additive model with $\beta_m = 1$ for all $1\le m\le M$ satisfies \Cref{assump: additive_model_multi}. The cases $c = -2$ and $c = -10$ represent moderate and strong two-way interactions in $\bm W_{t}$, respectively. 
Since the outcome at time $t-1$ affects the outcome at time $t$, previous treatments also form interactions in the outcome model when $c \neq 0$. To quantify the interpretable magnitude of the interaction term involving $c$, we compared its absolute value with the absolute value of the remaining terms in $\mu_{ti}$. Averaged across the simulated settings, this ratio was approximately 0.13 for $c=-2$, and 0.69 for $c=-10$. For each setting, we generate 500 datasets.

Next, we assess policy learning performance considering the following policy class, which we use to estimate the optimal policy over one and two time periods ($M \in \{ 1, 2 \}$):
\[
\Pi = \left\{ \ppi : \ppi_i\left(\lHt[-1]\right) = \mathds{1}\left(\gamma_0 + \gamma_1 X_{1ti} + \gamma_2 X_{2ti} > 0\right),\  i=1,\dots,n \right\}.
\]
In each setting, we employ estimators under linear and quadratic additive models over one or two time periods, as well as the test-based estimator. This yields three estimators in the $M = 1$ setting ($\hat V^{(1)}(\ppi), \hat V^{(2)}(\ppi)$ and test-based), and five estimators in the $M = 2$ setting ($\hat V^{(1,1)}(\ppi)$, $\hat V^{(2,2)}(\ppi)$, $\hat V^{(1,2)}(\ppi)$, $\hat V^{(2,1)}(\ppi)$, and test-based).
%For single-period policies, we compare three estimators for policy value: $\hat V^{(1)}(\ppi)$ and $\hat V^{(2)}(\ppi)$, which assume linear and quadratic outcome models respectively, and the test-based estimator described in \cref{sec:interaction}. For two-period policies, we allow either linear or quadratic outcome models at each stage, leading to five estimators: $\hat V^{(1,1)}(\ppi)$, $\hat V^{(2,2)}(\ppi)$, $\hat V^{(1,2)}(\ppi)$, $\hat V^{(2,1)}(\ppi)$, and the test-based estimator. The bootstrap test is conducted with $\alpha = 0.1$, $L = 1000$ candidate policies, and $B=1000$ bootstrap realizations.
% For both policy evaluation and learning, 
We estimate the propensity scores under correct model specification.

\subsection{Results} 

\begin{figure}[!t]
    \centering
    \includegraphics[width=\linewidth]{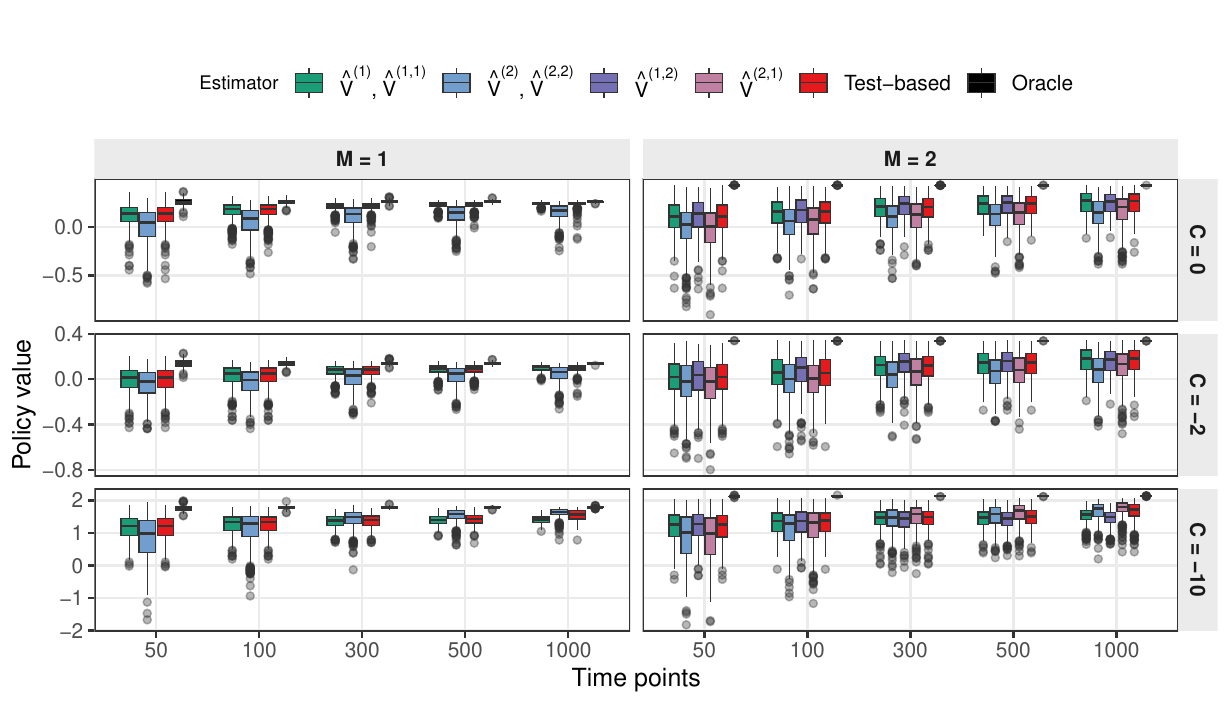}
    \caption{Boxplots of the policy value $V^{\hat\ppi_{\hat\ttheta}}$ under the learned policies, based on different estimators, for $n=50$ across different model specifications (rows) and policies of different lengths (columns). We consider estimators assuming a linear additive outcome model ($\hatVPS[][\bm 1]$, green), a quadratic additive outcome model ($\hatVPS[][\bm 2]$, blue) for all time periods, and a test-based estimator with $\alpha=0.05$ (red). For policies over two time periods ($M=2$), we additionally include two estimators, $\hatVPS[][1,2]$ (purple) and $\hatVPS[][2,1]$ (pink), which assume a linear additive outcome model for one time period and a quadratic additive outcome model for the other.}
    \label{fig:sim_res_M1M2_mix_estPS}
\end{figure}

\cref{fig:sim_res_M1M2_mix_estPS} shows the performance of the proposed policy learning methodology for $n=50$, where the first column reports the results for learning the optimal single-period policy, and the second column corresponds to learning the optimal two-period policy.  The rows correspond to the three scenarios for the strength of interactions (no, moderate, and strong interactions, respectively).

For single-period policies and in the scenario without interactions ($c=0$), the linear additive outcome model is correctly specified. We define the oracle policy as the policy that maximizes the true policy value over the specified policy class. As expected, $\hat V^{(1)}(\ppi)$ leads to policies with values closer to the oracle, with policy value converging to the oracle policy value as $T$ increases. Although $\hat V^{(2)}(\ppi)$ remains unbiased, its larger variance leads to substantially worse performance and lower policy value for the estimated optimal policy in finite samples. The test-based estimator achieves performance comparable to $\hat V^{(1)}(\ppi)$. 

A similar pattern is observed in the presence of moderate interactions in the potential outcome model ($c=-2$) and the estimator based on the linear additive model outperforms the estimator based on the correctly specified model that allows for interactions. However, when interactions are strong ($c=-10$), the estimator based on the model that allows for second-order interactions, $\hat V^{(2)}(\ppi)$, outperforms the estimator based on the linear additive model, $\hat V^{(1)}(\ppi)$, particularly as $T$ grows. The test-based estimator moves towards the best-performing estimator as $T$ increases. 

For policies over two time periods, the findings are broadly consistent with those from the single-period case. When the additive outcome model is correct ($c=0$) or involves only moderate second-order interactions ($c=-2$), estimators based on linear additive models tend to perform better than those with interactions, which can suffer from higher variance. However, when the interaction effects are strong ($c=-10$), the estimator based on the outcome model that allows for second-order interactions among the most recent treatments, $\hat V^{(2,1)}(\ppi)$,
performs best when $T$ is large, followed closely by the estimator based on the model that allows for interactions among treatments at the most recent two time periods, $\hat V^{(2,2)}(\ppi)$.
% which allows for second-order interactions among treatments at both time periods, does not necessarily outperform $\hat V^{(2,1)}(\ppi)$, which assumes a quadratic outcome model at the first time period and a linear outcome model at the second. 
% This is likely due to the high variability of $\hat V^{(2,2)}(\ppi)$. 
The test-based estimator again improves as $T$ increases and approaches the performance of the better-performing model choice. This convergence is slower in the strong-interaction setting, possibly due to the low test power of detecting the second-order interaction model when $T$ is relatively small.

\cref{a:sec:additional_sim} reports the policy evaluation results, together with the policy learning results for $n=10$ and $n=100$. Overall, the policy evaluation results demonstrate a bias-variance tradeoff in the choice of the outcome model. Higher-order models reduce bias under strong interaction effects but produce more variable estimates. Across the considered scenarios, the proposed mixing estimator generally has smaller variability and RMSE than the shifted estimator, particularly under quadratic additive outcome models and for small $T$. The mixing estimator also exhibits more stable finite-sample behavior than the H\'ajek-type estimator, which occasionally produces extreme estimates. 

In policy learning, the test-based estimator approaches the best-performing estimator more quickly in the lower-variance case with $n=10$. For $n=100$, the same qualitative pattern remains, but the improvement is slower in the strong-interaction setting because the test has lower power. In addition, the policy learning approach based on additive IPW weights consistently outperforms the standard H\'ajek estimator, as shown in \cref{fig:sim_res_M1M2_Hajek_compare_estPS}.

\section{Application}
\label{sec:application}
In this section, we apply the proposed methodology to real-world data from Iraq and study the optimal allocation of aid projects for reducing insurgent attacks. We find that, compared with the observed allocation, the learned policies reduce the probability of starting new projects in some districts that received projects frequently and increase it in districts that were treated less often.

\subsection{Setup}
We consider the optimal district-level aid assignment for minimizing insurgent attacks. Aid data are only available at the level of the district, and allocation is sparse over time for many districts. We therefore study district-level aid allocation at the weekly level. This also provides a reasonable time scale for our study, where the assignment of aid might have similar effects in practice regardless of the exact day of assignment within the short time window of a week. 

Our data cover 71 weeks between February 23, 2007, and July 5, 2008. The aid data contain records from several reconstruction, stabilization, and security programs, including CERP, ESF, ISFF, and IRRF. During the study period, the file contains 21,910 project records with positive recorded construction costs, totaling approximately \$4.97 billion. CERP and ESF account for about 95\% of these records. We define $W_{ti}=1$ if at least one new aid project began in district $i$ during week $t$, and $W_{ti}=0$ otherwise. We analyze the two most frequent types of insurgent attacks as the outcomes of interest. We use $Y_{ti}$ to represent the number of improvised explosive device (IED) or small arms fire (SAF) attacks in district $i$ during week $t$.

% Our data cover 71 weeks between February 23, 2007, and July 5, 2008. We define $W_{ti} = 1$ if district $i$ receives new aid in week $t$, and $W_{ti}=0$ otherwise. We analyze the two most frequent types of insurgent attacks as the outcomes of interest. 
% We use $Y_{ti}$ to represent the number of improvised explosive device (IED) or small arms fire (SAF) attacks in district $i$ during week $t$.

\cref{fig:aid_and_outcome} shows the observed aid assignment and outcomes for each district averaged over all time periods. Both IED and SAF attacks were concentrated near the major cities of Baghdad, Mosul, and Basra. In terms of aid, districts experiencing a large number of insurgent attacks were more likely to receive a new aid project, with some districts near Baghdad receiving aid in nearly every week of the study period. Five districts received no aid and are therefore excluded from the analysis.  

\begin{figure}[!b]
\centering
\subfloat[Aid]{
\includegraphics[width = 0.34\textwidth,trim = 0 0 0 0, clip]{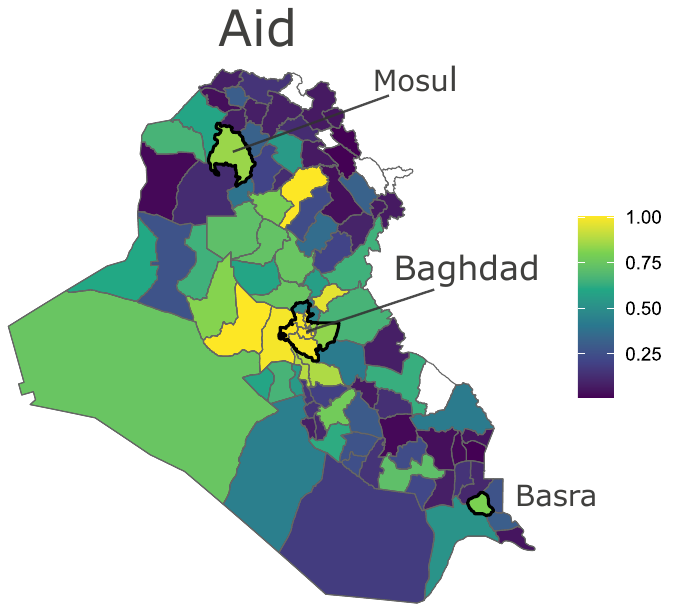}
\label{fig:aid}
}
\subfloat[Insurgent violence]{
\includegraphics[width = 0.6\textwidth,trim = 0 0 0 0, clip]{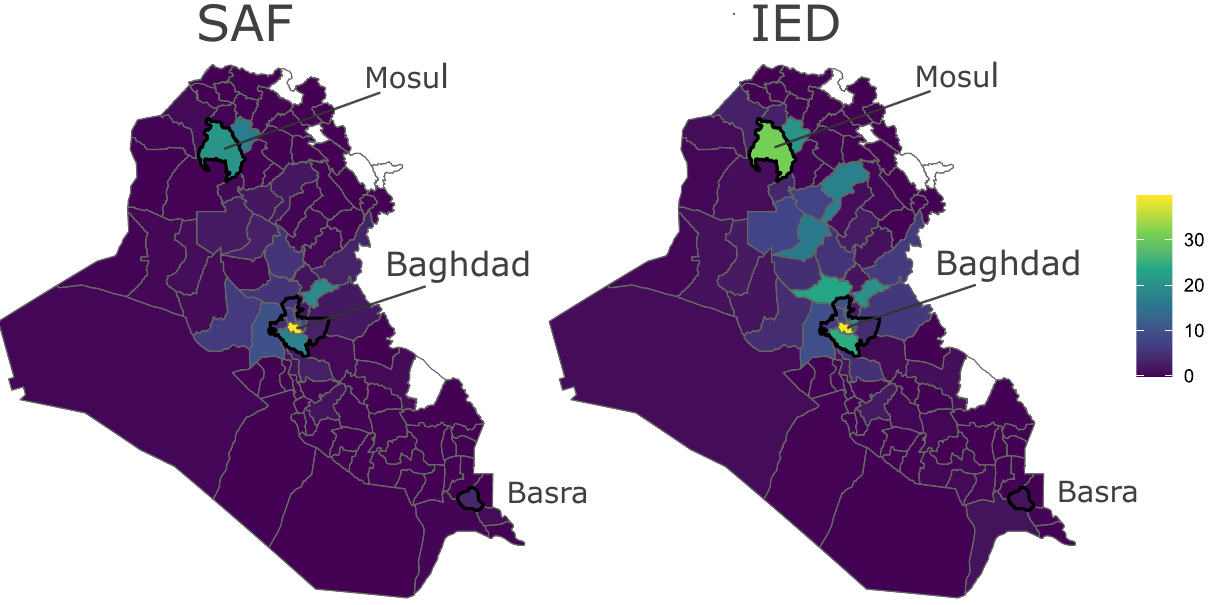}
\label{fig:outcome}
} 
\caption{Distribution of aid projects and insurgent attacks. Panels (a) and (b) show the district-level averages of aid and insurgent attacks (IED and SAF attacks), respectively, from February 23, 2007, to July 5, 2008, averaged across all time periods.}
\label{fig:aid_and_outcome}
\end{figure}

We study the optimal aid assignment over horizons of $M=1,2,3,4$ weeks. We consider a linear policy class based on whether aid was dispensed, the number of airstrikes, and the number of outcome events in the preceding time period. We also include U.S. troop density and an indicator for missing troop density in the policy covariate set. The troop density is measured as soldiers per 1,000 local residents for each district in the preceding week. The missingness indicator is included because missing troop records are spatially clustered, with missingness concentrated in northern and southern districts rather than scattered uniformly.
 
To minimize extrapolation to treatment assignments not observed frequently, we specify that the optimal policy always assigns treatment to districts that receive aid in at least 95\% of weeks during the study period, and control to districts that received aid in fewer than 5\% of weeks.
The optimal policy aims to minimize the number of outcome events across all $M$ periods (see \cref{a:sec:alternative_policy_value} for the definition of the policy value when using a summary outcome over multiple time periods).  

For the propensity score, we consider a logistic regression model of the form
$\P(W_{ti} \mid \bm X_{ti}) = \expit\left(\bm X_{ti}\ttheta\right)$,
where the treatment assignment is assumed to be conditionally independent across $i$ in each time period, in accordance with \Cref{assump: factor_ps}.
Here, $\bm X_{ti}$ includes spatial, temporal, and spatio-temporal information. For spatio-temporal variables, we include aid, airstrikes, shows of force, IED attacks, and SAF attacks over the previous one, two, and four weeks for each district, as well as district-specific U.S. troop density in the preceding week. Prior aid is included as whether aid was disbursed in the district and the amount given, during each time window. These 19 variables are designed to capture immediate conditions and recent accumulated activity. 

The adjustment set also includes variables that vary across districts but are fixed over time. These are distance to major cities, distance to road networks, distance to rivers, the log of population density measured in 2003 and an indicator variable for whether troop density is missing, giving five spatial covariates.  Finally, we include variables that vary over time but not across districts. These are three spline terms for time and an indicator for the surge period (March 2007--January 2008), giving four temporal covariates. In total, the propensity score model includes 28 covariates.

We evaluate the fitted propensity score model by examining the resulting covariate balance. We find that the 
absolute standardized differences of covariates are generally at or below the 0.1 threshold, indicating satisfactory balance across treatment groups (see \Cref{a: subsec:balance}).

\begin{table}[!t]
\centering
\begin{tabular}[t]{crrrrrrrr}
\toprule
\multicolumn{1}{c}{ } & \multicolumn{4}{c}{IED} & \multicolumn{4}{c}{SAF} \\
\cmidrule(l{3pt}r{3pt}){2-5} \cmidrule(l{3pt}r{3pt}){6-9}
$M$ & $\hat V$ & Lower & Upper & Observed & $\hat V$ & Lower & Upper & Observed\\
\midrule
1 & 1.86 & 1.63 & 2.08 & 2.215 & 1.77 & 1.30 & 2.24 & 2.230\\
2 & 3.42 & 3.07 & 3.77 & 4.430 & 3.20 & 2.58 & 3.83 & 4.460\\
3 & 4.84 & 4.40 & 5.29 & 6.645 & 4.54 & 3.65 & 5.43 & 6.689\\
4 & 6.19 & 5.68 & 6.70 & 8.860 & 4.92 & 1.12 & 8.72 & 8.919\\
\bottomrule
\end{tabular}
\caption{Estimated policy values ($\hat V$) of the learned optimal policies 
and corresponding 95\% confidence intervals for the IED and SAF analyses 
for $M=1,2,3,4$. Lower and Upper denote the lower and upper confidence 
limits, respectively, and Observed denotes the observed average number of 
attacks across all time periods and spatial units.}
\label{tab: policy_res}
\end{table}

\subsection{Results}\label{subsec:app_res}
The results of the learned aid allocation policy are reported in \cref{tab: policy_res}. The estimated policy value, $\hat V$, represents the expected number of outcome events over $M$ weeks if aid is allocated according to the learned policy, averaged across all districts and weeks.
We find that, for policies lasting 2 weeks or longer, the estimated policy values and their confidence intervals are below the corresponding observed averages for both IED and SAF attacks. The estimated reductions become larger as the policy horizon increases.

\begin{figure}[!t]
\centering
\includegraphics[width = \textwidth,trim = 0 0 0 0, clip]{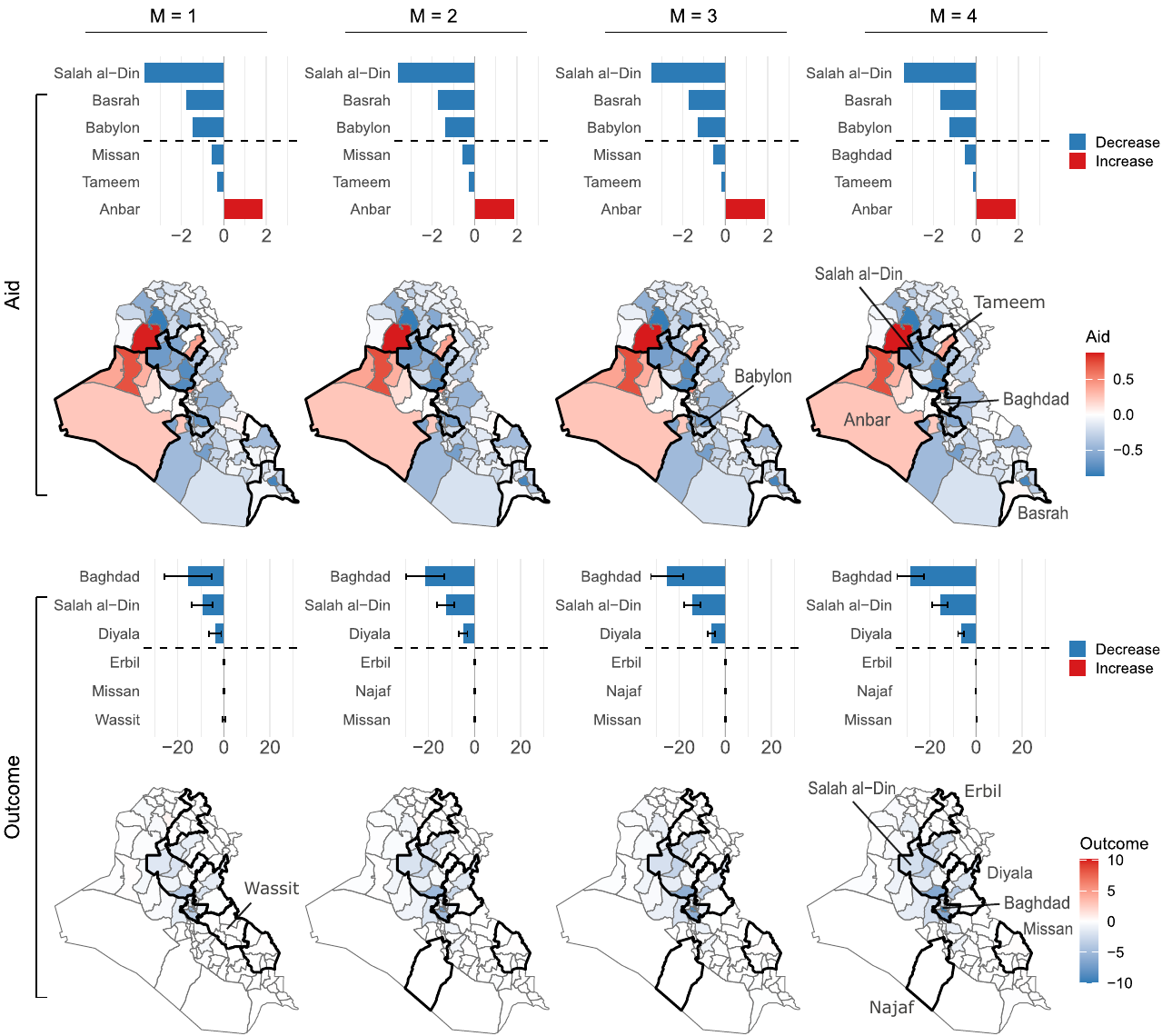}
\caption{Results for the analysis of IED attacks. The four columns correspond to policies learned over horizons of $M=1,2,3,4$ weeks. The first two rows display the difference in average aid assignment between the learned policy and the observed allocation, while the last two rows show the difference in the expected number of IED attacks relative to observed outcomes. The bar plots highlight the three governorates with the largest decreases and the three with the smallest (or most negative) decreases, aggregated at the governorate level. For the outcome differences, the 95\% confidence intervals are also shown.} 
\label{fig:IED_results}
\end{figure}

Here, we focus on results for IED attacks.
\Cref{fig:IED_results} displays the differences in aid assignment and expected number of attacks between the observed and learned policies, averaged across all policy periods. The spatial patterns are similar across different values of $M$. 
%, with the policy over four weeks producing the largest reduction in expected outcomes.
The learned policies often increase the probability of starting new projects in Anbar and other sparsely populated western governorates, while lowering this probability in Baghdad and nearby governorates, where new projects began especially frequently under the observed allocation. Although the learned policies lower the probability of starting a new project in Baghdad districts, much of the expected reduction in IED attacks occurs there. For $M=4$, Baghdad has the largest estimated reduction, followed by Salah al-Din and Diyala. This suggests that aid allocation can have far reaching spatial effects, with optimal assignments not necessarily occurring in the areas with the highest number of IED attacks. Results for SAF attacks in \Cref{a:subsec:SAF} return similar conclusions.

The learned policy coefficients for the IED analysis are summarized in \cref{tab:policy_coefficients_ied}. % The four rows describe the decision rules applied in the first, second, third, and fourth weeks of the policy horizon. The coefficients suggest that 
We find that the learned allocation rule is broadly similar across time horizons. In general, optimal aid assignment concentrates in districts with fewer IED attacks and airstrikes that did not receive a new aid project in the previous week, and higher U.S. troop density.

In \Cref{a:subsec:morecovariates}, we conduct sensitivity analyses that use an expanded list of covariates for the policy class. The results are broadly similar to the main analysis.
%with estimated policy values for learned policies remaining below the observed outcomes for both IED and SAF attacks. The expanded policy class gives slightly lower estimated values for the IED analysis, while the SAF results are very close to those from the simpler policy class and show a similar spatial pattern.

\section{Conclusion}
\label{sec:conclusion}
In this paper, we propose a methodology for selecting optimal policies over multiple time periods in a spatio-temporal setting with spillover and carryover effects. We introduce semiparametric model assumptions for the potential outcomes, which allow us to develop efficient weighting estimators based on the propensity score. We derive finite-sample regret bounds and propose a test for selecting the appropriate model within these assumptions.

There are several directions for future research. First, our asymptotic properties and regret analysis rely on the correct specification of the propensity score model. Future work could develop estimators that are robust to model misspecification by, for example, targeting covariate balance directly \citep{sengupta2025low}. Second, while our analysis has primarily focused on binary treatments, it would be valuable to extend the framework to continuous treatments. Lastly, the variance of the proposed estimator grows rapidly as the policy length increases, which limits the study of optimal policies over long time horizons. Future work could investigate methods to stabilize the estimator along the time dimension.

\bibliographystyle{chicago}
\bibliography{references,policylearning}

%-------------------------------------------------------
%------------------------Appendix-----------------------
%-------------------------------------------------------

\newpage

\doparttoc % Tell to minitoc to generate a toc for the parts
\faketableofcontents % Run a fake tableofcontents command for the partocs
\part{} % Start the document part

\setcounter{page}{1}

\vspace{-20pt}
\begin{center}
{\sc \LARGE Supplementary Appendix for ``Optimal Dynamic Policy for Spatio-temporal Data''}
\end{center}

\allowdisplaybreaks
\appendix
\setcounter{equation}{0}
\renewcommand{\theequation}{A.\arabic{equation}}

\setcounter{table}{0}
\renewcommand{\thetable}{A.\arabic{table}}

\setcounter{figure}{0}
\renewcommand{\thefigure}{A.\arabic{figure}}

\setcounter{theorem}{0}
\renewcommand{\thetheorem}{A.\arabic{theorem}}

\setcounter{assumption}{0}
\renewcommand{\theassumption}{A.\arabic{assumption}}

\setcounter{corollary}{0}
\renewcommand{\thecorollary}{A.\arabic{corollary}}

\setcounter{proposition}{0}
\renewcommand{\theproposition}{A.\arabic{proposition}}

\setcounter{lemma}{0}
\renewcommand{\thelemma}{A.\arabic{lemma}}

\setcounter{remark}{0}
\renewcommand{\theremark}{A.\arabic{remark}}

\makeatletter
\renewcommand{\theHequation}{A.\arabic{equation}}
\renewcommand{\theHtable}{A.\arabic{table}}
\renewcommand{\theHfigure}{A.\arabic{figure}}
\renewcommand{\theHtheorem}{A.\arabic{theorem}}
\renewcommand{\theHassumption}{A.\arabic{assumption}}
\renewcommand{\theHcorollary}{A.\arabic{corollary}}
\renewcommand{\theHproposition}{A.\arabic{proposition}}
\renewcommand{\theHlemma}{A.\arabic{lemma}}
\renewcommand{\theHremark}{A.\arabic{remark}}
\makeatother

\vspace{-60pt}
\addcontentsline{toc}{section}{Supplement} % Add the appendix text to the document TOC
\part{ } % Start the appendix part
\parttoc
\clearpage

\setstretch{1.3}

\section{Optimal mixing parameter}\label{a:sec:lambda}

In this appendix, we derive the expression of $\lambda^\ast$ that minimizes the asymptotic variance of estimator in \cref{eq: estimator_family}. Let $\ppi$ be the targeting policy. Under \Cref{assump: factor_ps,assump: unconfoundedness,assump: additive_model_multi} and the regularity condition in \cref{a: Asymptotic_property}, \cref{lemma:addIPW_true} shows that each estimator in \cref{eq: estimator_family} is asymptotically normal with mean $V^{\ppi}$ and variance 
$$\plim_{T\to\infty}\frac{1}{(T-M+1)^2}\sum_{t=M}^T\Var\left(\rhot(\overline{Y}_t-\lambda \hatV-(1-\lambda)\overline{Y})\mid \lHt[-M]\right).$$
We wish to find $\lambda\in[0,1]$ that minimizes the variance. Note
\begin{equation}
\begin{aligned}
    &\plim_{T\to\infty}\shiftmean\Var\left(\rhot(\overline{Y}_t-\lambda \hatV-(1-\lambda)\overline{Y})\mid \lHt[-M]\right)\\
    = \ &\plim_{T\to\infty}\shiftmean\Var\left(\rhot(\overline{Y}_t-\overline{Y})\mid \lHt[-M])\right)\\
    & -2\lambda\plim_{T\to\infty}\shiftmean(V^{\ppi}-\overline{Y})\Cov(\rhot(\overline{Y}_t-\overline{Y}),\rhot\mid\lHt[-M])\\
    & + \lambda^2\plim_{T\to\infty}\shiftmean(V^{\ppi}-\overline{Y})^2\Var(\rhot\mid\lHt[-M]).
\end{aligned} \label{eq:min_var}
\end{equation}
where we used the fact that $\hatV$ is consistent for $V^{\ppi}$.
Note that \cref{eq:min_var} is quadratic in $\lambda$ under the assumption that $V^{\ppi}-\overline{Y}\neq 0$, which is generally the case so long as the targeting policy is not equal to the observed one. Since $\E[\rhot\mid\lHt[-M]] = 1$ by \cref{prop1_multi}, we have 
\begin{align*}
 \lambda^\ast &= \text{Proj}_{[0,1]}\left(\frac{\plim_{T\to\infty}\shiftmean\Cov(\rhot(\overline{Y}_t-\overline{Y}),\rhot\mid \lHt[-M])}{\plim_{T\to\infty}\shiftmean(V^{\ppi}-\overline{Y})\Var(\rhot\mid\lHt[-M])}\right)\\
 &=\text{Proj}_{[0,1]}\left(\frac{\plim_{T\to\infty}\shiftmean \E[\rhot(\overline{Y}_t-\overline{Y})(\rhot-1)\mid\lHt[-M]]}{\plim_{T\to\infty}\shiftmean(V^{\ppi}-\overline{Y})\E[(\rhot-1)^2\mid\lHt[-M]]}\right), 
\end{align*}
where $\text{Proj}_{[0,1]}(x) = \min(\max(0,x),1),$ is the unique minimizer over $\lambda\in[0,1].$
\paragraph{Consistent estimator for $\lambda^\ast$.} Next we construct a consistent estimator for $\lambda^\ast.$ Define,
\begin{align*}
A_t & = \E[\rhot(\overline{Y}_t-\overline{Y})(\rhot-1)\mid\lHt[-M]],\\ 
B_t & = (V^{\ppi}-\overline{Y})\E[(\rhot-1)^2\mid\lHt[-M]].
\end{align*}
Moreover, define
$\hat A_t = \rhot(\overline{Y}_t-\overline{Y})(\rhot-1)$ and $\hat B_t(V^{\ppi}) = (V^{\ppi}-\overline{Y})(\rhot-1)^2.$ Observe that 
\[
\E\left[\hat A_t-A_t\mid \lHt[-M]\right] = 0\quad\text{and} \quad \E\left[\hat B_t(V^{\ppi})-B_t\mid \lHt[-M]\right]=0
\]

Moreover, by \Cref{assump:Policy_est1}.\ref{assump: 1a} and \Cref{assump: factor_ps}, $\hat A_t-A_t$ and $\hat B_t(V^{\ppi})-B_t$ are bounded, and so is $(\hat A_t-A_t)^2$ and $(\hat B_t(V^{\ppi})-B_t)^2.$ Then, Theorem~1 of \cite{csorgHo1968strong} implies that $\plim_{T\to\infty}\shiftmean(\hat A_t-A_t) = 0,$ and $\plim_{T\to\infty}\shiftmean(\hat B_t(V^{\ppi})-B_t) = 0.$
Since $\hatV[H,]$ is consistent for $V^{\ppi}$, $\plim_{T\to\infty}\shiftmean(\hat B_t(\hatV[H,])-B_t) = 0$ by Slutsky's theorem. Since $\text{Proj}_{[0,1]}(x)$ is continuous, a consistent estimator for $\lambda^\ast$ is given by,$$\text{Proj}_{[0,1]}\left(\dfrac{\shiftmean\hat A_t}{\shiftmean\hat B_t(\hatV[H,])}\right).$$
To ensure the smoothness of the objective function in policy learning step, we consider a smooth approximation of $\text{Proj}_{[0,1]}(x).$ Define
\[
g_T(x) =
\begin{cases}
0, & x \le 0, \\
\dfrac{x^{2}}{2\epsilon_T}, & 0 < x < \epsilon_T, \\
x, & \epsilon_T \le x \le 1-\epsilon_T, \\
1 - \dfrac{(1-x)^{2}}{2\epsilon_T}, & 1-\epsilon_T < x < 1, \\
1, & x \ge 1,
\end{cases}
\]
where $\epsilon_T$ is a positive value with $\lim_{T\to\infty}\epsilon_T=0.$ Then, it is not difficult to show that $g(x)$ converges to $\text{Proj}_{[0,1]}(x)$ uniformly. Thus, $$g_T\left(\dfrac{\shiftmean\hat A_t}{\shiftmean\hat B_t(\hatV[H,])}\right)$$ is a consistent estimator for $\lambda^\ast.$

\section{Asymptotic properties of the proposed estimator}\label{a: Asymptotic_property}

In this section, we present asymptotic properties of the proposed estimator under both true and estimated propensity scores for policies over multiple time periods and provide the corresponding proofs. This section complements \cref{sec:estimation}.

\subsection{Regularity conditions}\label{subsec: regularity}

Consider a policy sequence over $M$ time period $\ppi=(\ppi_1,\dots,\ppi_M).$ We first state the regularity assumptions for establishing the asymptotic normality of the proposed estimator $\hat V^{(\bbeta)}(\ppi)$ when the true propensity score is known, where $\bbeta = (\beta_1,\dots,\beta_M)$ is the sequence of the selecting order of interaction.

\begin{assumption}[Regularity conditions for the proposed estimator with the true propensity score]\label{assump:Policy_est1}  For $\bbeta\in\{1,\dots,n\}^M$ and $\ppi\in\Pi^M.$ The following two conditions hold.
\begin{enumerate}[label=(\alph*)]
    \item\label{assump: 1a}(Bounded outcome and covariates) There exist positive constants $0<\delta_{Y}<\infty$ and $0<\delta_{X}<\infty$ such that $|Y_{ti}|<\delta_{Y}$ and $||\bm X_{ti}||<\delta_{X}$ for all $Y_{ti}\in\lct{\mathcal{Y}}{T}$, $\bm X_{ti}\in\lct{\mathcal{X}}{T}.$
    \item \label{assump: 1b} (Asymptotic variance) There exists a positive definite matrix $\widetilde{\mat{\Sigma}}$ such that 
    $$\frac{1}{T-M+1}\sum_{t=M}^T \Var[\aat\mid \lHt[-M]^\ast]\pto \widetilde{\mat{\Sigma}}$$ as $T\to \infty,$ where $\aat = \left( \widetilde V_{t}^{(\bbeta)}(\ppi;\ttheta^\ast),\rho_{t}^{(\bbeta)}(\ppi;\ttheta^\ast)\right)^\top.$
\end{enumerate}

\end{assumption}
\Cref{assump:Policy_est1}.\ref{assump: 1a} states that there exist uniform upper bounds on the outcome over all time periods and units. \Cref{assump:Policy_est1}.\ref{assump: 1b} requires that the mean of a sequence of conditional variance converges to a positive definite matrix. The assumption about the convergence of the covariance matrix is relatively mild, given that the random vector $\aat$ is bounded above and bounded away from $\bm 0$.

The following assumptions are needed for the asymptotic normality of the estimator with the estimated propensity score.

\begin{assumption}[Regularity conditions for the propensity score model]\label{assump: ps_model}
Assume that the parametric form of the propensity score indexed by $\ttheta,  
\P\left(\bm W_t=\bm w_t \mid \lHt[-1] ; \ttheta\right)$, is correctly specified and differentiable with respect to $\ttheta\in\mathbb{R}^K$, and let $$\psi\left(W_t, \lHt[-1] ; \ttheta\right)=\frac{\partial}{\partial \ttheta} \log \P\left(W_t=w_t \mid \lHt[-1]=\lct{h}{t-1} ; \ttheta\right)$$ be twice continuously differentiable score functions. Let $\ttheta^\ast$ denote the true values of the parameters, where $\ttheta^\ast$ is in an open subset of the Euclidean space. Denote $\mathcal{F}_t=\lHt[-M+1]^\ast$. We assume that the following conditions hold: 
\begin{enumerate}
    \item \label{assump: ps_model1}
    \begin{enumerate}[label = (\alph*)]
        \item $\E_{\ttheta^\ast}\left[\left\|\psi\left(W_t, \lHt[-1] ; \ttheta^\ast\right)\right\|^2\right]<\infty$,
        \item \label{assump: ps_model1b}There exists a positive definite matrix $\mat{\Sigma}_{ps}$ such that
        $$
        \frac{1}{T} \sum_{t=1}^T \E_{\ttheta^\ast}\left(\psi\left(W_t, \lHt[-1] ; \ttheta^\ast\right) \psi\left(W_t, \lHt[-1] ; \ttheta^\ast\right)^{\top} \mid \mathcal{F}_{t-1}\right) \pto \mat{\Sigma}_{ps}
        $$
        \item $\frac{1}{T} \sum_{t=1}^T \E_{\ttheta^\ast}\left[\left\|\psi\left(W_t, \lHt[-1] ; \ttheta^\ast\right)\right\|^2 I\left(\left\|\psi\left(W_t, \lHt[-1] ; \ttheta^\ast\right)\right\|>\epsilon \sqrt{T}\right) \mid \mathcal{F}_{t-1}\right] \pto 0$, for all $\epsilon>0$, as $T\to\infty.$
    \end{enumerate}
\item

For all $k, j$, if we denote the $k^{th}$ element of the $\psi\left(w_t, \lHt[-1] ; \ttheta\right)$ vector by $\psi_k\left(w_t, \lHt[-1] ; \ttheta\right)$ and $P_{kj t}=\frac{\partial}{\partial \ttheta_j} \psi_k\left(W_t, \lHt[-1] ; \ttheta\right)|_{\ttheta^\ast}$, then $\E_{\ttheta^\ast}\left[\left|P_{k j t}\right|\right]<\infty$ and there exists $0<r_{k j} \leq 2$ such that $\sum_{t=1}^T \frac{1}{t^{r_{k j}}} \E_{\ttheta^\ast}\left(\left|P_{k j t}-\E_{\ttheta^\ast}\left(P_{k j t} \mid \mathcal{F}_{t-1}\right)\right|^{r_{k j}} \mid \mathcal{F}_{t-1}\right) \pto 0$
\item\label{assump: IPWest3} There exists an integrable function $\ddot{\psi}\left(w_t, \lct{h}{t-1}\right)$ such that $\ddot{\psi}\left(w_t, \lct{h}{t-1}\right)$ dominates the second partial derivatives of $\psi\left(w_{t}, \lct{h}{t-1} ; \ttheta\right)$ in a neighborhood of $\ttheta^\ast$ for all $\left(w_{t}, \lct{h}{t-1}\right)$
\item  Let  $\psi_T(\ttheta) = \sum_{t = M}^T\psi\left(W_t, \lHt[-1] ; \ttheta\right)$ and $\psi(\ttheta) = E_{\ttheta^\ast}[\psi_T(\ttheta)|\lHt[-1]^\ast].$ There exists a unique solution $\hat{\ttheta}$ to $\psi_T(\ttheta)=0$ and 
$\sup_{{\ttheta}\in \mathcal{N}}|\psi_T(\ttheta)-\psi(\ttheta)|\pto0$ as $T\to\infty,$ where $\mathcal{N}$ is an open set that contains $\hat\ttheta$ and $\ttheta^\ast.$
\item\label{assump: ps_model5} Let $$s(\lHt[-1],\bm W_t,\bm Y_t;\ttheta) = \left(\widetilde{V}^{(\bbeta)}_{t}(\ppi;\ttheta)-V_t^{\ppi}\left(\lHt[-M]\right),\ \rho_{t}^{(\bbeta)}(\ppi)-1\right)^\top,$$ the following conditions hold.
\begin{enumerate}[label=(\alph*)]
    \item\label{assump: IPWestadda} There exists $\mat{U}\in\mathbb{R}^{K\times2}$ such that $$\shiftmean \E_{\ttheta^\ast}[\psi(W_t,\lHt[-1];\ttheta^\ast)s(\lHt[-1],W_t,Y_t;\ttheta^\ast)^\top\mid\lHt[-M]^\ast]\pto \mat{U},\ \mathrm{and}\ $$
    $\begin{bmatrix}
        \widetilde{\mat{\Sigma}} & \mat{U}^\top\\
        \mat{U} & \mat{\Sigma}_{ps}
        \end{bmatrix}$ is positive definite.
    \item\label{assump: IPWestaddb} If $P_{jt} = \frac{\partial}{\partial \ttheta_j} s(\lHt[-1],W_t,Y_t;\ttheta)\Big|_{\ttheta^\ast}$, where $\ttheta_j$ is the $j^{th}$ entry of $\ttheta$, then there exists $r_j\in(0,2]$ such that
 $$\shiftmean \frac{1}{t^{r_j}}(|P_{jt}-\E_{\ttheta^\ast}[P_{jt}\mid\lHt[-M]^\ast]|)\pto 0.$$
\end{enumerate}
\end{enumerate}

\end{assumption}

\Cref{assump: ps_model} states that the model is correctly specified and the score function is bounded in the $L^2$-norm. Moreover, the average expectation of the product of the score function with itself over all time periods converges to a specific positive definite matrix. The assumption also limits the extent to which the derivative of the score function varies around its conditional expectation.  Moreover, the assumption ensures that the second partial derivatives of the score function are bounded in magnitude by an integrable function, which controls their growth and guarantees stability in a neighborhood of $\ggamma^\ast.$ Lastly, the assumption regulates the behavior of function $s(\lHt[-1],W_t,Y_t;\ttheta)$ and the propensity score function. Specifically, it states that the average expectation of the outer product of $s(\lHt[-1],W_t,Y_t;\ttheta)$ and $s(\lHt[-1],W_t,Y_t;\ttheta)$ converges to a matrix as $T$ grows and $\mat{\Sigma}^\ast$ is positive definite. Furthermore, it limits how much the derivative of $s(\lHt[-1],W_t,Y_t;\ttheta)$
 can vary around its conditional
expectation.

\begin{remark}\label{rmk: ps}
    If Assumption~\ref{assump: unconfoundedness} holds, and $\psi(\bm w_t,\lct{h}{t-1};\ttheta)$ are score functions that satisfy Assumption~\ref{assump: ps_model}, then
    \begin{enumerate}
        \item $E_{\ttheta^\ast}[\psi(\bm W_t,\lHt[-1];\ttheta^\ast)] = 0$, $E_{\ttheta^\ast}\left[||\psi(\bm W_t,\lHt[-1];\ttheta^\ast)||\right]<\infty$, and
        \item $E_{\ttheta^\ast}\left(-\frac{\partial}{\partial\ttheta^\top}\psi(\bm W_t,\lHt[-1];\ttheta)\big|_{\ttheta^\ast}\mid \ft\right) = E_{\ttheta^\ast}\left(\psi(\bm W_t,\lHt;\ttheta)\psi(\bm W_t,\lHt;\ttheta)^\top\mid \ft\right)$ which implies that $\frac{1}{T}\sum_{t= 1}^T E_{\ttheta^\ast}\left(-\frac{\partial}{\partial\ttheta^\top}\psi(\bm W_t,\lHt[-1];\ttheta)\big|_{\ttheta^\ast}\mid \ft\right)\pto\mat{\Sigma}_{ps}.$
        \item $\hat\ttheta\pto\ttheta^\ast$ as $T\to\infty.$
    \end{enumerate}
\end{remark}
\begin{proof}
    The proof of 1 and 2 is similar to the proof of Lemma A.1 in \cite{papadogeorgou2022causal}, and thus omitted. We will show that $\hat\ttheta$ is consistent for $\ttheta^\ast.$ By Remark~\ref{rmk: ps} and the definition of $\hat\ttheta,$ we have
    $$|\psi(\hat\ttheta)-\psi(\ttheta^\ast)| = |\psi(\hat\ttheta)-\psi_T(\hat\ttheta)|\leq \sup_{\ttheta\in\mathcal{N}}|\psi(\ttheta)-\psi_T(\ttheta)|\pto 0 $$ as $T\to 0.$ Let $\epsilon>0$ be arbitrary. Since $\hat\ttheta$ is the unique solution to $\psi_T(\ttheta)=0$, there exist $\delta>0$ such that $|\psi(\ttheta)-\psi(\ttheta^\ast)|>\delta$ for all $\ttheta$ with $|\ttheta-\ttheta^\ast|>\epsilon.$ Thus,
    $$\P(|\hat\ttheta-\ttheta^\ast|>\epsilon)\leq \P(|\psi_T(\ttheta)-\psi(\ttheta^\ast)|>\delta)\to 0,$$ as $T\to\infty$. So $\hat\ttheta\pto\ttheta^\ast,$ as $T\to\infty$.
\end{proof}

\subsection{Asymptotic normality based on true propensity score}\label{a:subsec:normality_true_ps}
In this section, we establish the asymptotic normality of the proposed estimator based on true propensity score by utilizing the martingale theorem. The results shown in this section is similar to the asymptotic result in \cite{zhou2024estimating}. Consider a policy over $M$ time periods, $\ppi^{(M)} = (\ppi_1,\dots,\ppi_M)$ and $\bbeta\in\{1,\dots,n\}^M$.
Define 
\begin{equation}\label{a:eq:tildeV}
 \tildeVPS[\bbeta][\ppi][\ttheta][_{ti}][] = \rhotPS[\bbeta][\ppi][\ttheta][]Y_{ti}\quad\text{and}\quad \tildeVPS[\bbeta][\ppi][\ttheta][_{t}][] = \frac{1}{n}\sum_{i=1}^n \tildeVPS[\bbeta][\ppi][\ttheta][_{ti}][].
\end{equation}
Let
\begin{equation}\label{a:eq:hatVt}
\hat V_t^{(\bbeta)}(\ppi^{(M)};\ttheta)=\hat\lambda\dfrac{\rhotPS[\bbeta][\ppi^{(M)}][\ttheta][]}{\rhobar}+(1-\hat\lambda)(\rhotPS[\bbeta][\ppi^{(M)}][\ttheta][]+(1-\rhobar))\overline{Y}_t
\end{equation} denote the time-$t$ contribution to $\hatVPS[][\bbeta][\ppi^{(M)}][\ttheta
].$ Let 
\begin{equation*}
  V_t^{\ppi^{(M)}}\left(\lHt[-M]\right) = \frac{1}{n}\sum_{i=1}^nY_{ti}^{\ppi^{(M)}}\left(\lHt[-M]\right)  
\end{equation*}
be the true policy value averaged across all units for time period $t.$
We begin with proving the following sequence is a martingale difference sequence. 

\begin{lemma}\label{lemma1}
Let $\bbeta=(\beta_1,\dots,\beta_M)\in\{1,2,\dots,n\}^M,$ and $\ppi^{(M)} = (\ppi_1,\dots,\ppi_M)\in\Pi^M$  Under \Cref{assump: unconfoundedness}, $\E_{\ttheta^\ast}[\rho^{(\bbeta)}_{t}(\ppi^{(M)};\ttheta^\ast)\mid \lHt[-M]]$ for all $1\le t\le T.$
\end{lemma}

\begin{proof}
        For all $1\le i<j\le n$, $1\le t\le T$, we have
\begin{align*}
   \E_{\ttheta^\ast}[\rho_{ti}(\ppi_1;\ttheta^\ast)\mid\lHt[-1]^\ast] & = \E_{\ttheta^\ast}\left[\frac{\mathds{1}(W_{ti} = \pi_{1i}\left(\lHt[-1]\right))}{e_{ti}(W_{ti};\ttheta^\ast)}\mid \lHt[-1]^\ast\right] \\
   & = \frac{\P(W_{ti} = \pi_{1i}\left(\lHt[-1]\right)\mid\lHt[-1]^\ast)}{e_{ti}(W_{ti};\ttheta^\ast)}\\
   & = \frac{\P(W_{ti} = \pi_{1i}\left(\lHt[-1]\right)\mid\lHt[-1])}{e_{ti}(W_{ti};\ttheta^\ast)} \hspace{2cm}\text{(Assumption 1)}\\
   & = \frac{e_{ti}(W_{ti};\ttheta^\ast)}{e_{ti}(W_{ti};\ttheta^\ast)}=1
\end{align*}
Moreover, by Assumption \ref{assump: factor_ps}, we have $$\E_{\ttheta^\ast}[\rho_{ti}(\ppi_1;\ttheta^\ast)\rho_{tj}(\ppi_1;\ttheta^\ast)\mid\lHt[-1]^\ast] = \E_{\ttheta^\ast}[\rho_{ti}(\ppi_1;\ttheta^\ast)\mid\lHt[-1]^\ast]\E_{\ttheta^\ast}[\rho_{tj}(\ppi_1;\ttheta^\ast)\mid\lHt[-1]^\ast] = 1.$$ 
Thus,
\begin{align*}
    &\E_{\ttheta^\ast}\left[\rho_t^{(\beta_1)}(\ppi_1;\ttheta^\ast)\mid\lHt[-1]^\ast\right] \\
    =& \E_{\ttheta^\ast}\Big[1+\sum_{i=1}^n(\rho_{ti}(\ppi_1;\ttheta^\ast) -1) +\sum_{1\leq i_1<i_2\leq n}(\rho_{ti_1}(\ppi_1;\ttheta^\ast) -1)(\rho_{ti_2}(\ppi_1;\ttheta^\ast) -1)+\dots\\
    &+\sum_{1\le i_1<\dots<i_{\beta_1}\le n}(\rho_{ti_1-1}(\ppi_1;\ttheta^\ast)-1 )\dots(\rho_{ti_{\beta_1}}(\ppi_1;\ttheta^\ast) -1)\mid\lHt[-1]^\ast\Big]=1.
\end{align*}
Observe that
\begin{align*}
    &\E_{\ttheta^\ast}[\rho_t^{(\beta_1,\beta_2)}(\ppi_1,\ppi_2;\ttheta^\ast)\mid \lHt[-2]]\\ = &\E_{\ttheta^\ast}\left[\rho_{t-1}^{(\beta_2)}(\ppi_2;\ttheta^\ast)\E_{\ttheta^\ast}[\rho_t^{(\beta_1)}(\ppi_1;\ttheta^\ast)\mid \lHt[-1]]\mid\lHt[-2]\right]\\
    = &\E_{\ttheta^\ast}[\rho_{t-1}^{(\beta_2)}(\ppi_2;\ttheta^\ast)\mid\lHt[-2]].
\end{align*}
Using a similar argument as above, $\E_{\ttheta^\ast}[\rho_{t-1}^{(\beta_2)}(\ppi_2;\ttheta^\ast)\mid\lHt[-2]]=1$, and thus $\E_{\ttheta^\ast}[\rho_t^{(\beta_1,\beta_2)}(\ppi_1,\ppi_2;\ttheta^\ast)\mid \lHt[-2]]=1.$
By induction, we can show that $E_{\ttheta^\ast}[\rho_{t}^{(\bbeta)}(\ppi^{(M)};\ttheta^\ast)\mid\lHt[-M]] = 1,$ as desired.
\end{proof}

\begin{lemma}\label{a: lemma_MDS}
  Under \Cref{assump: unconfoundedness,assump: additive_model_multi,assump: factor_ps,assump:Policy_est1}, the sequence \\
  $\aat = \left(\widetilde{V}^{(\bbeta)}_{t}(\ppi^{(M)};\ttheta^\ast)-V_t^{\ppi}\left(\lHt[-1]\right),\rho_{t}^{(\bbeta)}(\ppi^{(M)};\ttheta^\ast)-1\right)^\top$ is a martingale difference sequence with respect to filtration $\mathcal{F}_t = \lHt[-M+1]^\ast.$
\end{lemma}

\begin{proof}
By \cref{lemma1}, we have $\E_{\ttheta^\ast}[\rhot\mid\lHt[-M]]=1.$ Then we show that $\E_{\ttheta^\ast}[\widetilde{V}_t^{(\bbeta)}(\ppi;\ttheta^\ast)\mid \lHt[-M]] = V_t^{\ppi}\left(\lHt[-M]\right).$
Note that under \Cref{assump: factor_ps}, the matrix $$\E\left[\phi^{(\beta_1)}(\bm W_t)\phi^{(\beta_1)}(\bm W_t)^\top \mid \lHt[-1]\right]$$ is invertible and $\widetilde{V}^{(\beta_1)}(\ppi_1)$ can be rewritten as 
\begin{align}\label{eq:V_matrix_form}
    \widetilde{V}^{(\beta_1)}_t(\ppi_1;\ttheta^\ast) &= \frac{1}{n}\sum_{i=1}^n\phi^{(\beta_1)}(\ppi_1\left(\lHt[-1]\right))^\top\E\left[\phi^{(\beta_1)}(\bm W_t)\phi^{(\beta_1)}(\bm W_t)^\top \mid \lHt[-1]\right]^{-1}\phi^{(\beta_1)}(\bm W_t)Y_{ti}
\end{align}
By \Cref{assump: additive_model_multi} and \cref{eq:V_matrix_form}
\begin{align*}
    &\E_{\ttheta^\ast}\left[\widetilde V_t^{(\beta_1)}(\ppi_1;\ttheta^\ast)\mid \lHt[-1] \right]\\
    =&\E_{\ttheta^\ast}\left[\frac{1}{n}\sum_{i=1}^n\phi^{(\beta_1)}(\ppi_1\left(\lHt[-1]\right))^\top\E\left[\phi^{(\beta_1)}(\bm W_t)\phi^{(\beta_1)}(\bm W_t)^\top \mid \lHt[-1]\right]^{-1}\phi^{(\beta_1)}(\bm W_t)Y_{ti}\mid\lHt[-1]\right]\\
    =&\frac{1}{n}\sum_{i=1}^n\phi^{(\beta_1)}(\ppi_1\left(\lHt[-1]\right))^\top\E\left[\phi^{(\beta_1)}(\bm W_t)\phi^{(\beta_1)}(\bm W_t)^\top \mid \lHt[-1]\right]^{-1}\E\left[\phi^{(\beta_1)}(\bm W_t)\phi^{(\beta_1)}(\bm W_t)^\top\bm g_{1ti}\left(\lHt[-1]\right)\mid\lHt[-1]\right]\\
    =& \frac{1}{n}\sum_{i=1}^n\bm g_{1ti}\left(\lHt[-1]\right)^\top \phi^{(\beta_1)}(\ppi_1\left(\lHt[-1]\right)) \\
    =&\frac{1}{n}\sum_{i=1}^n Y_{ti}^{\ppi_1}\left(\lHt[-1]\right) = V_t^{\ppi_1}\left(\lHt[-1]\right).
\end{align*}
Similarly,
 $\widetilde{V}^{(\beta_1,\beta_2)}(\ppi_1,\ppi_2)$ can be rewritten as 
\begin{align}
    \widetilde{V}^{(\beta_1,\beta_2)}_t(\ppi_1,\ppi_2;\ttheta^\ast) &= \frac{1}{n}\sum_{i=1}^n\phi^{(\beta_2)}\left(\ppi_2(\lHt[-2])\right)^\top\E\left[\phi^{(\beta_2)}(\bm W_{t-1})\phi^{(\beta_2)}(\bm W_{t-1})^\top \mid \lHt[-2]\right]^{-1}\phi^{(\beta_2)}(\bm W_t)\widetilde{V}^{(\beta_1)}_{ti}(\ppi_1),
\end{align}
and we have
\begin{align*}
    &\E_{\ttheta^\ast}\left[\widetilde V_t^{(\beta_1,\beta_2)}(\ppi_1,\ppi_2;\ttheta^\ast)\mid \lHt[-2] \right]\\
    =&\E_{\ttheta^\ast}\left[E_{\ttheta^\ast}\left[\widetilde V_t^{(\beta_1,\beta_2)}(\ppi_1,\ppi_2;\ttheta^\ast)\mid\lHt[-1]\right]\mid \lHt[-2] \right]\\
    =&\E_{\ttheta^\ast}\left[\phi^{(\beta_2)}(\ppi_2\left(\lHt[-2]\right))^\top\E\left[\phi^{(\beta_2)}(\bm W_{t-1})\phi^{(\beta_2)}(\bm W_{t-1})^\top \mid \lHt[-2]\right]^{-1}\phi^{(\beta_2)}(\bm W_t)E_{\ttheta^\ast}\left[\widetilde V_t^{(\beta_1)}(\ppi_1;\ttheta^\ast)\mid\lHt[-1]\right]\mid \lHt[-2] \right]\\
    =&\E_{\ttheta^\ast}\left[\phi^{(\beta_2)}(\ppi_2\left(\lHt[-2]\right))^\top\E\left[\phi^{(\beta_2)}(\bm W_{t-1})\phi^{(\beta_2)}(\bm W_{t-1})^\top \mid \lHt[-2]\right]^{-1}\phi^{(\beta_2)}(\bm W_t)Y_{ti}^{\ppi_1}\left(\lHt[-1]\right)\mid \lHt[-2] \right]\\
    =&\frac{1}{n}\sum_{i=1}^n\phi^{(\beta_2)}(\ppi_2\left(\lHt[-2]\right))^\top\E\left[\phi^{(\beta_2)}(\bm W_{t-1})\phi^{(\beta_2)}(\bm W_{t-1})^\top \mid \lHt[-2]\right]^{-1}\\&\qquad\E\left[\phi^{(\beta_2)}(\bm W_{t-1})\phi^{(\beta_2)}(\bm W_{t-1})^\top\bm g_{2ti}\left(\lHt[-2]\right)\mid\lHt[-2]\right]\\
    =& \frac{1}{n}\sum_{i=1}^n\bm g_{2ti}\left(\lHt[-2]\right)^\top \phi^{(\beta_2)}(\ppi_2\left(\lHt[-2]\right)) \\
    =&\frac{1}{n}\sum_{i=1}^n Y_{ti}^{\ppi_1,\ppi_2}\left(\lHt[-2]\right) = V_t^{\ppi_1,\ppi_2}\left(\lHt[-2]\right).
\end{align*}
By induction, we can show that $\E_{\ttheta^\ast}[\widetilde{V}_t^{(\bbeta)}(\ppi;\ttheta^\ast)\mid \lHt[-M]]=V_t^{\ppi}.$
Combining these results, we have $\E[\aat\mid \lHt[-M]] = \bm 0.$ According to Assumption ~\Cref{assump:Policy_est1}.\ref{assump: 1a} and Assumption~\ref{assump: factor_ps}, both $\widetilde V_t^{(\bbeta)}(\ppi^{(M)};\ttheta^\ast)$ and $\rho_t^{(\bbeta)}(\ppi^{(M)};\ttheta^\ast)$ are bounded, so are their expectation. Thus, $\aat$ is a martingale difference sequence as desired.

\end{proof}

\begin{proof}[Proof of \cref{lemma:addIPW_true}]
    Let $$\aat^\ast = \left(\widetilde{V}^{(\bbeta)}_{t}(\ppi^{(M)};\ttheta^\ast)-V_t^{\ppi}\left(\lHt[-M]\right),\rhotPS[\bbeta][\ppi][\ttheta^\ast][]\right)^\top,$$ By \cref{a: lemma_MDS}, $\aat^\ast$ is a martingale difference sequence. By Theorem A.4 in \cite{zhou2024estimating},
    \[
    \frac{1}{\sqrt{T-M+1}}\sum_{t=M}^T\aat^*\dto N(\bm 0, \widetilde{\mat{\Sigma}}),
    \] as $T\to 0.$
Let
\begin{align*}
    &\mmu =\left(V^{\ppi},\;\rho\right)^\top,\\
    &\hat\mmu = \left(\tildeVTruePS,\;\rhobar\right)^\top,
\end{align*}
where $\rho=1.$
Finally, we apply the Delta method.
Define $$h(\mmu) = \lambda\frac{\tildeVTruePS}{\rho}+(1-\lambda)\left(\tildeVTruePS+(1-\rhobar)\overline{Y}\right)$$ The Jacobian matrix is 
$$
\mat{J}:=\mat{J}(\mmu^\ast) = 
\begin{bmatrix}
    1 &-\lambda V^{\ppi}-(1-\lambda)\overline{Y}
\end{bmatrix}
$$ Then, we obtain, 
$$\sqrt{T-M+1}(h(\hat\mmu)-h(\mmu^\ast))\dto N(\bm 0,\mat{J}\widetilde{\mat{\Sigma}}\mat{J}^\top)$$ as $T\to\infty.$ Since $h(\hat\mmu) = \hat V^{(\bbeta)}(\ppi,\lambda)$ and $h(\mmu^\ast) = V^{\ppi},$ we have
$$\sqrt{T-M+1}\left(\hat V^{(\bbeta)}(\ppi,\lambda)-V^{\ppi}\right)\overset{d}{\to}N(\bm 0,\mat{J}\widetilde{\mat{\Sigma}}\mat{J}^\top),$$ as $T\to\infty.$
\end{proof}

\begin{theorem}[Asymptotic normality of the normalized estimator using the true propensity score]\label{thm:addIPW_true}
Suppose that \Cref{assump: unconfoundedness,assump: additive_model_multi,assump: factor_ps,assump:Policy_est1} hold.  Then
$$\sqrt{T-M+1}\left(\hatV-V^{\ppi}\right)\overset{d}{\to}N(\bm 0,\mat{J}\widetilde{\mat{\Sigma}}\mat{J}^\top),$$ where $\widetilde{\mat{\Sigma}}$ represents the probability limit of $\frac{1}{T-M+1}\sum_{t=M}^T \mat{\Sigma}_t$ as $T\to \infty$ with
$\mat{\Sigma}_t = \Var(\widetilde V^{(\bbeta)}_t(\ppi;\ttheta^\ast)\mid \lHt[-M]^\ast)$ for $t\ge M$, and $\mat{J} = \begin{pmatrix}
    1  &-\lambda^\ast V^{\ppi}-(1-\lambda^\ast)\overline{Y}
\end{pmatrix}.$
\end{theorem}

\begin{proof}
 By \cref{lemma:addIPW_true} with $\lambda = \lambda^\ast$,
 \[
 \sqrt{T-M+1} \hat V^{(\bbeta)}(\ppi,\lambda^\ast)-V^{\ppi}\dto N(\bm 0, \bm J\widetilde{\bm\Sigma}\bm J^\top),
 \]
 as $T\to\infty.$ Moreover, by \cref{lemma:addIPW_true} with $\lambda = 1,$ we have $\sqrt{T-M+1}\hatV[H,]$ is asymptotically normal and thus  $\sqrt{T-M+1}\hatV[H,] = O_p(1).$ Similarly, $$\sqrt{T-M+1}\hatV[S,] = O_p(1).$$ Since $\hat\lambda$ is a consistent estimator for $\lambda^\ast,$ we have $(\hat\lambda-\lambda^\ast) = o_p(1).$ Combining these results,
 \begin{align*}
     &\sqrt{T-M+1}\left(\hatV-\hat V^{(\bbeta)}(\ppi,\lambda^\ast)\right)\\
     =&\sqrt{T-M+1}(\hat\lambda-\lambda^\ast)(\hatV[H,]-\hatV[S,])\\
     =& o_p(1)O_p(1) = o_p(1).
 \end{align*}
 So
 \begin{align*}
     &\sqrt{T-M+1}\left(\hatV-V^{\ppi}\right)\\
     =&\sqrt{T-M+1}\left(\hat V^{(\bbeta)}(\ppi,\lambda^\ast)-V^{\ppi}\right)+\sqrt{T-M+1}\left(\hatV-\hat V^{(\bbeta)}(\ppi,\lambda^\ast)\right)\\
     =&\sqrt{T-M+1}\left(\hat V^{(\bbeta)}(\ppi,\lambda^\ast)-V^{\ppi}\right)+o_p(1)\dto N(\bm 0, \bm J\widetilde{\bm\Sigma}\bm J^\top),
 \end{align*}
 as $T\to\infty.$
\end{proof}

\begin{remark}\label{rmk: bound_trueps}
Since $\E[\aat\aat^\top\mid\lHt[-M]]-\Var(\aat^*\mid\lHt[-M])$ is positive semidefinite,
\begin{align*}
B&:=\shiftmean\bm J\E[\aat\aat^\top\mid\lHt[-M]]\bm J^\top
\end{align*}
is an upper bound for $\bm J\widetilde{\bm \Sigma}\bm J^\top.$ We can construct a consistent estimator for this variance bound. Let $B_t = \bm J\E[\aat\aat^\top\mid\lHt[-M]]\bm J^\top.$ By simplification, 
$$B_t = \E\left[\rhotTruePS^2(\overline{Y}_t - \lambda^\ast V^{\ppi}-(1-\lambda^\ast)\overline{Y})^2\mid \lHt[-M]\right].$$ Let $\hat B_t(\lambda^\ast,V^{\ppi}) = \rhotTruePS^2(\overline{Y}_t - \lambda^\ast V^{\ppi}-(1-\lambda^\ast)\overline{Y})^2.$ Note that $\E[\hat B_t(\lambda^\ast,V^{\ppi})-B_t\mid \lHt[-M]]$ and $\hat B_t-B_t$ is bounded.  Then, Theorem~1 of \cite{csorgHo1968strong} implies:
\[
\shiftmean \left(\hat B_t(\lambda^\ast,V^{\ppi})-B_t\right)\pto 0,
\]
as $T\to\infty.$
Since $\hat\lambda\pto\lambda^\ast,$  and $\hatVTruePS\pto V^{\ppi}$, as $T\to\infty,$ we have \[
\shiftmean \left(\hat B_t(\hat\lambda,\hatVTruePS)-B_t\right)\pto 0,
\]
as $T\to\infty.$ Let 
\[
E = \shiftmean\left(\hat\lambda\rhotPS/\rhobar+(1-\hat\lambda)\right)^2\left(\overline{Y}_t - \hat\lambda V^{\ppi}-(1-\hat\lambda)\overline{Y}\right)^2
\]
Since $E$ is bounded and $\rhobar\pto 1$ as $T\to\infty,$ we have $(1-\rhobar)E\pto 0$ and $(1-\rhobar)^2E\pto 0.$ So 
 \[
\shiftmean \left(\hat B_t(\hat\lambda,\hatVTruePS)-2(1-\rhobar)E+(1-\rhobar)^2E-B_t\right)\pto 0,
\] By direct calculation, we have that
$$\text{Var}^{\text{UB}}(\hatVTruePS) = \frac{1}{(T-M+1)^2}\sum_{t=M}^T\E\left[\rhotTruePS^2\left( \overline{Y}_t-\left(\lambda^\ast V^{\ppi}+(1-\lambda^\ast)\overline{Y}\right)\right)^2\mid \lHt[-M]\right]$$ is an upper bound of the asymptotic variance of  $\hatVTruePS.$ Moreover, 
\begin{equation*}
\varub(\hatVPS) = \frac{1}{(T-M+1)^2}\sum_{t=M}^T\tilde\rho_t^{(\bbeta)}(\ppi;\ttheta^\ast)^2\left(\overline{Y}_{t}-\left(\hat\lambda\hatVTruePS+(1-\hat\lambda)\overline{Y}\right)\right)^2,
\end{equation*}
where $\tilde\rho_t^{(\bbeta)}(\ppi;\ttheta^\ast) = \hat\lambda\dfrac{\rhotTruePS}{\rhobar}+(1-\hat\lambda)(\rhotTruePS+(1-\rhobar))$,
is a consistent estimator for $\text{Var}^{\text{UB}}(\hatVTruePS).$
%     Let
%   \begin{align*}
%        &\hat{\mat{\Sigma}}(\ttheta^\ast) = \shiftmean\hat{\mat{\Sigma}}_t(\ttheta^\ast)\text{ with }\hat{\mat{\Sigma}}_t(\ttheta^\ast) = \aat(\ttheta^*){\aat(\ttheta^\ast)}^\top, \\
%        & \aat(\ttheta^\ast) = \left(\tildeVTruePS,\rhotPS\right)^\top,\ \mathrm{and}\
%        \\
%        &\hat{\mat{J}}(\ttheta^\ast) = \begin{bmatrix}
%     \hat\lambda\dfrac{1}{\rhobar}  &-\left(\hat\lambda\hatVTruePS+(1-\hat\lambda)\overline{Y}\right)
% \end{bmatrix} 
% \end{align*}

\end{remark}

\begin{remark}\label{a:projection}

When the interaction orders $\bbeta$ in \cref{assump: additive_model_multi} are misspecified, our estimator in \cref{eq:hatV_lambdahat} targets a projection estimand obtained through sequential projections of the outcome surfaces. In particular, for a one-period policy $\ppi_1$, the projected policy value at time $t$ is defined as
\begin{align*} &\projVt[\beta_1][\ppi_1][-1] \\:= &\phi^{(\beta_1)}(\ppi_1\left(\lHt[-1]\right))^\top\E\left[\phi^{(\beta_1)}(\bm W_t)\phi^{(\beta_1)}(\bm W_t)^\top \mid \lHt[-1]\right]^{-1}\E\left[\phi^{(\beta_1)}(\bm W_t)\overline{Y}_{t}\mid\lHt[-1]\right] 
\end{align*}
The quantity $\projVt[\beta_1][\ppi_1][-1]$ can be interpreted as the policy value obtained by projecting the average outcome surface
$
\overline{Y}_{t}(\lWt[-1],\bm w_t)
$
onto the linear space spanned by $\phi^{(\beta_1)}(\bm w_t)$. The projection is defined with respect to the conditional distribution of $\bm W_t$ given $\lHt[-1]$ and is evaluated at the treatment assignment specified by $\ppi_1$. For a two-period policy $(\ppi_1,\ppi_2)$, we similarly define the projected policy value by projecting $\projVt[\beta_1][\ppi_1][-1]$ at the preceding period:
\begin{align*} &\projVt[\beta_1,\beta_2][\ppi_1,\ppi_2][-2]\\:= &\phi^{(\beta_2)}(\ppi_2\left(\lHt[-2]\right))^\top\E\left[\phi^{(\beta_2)}(\bm W_t)\phi^{(\beta_2)}(\bm W_t)^\top \mid \lHt[-2]\right]^{-1}\E\left[\phi^{(\beta_2)}(\bm W_t)\projVt[\beta_1][\ppi_1][-1]\mid\lHt[-2]\right] \end{align*}

Continuing this construction, we define the projected policy value at time $t$ as
\begin{align}\label{a:eq:projection} &\projVt[\bbeta][\ppi^{(M)}][-M]\notag \\:= &\phi^{(\beta_M)}(\ppi_M\left(\lHt[-M]\right))^\top\E\left[\phi^{(\beta_M)}(\bm W_t)\phi^{(\beta_M)}(\bm W_t)^\top \mid \lHt[-M]\right]^{-1}\\&\quad\E\left[\phi^{(\beta_M)}(\bm W_t) \projVt[\beta_1,\dots,\beta_{M-1}][\ppi_1,\dots,\ppi_{M-1}][-M+1]\mid \lHt[-M]\right]\notag \end{align}

Using arguments similar to those in the proofs of \cref{lemma1} and \cref{thm:addIPW_true}, we can show that $\tildeVPS[\bbeta][\ppi][\ttheta][_{t}][]$ in \cref{a:eq:tildeV} is conditionally unbiased for $\projVt[\bbeta][\ppi^{(M)}][-M]$ given $\lHt[-M]$. Therefore, without requiring the interaction orders $\bbeta$ in \cref{assump: additive_model_multi} to be correctly specified, our mixing estimator in \cref{eq:hatV_lambdahat} is consistent for the projection estimand
$$\projV := \shiftmean \projVt.$$
Let 
\begin{equation}\label{a:eq：true_Value_t}
  V_t^{\ppi^{(M)}}\left(\lHt[-M]\right) = \frac{1}{n}\sum_{i=1}^nY_{ti}^{\ppi^{(M)}}\left(\lHt[-M]\right)  
\end{equation}
be the true policy value averaged across all units for time period $t$. If the interaction orders $\bbeta$ in \cref{assump: additive_model_multi} are correctly specified, the projected policy value equals the unprojected policy value:
\begin{equation*}
\projVt = V_t^{\ppi,(\bbeta)}\quad\text{and}\quad \projV = V^{\ppi,(\bbeta)}.
\end{equation*}

\end{remark}

\subsection{Asymptotic normality based on estimated propensity score}
In this section, we show the asymptotic properties of the proposed estimator based on estimated propensity score.

\begin{theorem}\label{prop:addIPW_est}
Suppose that \Cref{assump: unconfoundedness,assump: additive_model_multi,assump: factor_ps,assump:Policy_est1} hold. Let $\lambda\in\mathbb{R}.$ Define 
$\hat V^{(\beta)}(\ppi^{(M)};\lambda,\hat\ttheta) = \lambda\hatVPS[H,]+(1-\lambda)\hatVPS[S,].$  Then
$$\sqrt{T-M+1}\left(\hat V^{(\bbeta)}(\ppi;\lambda,\hat\ttheta)-V^{\ppi}\right)\overset{d}{\to}N(\bm 0,\mat{J}_\lambda\widetilde{\mat{\Sigma}}\mat{J}_\lambda^\top),$$ where $\widetilde{\mat{\Sigma}}$ represents the probability limit of $\frac{1}{T-M+1}\sum_{t=M}^T \mat{\Sigma}_t$ as $T\to \infty,$ 
where $\mat{\Sigma} = \widetilde{\mat{\Sigma}}-\mat{U}^\top \mat{\Sigma}_{ps}^{-1}\mat{U}$, with $\widetilde{\mat{\Sigma}}$ defined in \Cref{thm:addIPW_true}, $\mat{\Sigma}_{ps}$ and $\mat{U}$ defined in \Cref{assump: ps_model}, and $\mat{J}_\lambda = \begin{pmatrix}
    1  &-\lambda V^{\ppi^{(M)}}-(1-\lambda)\overline{Y}
\end{pmatrix}.$
\end{theorem}

\begin{proof}[Proof of \cref{prop:addIPW_est}]
    The estimator $\tildeVPS$ based on estimated propensity can be constructed using the estimating equation
    $$s(\lHt[-1],\bm W_t,\bm Y_t;\mmu)=\begin{pmatrix}
    \hat V_{t}^{(\bbeta)}(\ppi^{(M)};\ttheta)-V_t^{\ppi^{(M)}}-\mu_1\\ 
    \rho_t^{(\bbeta)}(\ppi^{(M)};\ttheta)-1-\mu_2 \\
    \psi(\bm W_t,\lHt[-1];\ttheta)
\end{pmatrix}$$ where $\mu_1 = \mu_2=0,$ and $\mmu=(\mu_1,\mu_2,\ttheta^\top)^\top$, and the delta method.
Let $s(\lHt[-1],\bm W_t,\bm Y_t;\mmu)$ denote the first two entries of $s(\lHt[-1],\bm W_t,\bm Y_t;\mmu),$ and $\hat\mmu_T$ be the solution to
\[
\shiftmean s(\lHt[-1],\bm W_t,\bm Y_t;\mmu)=0.
\]
By the proof of Theorem A.2 in \cite{zhou2024estimating}, as $T\to\infty,$ \[
\sqrt{T-M+1}(\hat\mmu_T -\mmu^\ast)\dto N(0,\bm\Sigma_{\mmu}),
\]
where $\bm\Sigma_{\mmu} = \mat{A}^{-1}\mat{B}(\mat{A}^{-\top})$ with
\begin{equation}\label{a:eq4}
   \mat{A} = \begin{bmatrix}
 \mat{I}_2 &\mat{U}^\top\\
 \bm 0 &\mat{\Sigma}_{ps}
 \end{bmatrix}\hspace{4mm}\mathrm{and}\hspace{4mm}\mat{B} = 
 \begin{bmatrix}
     \widetilde{\mat{\Sigma}} &\mat{U}^\top\\
     \mat{U} & \mat{\Sigma}_{ps}
 \end{bmatrix}.  
\end{equation}
Note that $\hat\mmu_T = (\tildeVPS, \rhobar, \hat\ttheta^\top)^\top,$ we have
\[
\sqrt{T-M+1}\left((\tildeVPS, \rhobar)^\top-(V^{\ppi},1)^\top\right)\dto N(0, \bar{\bm \Sigma}),
\]
where $\bar{\bm\Sigma} = \widetilde{\bm\Sigma}-\bm U^\top \bm \Sigma_{ps}^{-1}\bm U,$ as $T\to\infty.$ 
Recall  \[\hat V(\ppi^{(M)};\lambda,\hat\ttheta) = \lambda\frac{\tildeVPS}{\rhobar}+(1-\lambda)(\hatVPS+(1-\rhobar)\overline{Y}).\]
By applying the delta method, we have
\[
\sqrt{T-M+1}(\hatVPS-V^{\ppi})\dto N(\bm J_\lambda\bar{\bm\Sigma}\bm J_\lambda^\top),\quad\text{as  }\;T\to\infty,
\]
where
$$
\mat{J}_\lambda= 
\begin{bmatrix}
    1 &-\lambda V^{\ppi^{(M)}}-(1-\lambda)\overline{Y}
\end{bmatrix}.
$$
\end{proof}

\begin{proof}[Proof of \cref{thm:addIPW_est}]
    According to \cref{prop:addIPW_est}, this result holds when $\hat\lambda = \lambda^\ast$. Then by similar arguments in \cref{thm:addIPW_true}, replacing $\lambda^\ast$ by a consistent estimator will not change the limiting distribution.  
\end{proof}

\begin{proof}[Proof of \cref{thm: Hajek efficiency}]
  According to Assumption~\ref{assump: ps_model}, $\mat{\Sigma}_{ps}$ is positive definite. Therefore, $\mat{U}^\top\mat{\Sigma}_{ps}^{-1}\mat{U}$ is positive semidefinite, which implies that the asymptotic variance of the proposed estimator based on the estimated propensity score is no greater than that based on the true propensity score.  
\end{proof}

\begin{proof}[Proof of \cref{prop: bound}]
   The proof follows directly from \cref{rmk: bound_trueps}, by replacing the true parameter $\ttheta^\ast$ in the propensity score model by the consistent estimator $\hat\ttheta.$
\end{proof}

\section{Standard IPW estimator}\label{a:sec:IPW estimator}
In this section, we introduce the standard IPW estimator for estimating policy values. We first present its form and state a proposition establishing the connection between additive IPW weights and standard IPW weights in \cref{a:subsec: connection}. We then provide the proof of the proposition in \cref{a:subsec:proof_prop}.

\subsection{Connection between addictive IPW weights and IPW weights}\label{a:subsec: connection}
Let $\ppi^{(M)} \in \Pi^M$ denote a policy over $M$ time periods. An alternative estimator for the policy value $V^{\ppi^{(M)}}$ is the standard IPW estimator. In our setting, the standard IPW estimator takes the following form:
\[
\hat V^{I}\left(\ppi^{(M)}\right) = \shiftmean \rho_t^{I}\left(\ppi^{(M)}\right)\,\overline{Y}_{t},
\]
where
\[
\rho_t^{I}\left(\ppi^{(M)}\right) = \prod_{m=1}^M \prod_{i=1}^n \rho_{t-m+1,i}\left(\ppi^{(M)}\right)
\]
is referred to as the IPW weight. By similar arguments in \cite{papadogeorgou2022causal}, the standard IPW estimator is asymptotically normal as $T\to\infty.$

\cref{prop1_multi} establishes the relationship between the additive IPW weights and the standard IPW weights, and characterizes their key properties.
% \begin{proposition}\label{prop1}
% For a policy $\ppi$ be a policy over a single time period. Under \Cref{assump: unconfoundedness,assump: factor_ps,assump: additive_model}, the following statements hold:
%     \begin{enumerate}[label=(\alph*)]
%         \item $\rhot[n][\ppi] = \prod_{i=1}^n \rho_{ti}(\ppi),$ 
%         \item  $\Var(\rhot[\beta][\ppi]\mid \lHt[-1]) \leq \Var(\rho_t^{(n)}(\ppi)\mid \lHt[-1])$ for all $\beta \in \{1,2,\dots,n\},$ 
%         \item $\E[\rhot[\beta][\ppi] Y_{ti}\mid\lHt[-1]] = Y_{ti}^{\ppi}\left(\lHt[-1]\right),$ and $\E[\rhot[\beta][\ppi]\mid\lHt[-1]]=1.$
%     \end{enumerate}
% \end{proposition}

\begin{proposition}\label{prop1_multi}
For a policy $\ppi^{(M)} = (\ppi_1,\dots,\ppi_M)$ over multiple time periods. Under \Cref{assump: unconfoundedness,assump: factor_ps,assump: additive_model_multi}, the following statements hold:
    \begin{enumerate}[label=(\alph*)]
        \item $\rhot[\bm n] = \prod_{m=1}^M\prod_{i=1}^n \rho_{t-m+1,i}(\ppi),$ 
        \item  $\Var(\rhot\mid \lHt[-M]) \leq \Var(\rho_t^{I}(\ppi)\mid \lHt[-M])$ for all $\bbeta \in \{1,2,\dots,n\}^M,$ 
        \item $\E[\rhot Y_{ti}\mid\lHt[-M]] = Y_{ti}^{\ppi},$ and $\E[\rhot\mid\lHt[-1]]=1.$
    \end{enumerate}
\end{proposition}

The additive IPW weights reduce to the standard IPW weights when $\bbeta = \bm n$ and yield a variance no greater than that of the standard IPW weights for all time $t$. 
% Based on \cref{prop1}(c), we can construct a consistent estimator of $V^{\ppi}$ using the additive IPW weights. 

Theoretically, the asymptotic variance of the estimator using standard IPW weights is not necessarily larger than that of a similar estimator based on additive IPW weights. For example, consider $\ppi\in\Pi,$ the asymptotic variance of $\hat V^{I}\left(\ppi\right)$ is $\frac{1}{T}\sum_{t=1}^T \Var\left(\rho_t^{(n)}(\ppi) \overline{Y}_t\mid \lHt[-1]\right).$ If we replace $\rho_t^{I}(\ppi)$ by $\rho_t^{(\beta)}$ for $\beta<n$, the variance of the resulting estimator is $$\frac{1}{T}\sum_{t=1}^T \Var\left(\rho_t^{(\beta)}(\ppi) \overline{Y}_t\mid \lHt[-1]\right).$$ Let $R_t(\ppi) = \rho_t^{(n)}(\ppi) - \rhot[\beta][\ppi]$ denote the higher-order terms in the additive IPW weights. The difference between the two asymptotic variance expressions can then be written as 
\[\frac{1}{T}\sum_{t=1}^T \Var(R_t(\ppi)\overline{Y}_t \mid \lHt[-1]) + 2\Cov(R_t(\ppi)\overline{Y}_t, \rhot[\beta][\ppi]\overline{Y}_t).\] 
This difference can be positive or negative depending on the covariance term. By the Cauchy–Schwarz inequality, the difference is positive if $\Var(R_t(\ppi)\overline{Y}_t \mid \lHt[-1]) \ge 4\Var(\rhot[\beta][\ppi]\overline{Y}_t \mid \lHt[-1])$, which is likely to hold when $\beta$ is small (e.g., $\beta = 1, 2$) and $n$ is large. Moreover, the standard IPW weights may degenerate to zero when $n$ is large. 
% and the standard H\'ajek estimator may not will defined in practice.

\subsection{Proof of \cref{prop1_multi}}\label{a:subsec:proof_prop}
\begin{proof}[Proof of \cref{prop1_multi}]
Let $\bbeta=(\beta_1,\dots,\beta_M)\in\{1,2,\dots,n\}^M,$ and $\ppi^{(M)} = (\ppi_1,\dots,\ppi_M)\in\Pi^M$. Assume \Cref{assump: unconfoundedness,assump: factor_ps} hold.
    Note that for $\ppi\in\Pi$
    \begin{align*}
        \rho_t^{I}(\ppi) &= \prod_{i=1}^n\rho_{ti}(\ppi)\\
        &=\prod_{i=1}^n\left(\left(1+\rho_{ti}(\ppi)-1\right)\right)\\
        &=1+\sum_{i=1}^n(\rho_{ti}(\ppi) -1) +\sum_{1\leq i_1<i_2\leq n}(\rho_{ti_1}(\ppi) -1)(\rho_{ti_2}(\ppi) -1)+\dots\\ &\quad+\sum_{1\le i_1<\dots<i_n\le n}(\rho_{ti_1-1}(\ppi)-1 )\dots(\rho_{ti_n}(\ppi) -1)\\
        &=\rho_t^{(n)}(\ppi)
    \end{align*}

    Similarly, $\rho_{t-m+1}^{I}(\ppi) = \rho_{t-m+1}^{(n)}(\ppi)$ for all $1\le m \le M.$ Thus,
    $$\rho_t^I(\ppi) = \prod_{m=1}^M\rho_{t-m+1}^I(\ppi_m) = \prod_{m=1}^M\rho_{t-m+1}^{(n)}(\ppi_m) = \rhot[\bm n].$$
    Next, we show that 
\[
\Var\!\big(\rhot\mid \lHt[-M]\big) \;\le\; \Var\!\big(\rho_t^{I}(\ppi)\mid \lHt[-M]\big)
\quad\text{for all }\bbeta \in \{1,2,\dots,n\}^M.
\]
We consider $M=1$. The general case $M>1$ follows by applying the law of total variance to reduce to the single-time policy case.

Define
\[
R_t(\ppi) := \rho_t^{I}(\ppi) - \rhot[\beta][\ppi].
\]
By construction, $R_t(\ppi)$ is a sum of terms of the form 
$(\rho_{ti_1}(\ppi)-1)\cdots(\rho_{ti_k}(\ppi)-1)$ with $k>\beta$. 
Note that $\rhot[\beta][\ppi]$ consists of products of the same form as $R_t(\ppi)$ but with strictly smaller length. By \cref{lemma1}, $\E[\rho_{ti}(\ppi)-1]=0$ for all $i$. Therefore every mixed product between $\rhot[\beta][\ppi]$ and $R_t(\ppi)$ has zero expectation, implying
\[
\E[\rhot[\beta][\ppi]R_t(\ppi)\mid\lHt[-1]]=0.
\]
Moreover, by \cref{lemma1},
\[
\E[R_t(\ppi)\mid\lHt[-1]] 
= \E[\rho_t^{(n)}(\ppi)\mid\lHt[-1]]-\E[\rhot[\beta][\ppi]\mid\lHt[-1]]
= 0-0=0.
\]
Thus,
\[
\Cov\!\big(\rhot[\beta][\ppi],R_t(\ppi)\mid\lHt[-1]\big) 
= \E[\rhot[\beta][\ppi]R_t(\ppi)\mid\lHt[-1]] - 0 = 0.
\]

Combining the pieces, we obtain
\begin{align*}
\Var(\rho_t^{I}(\ppi)\mid\lHt[-1])
&= \Var(\rhot[\beta][\ppi] + R_t(\ppi)\mid\lHt[-1]) \\
&= \Var(\rhot[\beta][\ppi]\mid\lHt[-1])
   + \Var(R_t(\ppi)\mid\lHt[-1])\\
&\ge \Var(\rhot[\beta][\ppi]\mid\lHt[-1]).
\end{align*}
Finally, by \cref{a: lemma_MDS,lemma1}, condition (c) holds.
\end{proof}

\section{Alternative definition of multi-period policy value}\label{a:sec:alternative_policy_value}

In this section, we introduce an alternative formulation of policy values over multiple time periods, denoted by $V^{\ppi_1,\dots,\ppi_M}_{\text{all}}$, which incorporates outcomes from all $M$ periods. The definition for a policy over a single time period remains the same. For a two-period policy $(\ppi_{1},\ppi_2)$, we define
\[
V^{\ppi_{1},\ppi_{2}}_{\text{all}} = V^{\ppi_{1},\ppi_{2}} + \frac{1}{(T-M)n}\sum_{t=M}^{T-1} \sum_{i=1}^n Y_{ti}^{\ppi_{2}}\left(\lHt[-1]\right)
\]
Here, the second term is the policy value of $\ppi_{2}$ defined in \cref{subsec: policyvalue} averaged over periods $M$ to $T-1$, rather than $M$ to $T$. By applying the same iterative procedure, one can define $V^{\ppi_1,\dots,\ppi_M}_{\text{all}}$. This formulation is essentially equivalent to the previous one, except that it sums policy values of different lengths. Consequently, the estimation and policy learning procedures can be easily adapted to this setting.

\section{Test validity}\label{a:sec:test_multi}

In this section,  we state the regularity conditions and present the Proof of \cref{thm:test_multi}

\subsection{Regularity assumption}\label{a:subsec: regularity_test}

In this section, we present the regularity conditions used for the asymptotic validity of the test. These conditions are analogous to those used for the asymptotic normality of the proposed estimators, but are stated jointly over the $L$ policies used in the test.

\begin{assumption}\label{assump: test_regularity}
Let $\ppi_1,\dots,\ppi_L$ be the $L$ selected policies in the hypothesis test. For each $1\le m\le M$, the following conditions hold.
\begin{enumerate}
    \item\label{assump: test_regularity1}
    There exists a finite positive semidefinite matrix $\widetilde{\mat{\Sigma}}{}^{H}$ such that
    \[
    \frac{1}{T-M+1}
    \sum_{t=M}^T
    \Var\left(\aat{}_{m} \mid \lHt[-M]^\ast\right)
    \pto
    \widetilde{\mat{\Sigma}}{}^{H},
    \]
    where
    \[
    \begin{aligned}
    \aat{}_{m}
    =
    \Big(
    &\widetilde V_t^{(\bbeta_{m1})}(\ppi_1^{(m)}),\dots,
    \widetilde V_t^{(\bbeta_{m1})}(\ppi_L^{(m)}),
    \widetilde V_t^{(\bbeta_{m2})}(\ppi_1^{(m)}),\dots,
    \widetilde V_t^{(\bbeta_{m2})}(\ppi_L^{(m)}), \\
    &\rho_t^{(\bbeta_{m1})}(\ppi_1^{(m)}),\dots,
    \rho_t^{(\bbeta_{m1})}(\ppi_L^{(m)}),
    \rho_t^{(\bbeta_{m2})}(\ppi_1^{(m)}),\dots,
    \rho_t^{(\bbeta_{m2})}(\ppi_L^{(m)})
    \Big)^\top .
    \end{aligned}
    \]

    \item\label{assump: test_regularity2}
    Let
    \[
    \begin{aligned}
    s(\lHt[-1],\bm W_t,\bm Y_t;\ttheta)
    =
    \Big(
    &\widetilde V_t^{(\bbeta_{m1})}(\ppi_1^{(m)})-V^{\ppi_1^{(m)}},\dots,
    \widetilde V_t^{(\bbeta_{m1})}(\ppi_L^{(m)})-V^{\ppi_L^{(m)}},\\
    &\widetilde V_t^{(\bbeta_{m2})}(\ppi_1^{(m)})-V^{\ppi_1^{(m)}},\dots,
    \widetilde V_t^{(\bbeta_{m2})}(\ppi_L^{(m)})-V^{\ppi_L^{(m)}},\\
    &\rho_t^{(\bbeta_{m1})}(\ppi_1^{(m)})-1,\dots,
    \rho_t^{(\bbeta_{m1})}(\ppi_L^{(m)})-1,\\
    &\rho_t^{(\bbeta_{m2})}(\ppi_1^{(m)})-1,\dots,
    \rho_t^{(\bbeta_{m2})}(\ppi_L^{(m)})-1
    \Big)^\top .
    \end{aligned}
    \]
    The following conditions hold.
    \begin{enumerate}[label=(\alph*)]
        \item\label{assump: IPWestadda_test}
        There exists a conformable matrix $\mat U^{H}$ such that
        \[
        \shiftmean
        \E_{\ttheta^\ast}
        \left[
        \psi(W_t,\lHt[-1];\ttheta^\ast)
        s(\lHt[-1],\bm W_t,\bm Y_t;\ttheta^\ast)^\top
        \mid \lHt[-M]^\ast
        \right]
        \pto
        \mat U^{H}.
        \]

        \item\label{assump: IPWestaddb_test}
        If
        \[
        P_{jt}
        =
        \frac{\partial}{\partial \ttheta_j}
        s(\lHt[-1],\bm W_t,\bm Y_t;\ttheta)
        \bigg|_{\ttheta^\ast},
        \]
        where $\ttheta_j$ is the $j$-th entry of $\ttheta$, then there exists $r_j\in(0,2]$ such that
        \[
        \shiftmean
        \frac{1}{t^{r_j}}
        \left|
        P_{jt}
        -
        \E_{\ttheta^\ast}
        \left[
        P_{jt}\mid\lHt[-M]^\ast
        \right]
        \right|
        \pto 0 .
        \]
    \end{enumerate}
\end{enumerate}
\end{assumption}

These conditions are similar to Assumption~\ref{assump:Policy_est1}.\ref{assump: 1b} and Assumption~\ref{assump: ps_model}. They are used to establish the joint asymptotic distribution of the value estimators evaluated at the $L$ selected policies. 

\subsection{Proof of \Cref{thm:test_multi}}\label{a:subsec:proof_validity}

\begin{proof}
We outline the proof for a fixed step $m$. The argument for the sequential procedure follows by applying the same argument at each step and using the Bonferroni correction.

let $\hat{\bm\Delta}{}_{mt}$ be defined as in \Cref{sec:interaction}.By the same martingale central limit theorem argument used in the proof of \Cref{thm:addIPW_est}, together with \Cref{assump: test_regularity}, we have
\[
\hat{\bm\Delta}_m = \frac{1}{\sqrt{T-m+1}}\sum_{t=m}^M\hat{\bm\Delta}_{mt}
\dto
N\left(\bm 0,\mat{\Sigma}{}_{m}^{H}\right),
\]
for some finite positive semidefinite matrix $\mat{\Sigma}{}_{m}^{H}$. Replacing the limiting mixing weights by their consistent estimators does not affect the limiting distribution by the same argument as in the proof of \Cref{thm:addIPW_est}.

Let
\[
\phi(\bm x)=\bm x^\top\bm x .
\]
The observed  test statistic can be written as
\[
T_m
=
\phi\left(\hat{\bm\Delta}_m\right).
\]
Therefore, by the continuous mapping theorem,
\[
T_m \dto \phi(\bm G_m),
\qquad
\bm G_m\sim N\left(\bm 0,\mat{\Sigma}{}_{m}^{H}\right).
\]

By the same argument as in \Cref{rmk: bound_trueps},
\[
\widehat{\mat{\Sigma}}{}_{m}^{H^\ast}
=
\frac{1}{T-M+1}
\sum_{t=M}^T
\hat{\bm\Delta}{}_{mt}\hat{\bm\Delta}{}_{mt}^{\top}
\pto
\mat{\Sigma}{}_{m}^{H^\ast},
\]
where $\mat{\Sigma}{}_{m}^{H^\ast}-\mat{\Sigma}{}_{m}^{H}$ is positive semidefinite. Thus, $\widehat{\mat{\Sigma}}{}_{m}^{H^\ast}$ consistently estimates a conservative upper bound for the asymptotic covariance of $\sqrt{T-M+1}\bm\Delta{}_{m}$.

For the multiplier bootstrap, let $\{Z_t^{(b)}\}_{t=M}^T$ be independent standard normal random variables. Define
\[
\bm\Delta{}_{m}^{(b)}
=
\frac{1}{\sqrt{T-M+1}}
\sum_{t=M}^T
Z_t^{(b)}\hat{\bm\Delta}_{mt}.
\]
Conditional on the observed data, $\bm\Delta{}_{m}^{(b)}$ is mean-zero Gaussian with covariance $\widehat{\mat{\Sigma}}{}_{m}^{H^\ast}$. Hence, using the convergence of $\widehat{\mat{\Sigma}}{}_{m}^{H^\ast}$,
\[
\bm\Delta{}_{m}^{(b)}
\dto
\bm G_m^\ast,
\qquad
\bm G_m^\ast\sim N\left(\bm 0,\mat{\Sigma}{}_{m}^{H^\ast}\right),
\]
in the bootstrap distribution.

The bootstrap statistic is
\[
T_m^{(b)}
=
\phi\left(\bm\Delta{}_{m}^{(b)}\right).
\]
Let $c_{m,1-\alpha/M}^{\ast}$ denote the $(1-\alpha/M)$ quantile of the distribution of $\phi(\bm G_m^\ast)$. Since
\[
\mat{\Sigma}{}_{m}^{H^\ast}-\mat{\Sigma}{}_{m}^{H}
\]
is positive semidefinite, $\bm G_m^\ast$ is at least as variable as $\bm G_m$ in the positive semidefinite order. By Anderson's inequality for centered Gaussian distributions and Euclidean balls,
\[
\P\left\{\phi(\bm G_m)>c_{m,1-\alpha/M}^{\ast}\right\}
\le
\P\left\{\phi(\bm G_m^\ast)>c_{m,1-\alpha/M}^{\ast}\right\}
\le
\frac{\alpha}{M}.
\]
It follows that
\[
\limsup_{T\to\infty}
\P\left\{
T_m>c_{m,1-\alpha/M}^{\ast}
\right\}
\le
\frac{\alpha}{M}.
\]

Equivalently, if the bootstrap $p$-value is defined as the proportion of bootstrap statistics $T_m^{(b)}$ that exceed the observed statistic $T_m$, then
\[
\limsup_{T\to\infty}
\P\left(
\text{$p$-value}<\frac{\alpha}{M}
\right)
\le
\frac{\alpha}{M}.
\]
Finally, applying the union bound over the $M$ sequential tests gives
\[
\limsup_{T\to\infty}
\P\left(
\text{at least one true null hypothesis is rejected}
\right)
\le
\alpha .
\]
This proves the desired result.
\end{proof}

\section{Stochastic policy approximation}\label{a: Stochastic_policy}

In this section, we introduce the stochastic policy approximation used in our policy learning procedure and discuss its theoretical and practical properties. In \cref{a:subsec:stochastic_policy_value}, we define the policy value under stochastic policies. In \cref{a:subsec:stochastic_proof}, we present the proof for \cref{thm: approx_optimize}. In \cref{a:subsec:deterministic_penalty}, we discuss how to encourage the learned stochastic policy to be close to a deterministic policy by adding a penalization term to the objective function.

\subsection{Policy value under stochastic policies}\label{a:subsec:stochastic_policy_value}

A stochastic policy $\ppi^{\text{stoch}}$ over one time period is a function that maps
$\lHt[-1]$ to $[0,1]^n$.
The quantity $\pi_{i}^{\text{stoch}}\left(\lHt[-1]\right)$ represents the probability that
unit $i$ receives treatment at time $t$ under policy $\ppi^{\text{stoch}}$. We focus on stochastic policies that assign treatments independently across units given the observed history, as this formulation % matches the individualized deterministic policy class and 
is sufficient for approximating deterministic policies arbitrarily well.
For a stochastic policy $\ppi^\text{stoch}$ over a single time period, the policy value at time $t$ for unit $i$ is defined as
\begin{equation}\label{eq: policyvalue_stoch}
  \begin{aligned}
 Y_{ti}^{\ppi^{\text{stoch}}}\left(\lHt[-1]\right) & = \E_{\bm W_t \sim \ppi^\text{stoch}}[Y_{ti}(\lWt[-1],\bm W_t)\mid \lHt[-1]]\\
 & = \sum_{\bm w_t\in \{0,1\}^n}  \prod_{j=1}^n \pi_{j}^{\text{stoch}}\left(\lHt[-1]\right)^{ w_{tj}}\left(1-\pi_{j}^{\text{stoch}}\left(\lHt[-1]\right)\right)^{1-w_{tj}}Y_{ti}(\lWt[-1],\bm w_{t}).
\end{aligned}  
\end{equation}
Therefore, the quantity $Y_{ti}^{\ppi^{\text{stoch}}}\left(\lHt[-1]\right)$ represents the expected outcome
for unit $i$ at time $t$ when the treatment at time $t$ is assigned according to
$\ppi_{i}^{\text{stoch}}\left(\lHt[-1]\right)$ and previous treatments are held fixed at their observed values.

A deterministic policy can be viewed as a special case of a stochastic policy that assigns
probability one to a single treatment vector. Moreover, the policy value under a deterministic policy in \cref{eq: policy_value_single} is a special case of the policy value under a stochastic policy as defined in 
\cref{eq: policyvalue_stoch}.

A stochastic policy over $M$ time periods is defined as a sequence of policies
$\ppi^{\text{stoch,(M)}} = (\ppi^{\text{stoch}}_1,\dots,\ppi^{\text{stoch}}_M)$.
Similar to \cref{subsec: policyvalue}, the policy value of $\ppi^{\text{stoch,(M)}}$
is defined recursively in a backward manner.
We first obtain $Y_{ti}^{\ppi^{\text{stoch}}_1}\left(\lHt[-1]\right)$ according to \cref{eq:policy_value}.
We then define

\begin{equation}
\begin{aligned}
&Y_{ti}^{\ppi^{\text{stoch}}_1,\ppi^{\text{stoch}}_2}\left(\lHt[-2]\right)\\
=& \E_{\bm W_{t-1} \sim \ppi_2^\text{stoch}} \left[
Y_{ti}^{\ppi^{\text{stoch}}_1}\left(\lHt[-1]\right)
\mid \lHt[-2]\right]\\
=& \sum_{\bm w_{t-1}\in \{0,1\}^n}
\prod_{j=1}^n
\pi_{2j}^{\text{stoch}}\left(\lHt[-2]\right)^{w_{t-1,j}}
\left(1-\pi_{2j}^{\text{stoch}}\left(\lHt[-2]\right)\right)^{1-w_{t-1,j}}
Y_{ti}^{\ppi^{\text{stoch}}_1}\left(\lHt[-1]^{\bm w_{t-1}}\right),
\end{aligned}
\end{equation}
where $\overline{H}_t^{\bm w_t} = \{\overline{H}_{t-1},\bm w_{t},\overline{\bm Y}_t(\overline{\bm W}_{t-1},\bm w_t),\overline{\bm X}_{t+1}(\overline{\bm W}_{t-1},\bm w_t)\}$
denotes the history that would have been observed up to time $t$ under the scenario, in which the treatment assignment at $t$ is $\bm w_t$
The quantity $Y_{ti}^{\ppi^{\text{stoch}}_1,\ppi^{\text{stoch}}_2}\left(\lHt[-2]\right)$ represents the expected
outcome for unit $i$ at time $t$ when the treatment at time $t-1$ follows
$\ppi^{\text{stoch}}_2$, the treatment assignment at time $t$ follows
$\ppi^{\text{stoch}}_1$, and the treatment history before time $t-1$
is held fixed at its observed value.
The value of $Y_{ti}^{\ppi^{\text{stoch,(M)}}}\left(\lHt[-M]\right)$ is defined continuing this recursive process.
Finally, the overall policy value is defined by averaging over all time periods and units:
\[
V^{\ppi^{\text{stoch,(M)}}}
= \frac{1}{(T-M+1)n}\sum_{t=M}^T\sum_{i=1}^n Y_{ti}^{\ppi^{\text{stoch,(M)}}}\left(\lHt[-M]\right).
\]

We estimate
the overall policy value \(V^{\ppi^{\text{stoch,(M)}}}\)  using estimators
analogous to those introduced in \cref{sec:estimation}, by replacing the unit-specific IPW weights for the deterministic policy weights defined in \cref{subsec:single}
with
\[
\rho_{ti}(\ppi_{m};\ttheta) = \dfrac{\pi_{mi}^{\text{stoch}}\left(\lHt[-1]\right))^{W_{ti}}(1-\pi_{mi}^{\text{stoch}}\left(\lHt[-1]\right)))^{1-W_{ti}}}{e_{ti}(W_{ti};\ttheta)}.
\]

\subsection{Proof of \cref{thm: approx_optimize}}\label{a:subsec:stochastic_proof}
 In this section, we present the Proof of \cref{thm: approx_optimize}. 

\begin{proof}[Proof of \cref{thm: approx_optimize}] 
Let $M\ge1 $ and $\bbeta\in \{1,\dots,n\}^M.$  Let $\hat\ppi^{(M)} = (\hat\ppi_1,\dots,\hat\ppi_M)$ and $\hat\ppi^{\text{stoch},(M)}(k)= (\hat\ppi^{\text{stoch}}_{1}(k),\dots,\hat\ppi^{\text{stoch}}_{M}(k))$ be the optimizers $\hatVPS$ based the deterministic policy class $\Pi^M$ and stochastic policy $\Pi^{\text{stoch},M}_k$ respectively. Note that $\expit(k\bm X_{ti}^P\ggamma)\to\mathds{1}\{\bm X_{ti}^P\ggamma>0\}$ 
as $k\to\infty$ for all $\ggamma$ with $X_{ti}^P\ggamma\neq 0.$  Moreover, $Y_{ti}$ and $\bm X_{ti}^P$ are bounded. Therefore, for any $\epsilon>0,$ there exist $K\ge 1$ and some $\tilde\ppi^{\text{stoch},(M)}(K)\in\Pi^{\text{stoch},M}(K)$ such that $|\tilde\pi_{Kmi}(\lHt[-1])-\hat\pi_{mi}(\lHt[-1])|<\epsilon$ for all $M\leq t\le T, 1\leq i\leq n$ and $1\le m\leq M.$ 
Here we implicitly assume that $\bm X_{ti}^P\ggamma\neq 0$ for all $t$ and $i$. If not, one may perturb $\ggamma$ by an arbitrarily small amount so that this condition hold.
Let $\epsilon_1>0.$ Then we can choose $K$ sufficiently large such that 
$$\hatVPS[][\bbeta][\hat\ppi^{\text{stoch,(M)}}(K)]\ge\hatVPS[][\bbeta][\tilde\ppi^{\text{stoch,(M)}}(K)] > \hatVPS[][\bbeta][\hat\ppi^{(M)}]-\epsilon_1.$$ According to \cref{thm:addIPW_est}, $\hatVPS[][\bbeta][\hat\ppi^{(M)}]-V^{\hat\ppi^{(M)}}\pto 0$ as $T\to\infty.$ Moreover, by similar arguments in the proof of \cref{thm:addIPW_est}, we can show  $\hatVPS[][\bbeta][\hat\ppi^{\text{stoch,(M)}}(K)]-V^{\hat\ppi^{\text{stoch,(M)}}(K)}\pto 0$ as $T\to\infty.$ 
Then for any $\epsilon_2>0$, we have
\[
\P\!\left(
\hatVPS[][\bbeta][\hat\ppi^{(M)}]
>
V^{\hat\ppi^{(M)}}-\epsilon_2
\right)\to 1
\]
and
\[
\P\!\left(
V^{\hat\ppi^{\text{stoch},(M)}(K)}
>
\hatVPS[][\bbeta][\hat\ppi^{\text{stoch},(M)}(K)]-\epsilon_2
\right)\to 1
\]
as $T\to\infty$. Moreover, by the argument above, we have
\[
\P\!\left(
\hatVPS[][\bbeta][\hat\ppi^{\text{stoch},(M)}(K)]
\ge
\hatVPS[][\bbeta][\hat\ppi^{(M)}]-\epsilon_1
\right)\to 1 .
\]
On the intersection of the above three events,
\begin{align*}
V^{\hat\ppi^{\text{stoch},(M)}(K)}
&>
\hatVPS[][\bbeta][\hat\ppi^{\text{stoch},(M)}(K)]-\epsilon_2 \\
&\ge
\hatVPS[][\bbeta][\hat\ppi^{(M)}]-\epsilon_1-\epsilon_2 \\
&>
V^{\hat\ppi^{(M)}}-\epsilon_1-2\epsilon_2 .
\end{align*}
Therefore,
\begin{align*}
&\P\!\left(
V^{\hat\ppi^{\text{stoch},(M)}(K)}
>
V^{\hat\ppi^{(M)}}-\epsilon_1-2\epsilon_2
\right) \\
\ge\;&
1-\P\!\left(
V^{\hat\ppi^{\text{stoch},(M)}(K)}
\le
\hatVPS[][\bbeta][\hat\ppi^{\text{stoch},(M)}(K)]-\epsilon_2
\right) \\
&-
\P\!\left(
\hatVPS[][\bbeta][\hat\ppi^{\text{stoch},(M)}(K)]
<
\hatVPS[][\bbeta][\hat\ppi^{(M)}]-\epsilon_1
\right) \\
&-
\P\!\left(
\hatVPS[][\bbeta][\hat\ppi^{(M)}]
\le
V^{\hat\ppi^{(M)}}-\epsilon_2
\right)
\to 1,
\end{align*}
as $T\to\infty$. Since $\epsilon_1$ and $\epsilon_2$ are arbitrary,for all $\epsilon>0,$ we have $$\P\!\left(
V^{\hat\ppi^{\text{stoch},(M)}(K)}
>
V^{\hat\ppi^{(M)}}-\epsilon
\right)\to 1 $$ as $T\to\infty.$
\end{proof}

\subsection{Penalization for deterministic policy learning}
\label{a:subsec:deterministic_penalty}

Finding the optimal policy over the stochastic policy class $\Pi^{\text{stoch}}(k)$ may result in an estimated optimal policy that is not close to a deterministic policy, with treatment assignment probabilities that are not close to $0$ or $1$.
The reason is that, under stochastic policies, the estimated policy value is generally not multilinear in the treatment probabilities $\pi_{i}^{\text{stoch}}(k)$.
As a result, even when $k$ is large, the maximizer of the smoothed objective may occur at a stochastic policy for which some assignment probabilities remain away from $0$ and $1$.

If a deterministic policy is required, we propose an adaptation to the optimization objective function that will allow estimated optimal policies to be closer to a deterministic one. Specifically, we propose modifying the algorithm to penalize treatment probabilities that are away from $0$ or $1$ in the objective of the optimization problem. Specifically,
for a single-period stochastic policy $\ppi^{\text{stoch}}(k)\in \Pi^{\text{stoch}}(k)$, we define the penalty 
\[
D_T(\ppi^{\text{stoch}}(k))
=
\frac{1}{nT}\sum_{t=1}^T\sum_{i=1}^n
\pi_{i}^{\text{stoch}}(k; \lHt[-1])
\left\{1-\pi_{i}^{\text{stoch}}(k; \lHt[-1])\right\}.
\]
The term
$\pi_{i}^{\text{stoch}}(k; \lHt[-1])\{1-\pi_{i}^{\text{stoch}}(k; \lHt[-1])\}$
is equal to $0$ when the treatment assignment probability is $0$ or $1$, and is positive otherwise, with larger values for probabilities closer to $0.5$.
To encourage estimated optimal policies that are nearly deterministic, we suggest maximizing the penalized objective
\[
\hat\ppi^{\text{pen}}(k)
=
\argmax_{\ppi^{\text{stoch}}(k)\in \Pi^{\text{stoch}}(k)}
\left\{
\hat V^{\ppi^{\text{stoch}}(k)}
-
c_k\,D_T(\ppi^{\text{stoch}}(k))
\right\},
\]
where $c_k$ may depend on $k$, as we discuss below and in \cref{thm: penalized_stochastic}.

For finding the optimal multi-period policy, we follow the sequential optimization procedure described in \cref{subsec: find_optimal_policy}. At each step $m$, and given the previously obtained policies $(\hat\ppi_1,\dots,\hat\ppi_{m-1})$, we optimize a single-period policy based on the penalized objective function. Therefore, it suffices to discuss the impact of penalization for the single-period optimal policy case.

As $k$ increases, the stochastic policy class expands to include more policies that resemble deterministic policies, with treatment assignment probabilities arbitrarily close to $0$ or $1$. Therefore, including the penalty term $D_T(\ppi^{\text{stoch}}(k))$ favors such policies. The coefficient $c_k$ controls the trade-off between maximizing the estimated policy value and encouraging near-deterministic assignments. If $c_k$ grows too quickly, the optimizationplaces a lot of weight on minimizing $D_T(\ppi^{\text{stoch}}(k))$, potentially at the expense of maximizing the estimated policy value. Therefore, the coefficient $c_k$ should increase with $k$ to encourage treatment assignment probabilities close to $0$ or $1$, but not so quickly that policies with high estimated value are overly penalized. 

In particular, let $\tilde\ppi^{\text{stoch}}(k)$ denote the sequence of stochastic policies constructed in the proof of \cref{thm: approx_optimize} that converges to the learned optimal deterministic policy as $k\to\infty$. We require
\begin{equation}\label{eq:constrain}
 c_k D_T(\tilde\ppi^{\text{stoch}}(k)) \to 0
\qquad \text{as } k\to\infty.   
\end{equation}
This condition ensures that the penalty becomes negligible for policies that closely approximate deterministic rules, while policies whose assignment probabilities remain away from $0$ and $1$ receive an increasingly large penalty as $c_k$ increases.

The next result shows that this penalization preserves the approximation property in \cref{thm: approx_optimize}.  It also shows that the learned stochastic policy is close to deterministic when $k$ is large.

\begin{theorem}
\label{thm: penalized_stochastic}
Suppose that \Cref{assump: unconfoundedness,assump: additive_model_multi,assump: factor_ps,assump:Policy_est1,assump: ps_model} hold.
Assume that \cref{eq:constrain} holds.
Then for any $\epsilon>0$ and $\delta>0$, there exists $K\ge 1$ such that for all $k\ge K$,
\[
\P\!\left(
V^{\hat\ppi^{\text{pen}}(k)}
>
V^{\hat\ppi}-\epsilon
\right)\to 1,
\qquad \text{as } T\to\infty,
\]
and
\[
D_T(\hat\ppi^{\text{pen}}(k))<\delta.
\]
\end{theorem}
In particular, when $k$ is large, the learned stochastic policy is close to deterministic in the sense that its treatment assignment probabilities are close to $0$ or $1$ on average. In practice, one may take $c_k=\sqrt{k}$, which diverges to infinity while increasing slowly with $k$, and is satisfying the conditions of the theorem.

\begin{proof}[Proof of \cref{thm: penalized_stochastic}]
Fix $\epsilon_1>0$. By the argument in the proof of \cref{thm: approx_optimize}, there exists $K_0\ge 1$ such that for all $k\ge K_0$, one can find some $\tilde\ppi^{\text{stoch}}(k)\in\Pi^{\text{stoch}}(k)$ satisfying
\[
\hatVPS[][\bbeta][\tilde\ppi^{\text{stoch}}(k)]
>
\hatVPS[][\bbeta][\hat\ppi]-\epsilon_1.
\]
By assumption, there exists $K_1\ge K_0$ such that for all $k\ge K_1$,
\[
c_k D_T(\tilde\ppi^{\text{stoch}}(k))<\epsilon_1.
\]
Hence, for all $k\ge K_1$,
\[
\hatVPS[][\bbeta][\tilde\ppi^{\text{stoch}}(k)]
-
c_k D_T(\tilde\ppi^{\text{stoch}}(k))
>
\hatVPS[][\bbeta][\hat\ppi]-2\epsilon_1.
\]
By optimality of $\hat\ppi^{\text{pen}}(k)$,
\[
\hatVPS[][\bbeta][\hat\ppi^{\text{pen}}(k)]
-
c_k D_T(\hat\ppi^{\text{pen}}(k))
\ge
\hatVPS[][\bbeta][\tilde\ppi^{\text{stoch}}(k)]
-
c_k D_T(\tilde\ppi^{\text{stoch}}(k)).
\]
Combining the above inequalities yields
\[
\hatVPS[][\bbeta][\hat\ppi^{\text{pen}}(k)]
-
c_k D_T(\hat\ppi^{\text{pen}}(k))
>
\hatVPS[][\bbeta][\hat\ppi]-2\epsilon_1.
\]
Since $c_k D_T(\hat\ppi^{\text{pen}}(k))\ge 0$, it follows that
\[
\hatVPS[][\bbeta][\hat\ppi^{\text{pen}}(k)]
>
\hatVPS[][\bbeta][\hat\ppi]-2\epsilon_1.
\]
Therefore, the same argument as in the proof of \cref{thm: approx_optimize}, with $\hat\ppi^{\text{pen}}(k)$ in place of $\hat\ppi^{\text{stoch}}(k)$, yields
\[
\P\!\left(
V^{\hat\ppi^{\text{pen}}(k)}
>
V^{\hat\ppi}-\epsilon
\right)\to 1
\]
for any $\epsilon>0$ and any fixed $k\ge K_1$.

It remains to show that the learned policy is close to deterministic. Since the outcomes and covariates are bounded, there exists a constant $B<\infty$ such that
\[
\left|
\hatVPS[][\bbeta][\ppi]
\right|
\le B
\qquad\text{for all } \ppi\in \Pi^{\text{stoch}}(k),
\]
uniformly in $k$ and $T$. Since
\[
\hatVPS[][\bbeta][\tilde\ppi^{\text{stoch}}(k)]
>
\hatVPS[][\bbeta][\hat\ppi]-\epsilon_1,
\]
and $\hatVPS[][\bbeta][\hat\ppi]\ge -B$, we have
\[
\hatVPS[][\bbeta][\tilde\ppi^{\text{stoch}}(k)]
>
-B-\epsilon_1.
\]
Thus, for all $k\ge K_1$,
\[
\hatVPS[][\bbeta][\tilde\ppi^{\text{stoch}}(k)]
-
c_k D_T(\tilde\ppi^{\text{stoch}}(k))
>
-B-2\epsilon_1.
\]
By optimality of $\hat\ppi^{\text{pen}}(k)$,
\[
\hatVPS[][\bbeta][\hat\ppi^{\text{pen}}(k)]
-
c_k D_T(\hat\ppi^{\text{pen}}(k))
>
-B-2\epsilon_1.
\]
On the other hand,
\[
\hatVPS[][\bbeta][\hat\ppi^{\text{pen}}(k)]
-
c_k D_T(\hat\ppi^{\text{pen}}(k))
\le
B-c_k D_T(\hat\ppi^{\text{pen}}(k)).
\]
Therefore,
\[
B-c_k D_T(\hat\ppi^{\text{pen}}(k))
>
-B-2\epsilon_1,
\]
which implies
\[
D_T(\hat\ppi^{\text{pen}}(k))
<
\frac{2B+2\epsilon_1}{c_k}.
\]
Since \(c_k\to\infty\), the right hand side converges to zero as
\(k\to\infty\).
Now fix any $\delta>0$. Since $c_k\to\infty$, there exists $K_2\ge K_1$ such that for all $k\ge K_2$,
\[
\frac{2B+2\epsilon_1}{c_k}<\delta.
\]
Hence, for all $k\ge K_2$,
\[
D_T(\hat\ppi^{\text{pen}}(k))<\delta.
\]
This proves the result.
\end{proof}

\section{Proofs for regret bound}\label{a: regret-bound-proof}

In this section we provide a complete proof of Theorem~\ref{thm:regret} and \cref{thm: regret_estimate_ps}.
In \cref{a:subsec:regret-regularity} we state and motivate
a regularity assumption on the capacity of the policy coordinate class and then 
give the proofs of the regret bound based on true and estimated propensity in \cref{a:subsec:proof-regret} and  \cref{a:subsec:proof-regret-estimateps} respectively.

\subsection{Regularity for the policy coordinate class}\label{a:subsec:regret-regularity}

Both the regularity assumption and the main proof involves the sequential
Rademacher complexity. We include the definition introduced in \cite{rakhlin2015sequential} for completeness.

\begin{definition}[Sequential Rademacher complexity]\label{def:src}
Let $\mathcal F\subseteq \mathbb R^{\mathcal Z}$ be a class of real-valued
functions on a domain $\mathcal Z$. A \emph{$\mathcal Z$-valued tree} of
depth $T$ is a sequence of maps $z_t:\{\pm 1\}^{t-1}\to \mathcal Z$,
$t=1,\dots,T$. Let $\epsilon_1,\dots,\epsilon_T$ be i.i.d. Rademacher signs.
The sequential Rademacher complexity of $\mathcal F$ at horizon $T$ is
\[
\mathfrak R_T^{\mathrm{seq}}(\mathcal F)
:= \sup_{(z_1,\dots,z_T)}\ \E_{\epsilon_{1:T}}
\Bigg[\ \sup_{f\in\mathcal F}\ \frac{1}{T}\sum_{t=1}^T
\epsilon_t f\big(z_t(\epsilon_{1:(t-1)})\big)\ \Bigg].
\]
\end{definition}

Sequential Rademacher complexity measures the richness of a function class in a sequential setting, where inputs may depend on past random quantities. 
It generalizes the classical Rademacher complexity to sequential settings by allowing the complexity measure to account for such dependence. In our context, it quantifies the effective complexity of the policy class when each treatment decision may depend on the previously observed history.

\begin{assumption}[Complexity of the policy coordinate class]\label{assump:coord-capacity}
Let $\Pi$ be a class of policies, and let
\[
\mathcal G=\big\{\overline{H}_{t-1}\to \pi_{i}(\overline{H}_{t-1}):\ \ppi\in\Pi,\ i=1,\ldots,n\big\}
\]
be the coordinate class. There exist constants $d<\infty$ such that, for all $T\ge 1$, its sequential
Rademacher complexity satisfies
\[
\mathfrak R_T^{\mathrm{seq}}(\mathcal G)\ \lesssim\ \sqrt{\frac{d}{T}}.
\], where $a \lesssim b$ denotes $a \le Cb$ for some universal constant $C>0$.
\end{assumption}

\Cref{assump:coord-capacity} will be used in our derivation of the regret bound to control the complexity of the policy class over time.  A simple sufficient case is when the coordinate class $\mathcal G$ is finite. This occurs, for example, if the relevant history variables take values in a finite set, so that there are only finitely many distinct mappings from the history to the treatment decision. In this case, standard bounds for sequential Rademacher complexity imply that \Cref{assump:coord-capacity} holds with $d$ of order $\log |\mathcal G|$ \citep{rakhlin2015sequential}.

When covariates are continuous, the condition depends on how the policy class is defined. For stochastic policies such as those introduced in \cref{a: Stochastic_policy}, where the policy is a smooth function of the history, small changes in the input typically lead to small changes in the output. Under additional regularity conditions (e.g., bounded covariates and Lipschitz continuity of the policy class), this smoothness can help control the sequential Rademacher complexity.

In contrast, for deterministic policies, the mapping from history to treatment is discontinuous, so small changes in the policy parameter or the history can lead to abrupt changes in the treatment decision when the history is close to the decision boundary. As a result, bounded covariates alone are generally not sufficient to guarantee \Cref{assump:coord-capacity}.

One way to relax this issue is to restrict attention to a smaller subclass of policies that avoids such boundary behavior. For example, for the linear policy class defined in \cref{eq: linearclass}, instead of allowing all $\ggamma \in \mathbb{R}^p$, one may consider a restricted class satisfying
\[
|\bm X_{ti}^P \ggamma| > \delta
\qquad \text{for all } t=1,\dots,T,\ i=1,\dots,n,
\]
for some constant $\delta > 0$. In this subclass defined by the constraint above, nearby parameter values are more likely to produce the same treatment decision, which can substantially reduce the effective complexity relative to the complete class of deterministic policies.

\subsection{Regret bound based on true propensity score}\label{a:subsec:proof-regret}

\begin{theorem}[Regret for the normalized estimator with true propensity score]\label{thm:regret}
Suppose \Cref{assump: unconfoundedness,assump: additive_model_multi,assump: factor_ps,assump:Policy_est1,assump:coord-capacity} hold. Define
\[
C_\beta(n,\eta)=\sum_{k=0}^{\beta}\binom{n}{k}\Big(\tfrac1\eta-1\Big)^{k}.\]
Then for any $\delta\in(0,1)$, with probability at least $1-\delta$,
\[
R(\hat\ppi^{(M)})\  \lesssim\ \frac{\delta_Y\prod_{m=1}^MC_{\beta_m}(n,\eta)\left(1+2\sqrt{d}+\sqrt{2\log(1/\delta)}\right)}{\sqrt{T}}.
\]
for $\delta_Y$ in \Cref{assump:Policy_est1}(a), $\eta$ in \Cref{assump: factor_ps}, and $d$ in \Cref{assump:coord-capacity}. 
\end{theorem}

\cref{a:subsec:proof-regret} establishes a finite-sample regret bound when the true propensity score is known.
 $\delta_Y$ scales with the outcome values, $\eta$ prevents instability from extreme propensity scores, $\prod_{m=1}^MC_{\beta_m}(n,\eta)$ grows with the interaction order and inflates variability, and $d$ measures the complexity of the policy class. These quantities together determine how quickly the learned policy approaches the performance of the optimal policy in finite samples.

\begin{proof}[Proof of Theorem~\ref{thm:regret}]
Let $\ppi^{(M)} = (\ppi_1,\ppi_2,\dots,\ppi_M)\in\Pi^M$ and $\bbeta = (\beta_1,\dots,\beta_M)\in\{1,\dots,n\}^M.$
For notational simplicity, write
Recall that the additive IPW estimator without normalization is
\[
N(\ppi)
\;=\;
\frac1T\sum_{t=1}^T \rhot\,\overline Y_t,
\qquad
D(\ppi)
\;=\;
\frac1T\sum_{t=1}^T \rhot,
\]
and the mixing estimator based on the true propensity scores can be written as
\[
\hatV
=
N(\ppi) + E(\ppi),
\]
where, with $\overline Y = T^{-1}\sum_{t=1}^T\overline Y_t$ and
$\hat\lambda$ is a consistent estimator for $\lambda^\ast$ and ,
\[
E(\ppi)
=
\Big\{\hat \lambda\frac{N(\ppi)}{D(\ppi)}
      + \big(1-\hat \lambda\big)\,\overline Y\Big\}
\big(1-D(\ppi)\big).
\]

By \cref{a: lemma_MDS} (martingale difference property),
\[
\E\big[\rhot\,\overline Y_t\mid \lHt[-1]\big]=V_t^{\ppi},
\qquad
\E\big[\rhot\mid \lHt[-1]\big]=1.
\]
Define
\[
f_t(\ppi):=\rhot\,\overline Y_t - V_t^{\ppi},
\qquad
g_t(\ppi):=\rhot-1.
\]
Then $\{f_t(\ppi)\}$ and $\{g_t(\ppi)\}$ are martingale difference sequences
with respect to the filtration generated by the history $\lHt[-M]$.
Moreover,
\[
N(\ppi)-V^{\ppi}
= \frac1T\sum_{t=1}^T f_t(\ppi),
\qquad
D(\ppi)-1
= \frac1T\sum_{t=1}^T g_t(\ppi),
\]
where $V^{\ppi}=\frac1T\sum_{t=1}^T V_t^{\ppi}$ is the target policy value.

By \Cref{assump: factor_ps} and \Cref{assump:Policy_est1}(a),
\[
|g_t(\ppi)|\ \le\ \prod_{m=1}^M C_{\beta_m}(n,\eta)-1,
\qquad
|\rhot\,\overline Y_t|\ \le\ \delta_Y\prod_{m=1}^MC_{\beta_m}(n,\eta),
\]
where
\[
C_\beta(n,\eta)
=
\sum_{k=0}^{\beta}\binom{n}{k}\Big(\tfrac1\eta-1\Big)^{k}.
\]

\paragraph{Uniform control of $D(\ppi)$.}
Define the function class
\[
\mathcal F_D
=
\Big\{\big(\lHt[-1],\bm W_t\big)\to \rhot : \ppi\in\Pi^M\Big\}.
\]
By \citet[Theorem~2]{rakhlin2015sequential},
\[
\E\Big[\sup_{\ppi\in\Pi^M}\Big|\frac1T\sum_{t=1}^T g_t(\ppi)\Big|\Big]
\ \le\ 2\mathfrak R_T^{\mathrm{seq}}(\mathcal F_D).
\]
The weight $\rhot$ is a polynomial in the coordinate indicators
$\mathds{1}\{W_{t-m+1,i}=\ppi_{i}(\lHt[-m])\}$, $1\le m\le M$, and is uniformly
bounded by $\prod_{m=1}^MC_{\beta_m}(n,\eta)$.
Using \citet[Corollary~19]{block2021majorizing}, we obtain
\[
\mathfrak R_T^{\mathrm{seq}}(\mathcal F_D)
\ \lesssim\
\prod_{m=1}^MC_{\beta_m}(n,\eta)
\Big(\mathfrak R_T^{\mathrm{seq}}(\mathcal G)+\mathfrak R_T^{\mathrm{seq}}(\mathcal G)\Big),
\]
where $\mathcal G$ is the coordinate policy class in \Cref{assump:coord-capacity}.
By \Cref{assump:coord-capacity},
\[
\E\Big[\sup_{\ppi\in\Pi^M}\big|D(\ppi)-1\big|\Big]
\ \lesssim\ 
\prod_{m=1}^MC_{\beta_m}(n,\eta)\sqrt{\frac{d}{T}}.
\]

To obtain a high-probability bound, we apply \citet[Theorem~1.1]{tropp2011freedman}
to the martingale
\[
M_T(\ppi):=\sum_{t=1}^T g_t(\ppi).
\]
The increments satisfy
\[
|g_t(\ppi)|
\;\le\;
\prod_{m=1}^MC_{\beta_m}(n,\eta)-1
=:b_0,
\]
and hence the conditional variance process is bounded by $T b_0^2$.
Freedman’s inequality, combined with the above expectation bound, yields that
for any $\delta\in(0,1)$, with probability at least $1-\delta$,
\begin{equation}\label{eq:denom-final-mixed}
\sup_{\ppi\in\Pi^M}\big|D(\ppi)-1\big|
\;\lesssim\;
\prod_{m=1}^MC_{\beta_m}(n,\eta)\,\sqrt{\tfrac{d}{T}}
\;+\;\Big(\prod_{m=1}^MC_{\beta_m}(n,\eta)-1\Big)\sqrt{\tfrac{2\log(1/\delta)}{T}}.
\end{equation}

\paragraph{Uniform control of $N(\ppi)$.}
Next consider
\[
\mathcal F_N
=
\Big\{\big(\lHt[-1],\bm W_t,\bm Y_t\big)\to \rhot\,\overline Y_t : \ppi\in\Pi^M\Big\}.
\]
By \citet[Theorem~2]{rakhlin2015sequential},
\[
\E\Big[\sup_{\ppi\in\Pi^M}\Big|\frac1T\sum_{t=1}^T f_t(\ppi)\Big|\Big]
\ \le\ 2\mathfrak R_T^{\mathrm{seq}}(\mathcal F_N).
\]
As above, $\rhot$ depends on the coordinate indicators, is bounded by
$\prod_{m=1}^MC_{\beta_m}(n,\eta)$, and $|\overline Y_t|\le\delta_Y$.
By \citet[Corollary~19]{block2021majorizing},
\[
\mathfrak R_T^{\mathrm{seq}}(\mathcal F_N)
\ \lesssim\
\delta_Y\prod_{m=1}^MC_{\beta_m}(n,\eta)
\mathfrak R_T^{\mathrm{seq}}(\mathcal G),
\]
so that, by \Cref{assump:coord-capacity},
\[
\E\Big[\sup_{\ppi\in\Pi^M}\big|N(\ppi)-V^{\ppi}\big|\Big]
\ \lesssim\ 
\delta_Y\prod_{m=1}^MC_{\beta_m}(n,\eta)\sqrt{\frac{d}{T}}.
\]

Applying Freedman’s inequality to the martingale
$M_T(\ppi):=\sum_{t=1}^T f_t(\ppi)$ with increments bounded by
$|f_t(\ppi)|\le \delta_Y\prod_{m=1}^MC_{\beta_m}(n,\eta)$
and conditional variance bounded by
$T\delta_Y^2\prod_{m=1}^MC_{\beta_m}(n,\eta)^2$, we obtain that for any
$\delta\in(0,1)$, with probability at least $1-\delta$,
\begin{equation}\label{eq:numer-final-mixed}
\sup_{\ppi\in\Pi^M}\big|N(\ppi)-V^{\ppi}\big|
\ \lesssim\ 
\delta_Y\prod_{m=1}^MC_{\beta_m}(n,\eta)\sqrt{\frac{d}{T}}
\ +\ 
\delta_Y\prod_{m=1}^MC_{\beta_m}(n,\eta)\sqrt{\frac{2\log(1/\delta)}{T}}.
\end{equation}

\paragraph{From $N(\ppi),D(\ppi)$ to the mixing estimator $\hatV$.}

We now bound the error of the mixing estimator
\[
\hatV = N(\ppi)+E(\ppi)
\]
relative to the target $V^{\ppi}$.
For any $\ppi\in\Pi^M$,
\[
\hatV-V^{\ppi}
=
\big\{N(\ppi)-V^{\ppi}\big\} + E(\ppi),
\]
so that
\[
\sup_{\ppi\in\Pi^M}|\hatV-V^{\ppi}|
\ \le\
\sup_{\ppi}\big|N(\ppi)-V^{\ppi}\big|
\ +\ 
\sup_{\ppi}|E(\ppi)|.
\]

We first express $E(\ppi)$ as
\[
E(\ppi)
=
\big(1-D(\ppi)\big)\Big\{
\hat\lambda\frac{N(\ppi)}{D(\ppi)}
+\big(1-\hat\lambda\big)\,\overline Y
\Big\}.
\]
By
\Cref{assump:Policy_est1}(a) we have $|V^{\ppi}|\le\delta_Y$, while $|\overline
Y_t|\le\delta_Y$ implies $|\overline Y|\le\delta_Y$.
Using $N(\ppi) = V^{\ppi} + \{N(\ppi)-V^{\ppi}\}$, we can write
\[
\frac{N(\ppi)}{D(\ppi)} - \overline Y
=
\frac{N(\ppi)-V^{\ppi}}{D(\ppi)}
+\frac{V^{\ppi}-\overline Y}{D(\ppi)}.
\]
Since $\hat\lambda\pto\lambda^\ast<\infty$ as $T\to\infty,$ $\hat\lambda$ is bounded in probability.
Consider the high-probability event
\[
\mathcal E_T
=
\Big\{\sup_{\ppi\in\Pi^M}|D(\ppi)-1|\le\tfrac12\Big\}.
\]
By \cref{eq:denom-final-mixed}, for $T$ sufficiently large this event has probability at least $1-\delta$.
On $\mathcal E_T$, we have $D(\ppi)\in[\tfrac12,\tfrac32]$ for all $\ppi\in\Pi^M$, and hence
$|1/D(\ppi)|\le 2$ and $|1-D(\ppi)|\le |D(\ppi)-1|$.
Moreover, $|V^{\ppi}-\overline Y|\le 2\delta_Y$ for all $\ppi$.

Using these facts, on $\mathcal E_T$ we obtain
\begin{align*}
|E(\ppi)|
&\le
|1-D(\ppi)|\Big(
\big|\hat \lambda\big|\Big|\frac{N(\ppi)}{D(\ppi)}-\overline Y\Big|
+\big|1-\hat \lambda\big|\,|\overline Y|
\Big)\\
&\le
|1-D(\ppi)|\Big(
C_\lambda\Big(\frac{|N(\ppi)-V^{\ppi}|}{|D(\ppi)|}
+\frac{|V^{\ppi}-\overline Y|}{|D(\ppi)|}\Big)
+ |\overline Y|
\Big)\\
&\lesssim
|D(\ppi)-1|\Big(
|N(\ppi)-V^{\ppi}| + \delta_Y
\Big),
\end{align*}
where $C_\lambda$ is a finite bound on $|\hat\lambda|$ and we used
$|D(\ppi)|^{-1}\le 2$, $|\overline Y|\le\delta_Y$, and $|V^{\ppi}-\overline Y|\le 2\delta_Y$.
Taking the supremum over $\ppi\in\Pi^M$, we obtain on $\mathcal E_T$,
\begin{equation}\label{eq:E-bound}
\sup_{\ppi\in\Pi^M}|E(\ppi)|
\ \lesssim\
\sup_{\ppi}|D(\ppi)-1|\cdot\sup_{\ppi}|N(\ppi)-V^{\ppi}|
\ +\
\delta_Y\sup_{\ppi}|D(\ppi)-1|.
\end{equation}

By \cref{eq:denom-final-mixed} and \cref{eq:numer-final-mixed},
$\sup_{\ppi}|D(\ppi)-1|=O_p\big(\prod_{m=1}^MC_{\beta_m}(n,\eta)\sqrt{(d+\log(1/\delta))/T}\big)$ and
$\sup_{\ppi}|N(\ppi)-V^{\ppi}|=O_p\big(\delta_Y\prod_{m=1}^MC_{\beta_m}(n,\eta)\sqrt{(d+\log(1/\delta))/T}\big)$.
The product term in \cref{eq:E-bound} is then of order $O_p(T^{-1})$ and is dominated by the $T^{-1/2}$ terms.
Combining these bounds yields, for any $\delta\in(0,1)$ and $T$ large enough,
\[
\sup_{\ppi\in\Pi^M}|E(\ppi)|
\ \lesssim\
\delta_Y\prod_{m=1}^MC_{\beta_m}(n,\eta)
\Big(
\sqrt{\tfrac{d}{T}}+\sqrt{\tfrac{\log(1/\delta)}{T}}
\Big)
\]
with probability at least $1-\delta$.

Finally, using
\[
\sup_{\ppi}|\hatV-V^{\ppi}|
\ \le\
\sup_{\ppi}|N(\ppi)-V^{\ppi}|
\ +\
\sup_{\ppi}|E(\ppi)|,
\]
and plugging in \cref{eq:numer-final-mixed} and the above bound on $\sup|E(\ppi)|$, we conclude that, for any $\delta\in(0,1)$, with probability at least $1-\delta$,
\[
\sup_{\ppi\in\Pi^M}\big|\hatV-V^{\ppi}\big|
\ \lesssim\
\frac{\delta_Y\prod_{m=1}^MC_{\beta_m}(n,\eta)\left(1+2\sqrt{d}+\sqrt{2\log(1/\delta)}\right)}{\sqrt{T}}.
\]

\paragraph{From uniform value error to regret.}
By the definition of regret,
\[
R(\hat\ppi)
=
V^{\ppi^{\ast}}-V^{\hat\ppi},
\]
where $\ppi^{\ast}\in\argmax_{\ppi\in\Pi^M}V^{\ppi}$ and
$\hat\ppi\in\argmax_{\ppi\in\Pi^M}\hatV$.
A standard argument gives
\[
R(\hat\ppi)
\le
2\sup_{\ppi\in\Pi^M}|\hatV-V^{\ppi}|.
\]
Combining this with the last display completes the proof of Theorem~\ref{thm:regret}.
\end{proof}

\subsection{Proof of regret bound based on estimated propensity score}\label{a:subsec:proof-regret-estimateps}

\begin{lemma}[$\sqrt{T}$-consistency of the PS estimator]\label{lem:ps_rate}
Under Assumption~\ref{assump: ps_model}, the solution $\hat\ttheta$ to $\psi_T(\ttheta)=0$ satisfies
\[
\sqrt{T}(\hat\ttheta-\ttheta^\ast)=O_p(1).
\]
\end{lemma}

\begin{proof}
By Assumption~\ref{assump: ps_model}(1a–c), the score array
$\{\psi(W_t,\lHt[-1];\ttheta^\ast)\}_{t=M}^T$ is an $L^2$ martingale difference sequence with asymptotic variance
$\mat\Sigma_{ps}$. Applying \citet[Theorem~A.3]{papadogeorgou2022causal}, we obtain
\[
\sqrt{T}(\hat\ttheta-\ttheta^\ast)\dto N(\bm 0, \mat\Sigma_{ps}).
\]
In particular, $\sqrt{T}(\hat\ttheta-\ttheta^\ast)=O_p(1)$.
\end{proof}

\begin{proof}[Proof of Theorem~\ref{thm: regret_estimate_ps}]
For notational simplicity, write the mixing estimator with propensity parameter
$\ttheta$ as
\[
\hatVPS[][\bbeta][\ppi][\ttheta][]
=
N_{\ttheta}(\ppi) + E_{\ttheta}(\ppi),
\]
where
\[
N_{\ttheta}(\ppi)
=
\frac{1}{T}\sum_{t=1}^T \rhotPS[\bbeta][\ppi][\ttheta][]\,\overline Y_t,
\qquad
D_{\ttheta}(\ppi)
=
\frac{1}{T}\sum_{t=1}^T \rhotPS[\bbeta][\ppi][\ttheta][],
\]
and
\[
E_{\ttheta}(\ppi)
=
\Big\{\hat\lambda\,\frac{N_{\ttheta}(\ppi)}{D_{\ttheta}(\ppi)}
      + \big(1-\hat\lambda\big)\,\overline Y\Big\}
\big(1-D_{\ttheta}(\ppi)\big),
\]
with $\overline Y = T^{-1}\sum_{t=1}^T\overline Y_t$ and
$\hat\lambda$ as defined in \cref{a:sec:lambda}.

Let $\ppi^{\ast}\in\argmax_{\ppi\in\Pi^M}V^{\ppi}$ be an optimal policy and
$\hat\ppi_{\hat\ttheta}\in\argmax_{\ppi\in\Pi^M}\hatVPS[][\bbeta][\ppi][\hat\ttheta][]$
be the learned policy based on the estimated propensity score.
By optimality and the usual decomposition for regret,
\[
R(\hat\ppi_{\hat\ttheta})
= V^{\ppi^{\ast}}-V^{\hat\ppi_{\hat\ttheta}}
\le 2\sup_{\ppi\in\Pi^M}\big|\hatVPS[][\bbeta][\ppi][\hat\ttheta][]-V^{\ppi}\big|.
\]
We further split
\[
\sup_{\ppi}\big|\hatVPS[][\bbeta][\ppi][\hat\ttheta][]-V^{\ppi}\big|
\;\le\;
\sup_{\ppi}\big|\hatVPS[][\bbeta][\ppi][\ttheta^\ast][]-V^{\ppi}\big|
+
\sup_{\ppi}\big|\hatVPS[][\bbeta][\ppi][\hat\ttheta][]-\hatVPS[][\bbeta][\ppi][\ttheta^\ast][]\big|.
\]
The first term is $O_p(T^{-1/2})$ by Theorem~\ref{thm:regret}, since that result
was proved for the mixing estimator under the true propensity scores
$\ttheta^\ast$. Thus it remains to show that
\[
\sup_{\ppi\in\Pi^M}\big|\hatVPS[][\bbeta][\ppi][\hat\ttheta][]-\hatVPS[][\bbeta][\ppi][\ttheta^\ast][]\big|
=O_p(T^{-1/2}).
\]

\medskip\noindent
\textbf{Step 1: Mean value expansion in $\ttheta$.}
For each $\ppi\in\Pi^M$, by the mean value theorem there exists
a point $\bar\ttheta$ on the line segment joining $\ttheta^\ast$ and $\hat\ttheta$
such that
\[
\hatVPS[][\bbeta][\ppi][\hat\ttheta][]
-
\hatVPS[][\bbeta][\ppi][\ttheta^\ast][]
=
\big[\partial_{\ttheta}\hatVPS[][\bbeta][\ppi][\bar\ttheta][]\big]^\top
(\hat\ttheta-\ttheta^\ast).
\]
Hence
\[
\sup_{\ppi}\big|\hatVPS[][\bbeta][\ppi][\hat\ttheta][]-\hatVPS[][\bbeta][\ppi][\ttheta^\ast][]\big|
\le
\Big(\sup_{\ppi}\big\|\partial_{\ttheta}\hatVPS[][\bbeta][\ppi][\bar\ttheta][]\big\|\Big)
\|\hat\ttheta-\ttheta^\ast\|.
\]
By Lemma~\ref{lem:ps_rate}, $\|\hat\ttheta-\ttheta^\ast\|=O_p(T^{-1/2})$.
Therefore it suffices to show that
\[
\sup_{\ppi}\big\|\partial_{\ttheta}\hatVPS[][\bbeta][\ppi][\bar\ttheta][]\big\|=O_p(1).
\]

\medskip\noindent
\textbf{Step 2: Bounding the derivative of the mixing estimator.}
We first compute the derivatives of $N_{\ttheta}(\ppi)$ and $D_{\ttheta}(\ppi)$.
By definition,
\[
N_{\ttheta}(\ppi)
=
\frac{1}{T}\sum_{t=1}^T \rhotPS[\bbeta][\ppi][\ttheta][]\,\overline Y_t,
\qquad
D_{\ttheta}(\ppi)
=
\frac{1}{T}\sum_{t=1}^T \rhotPS[\bbeta][\ppi][\ttheta][],
\]
so
\[
\partial_{\ttheta} N_{\ttheta}(\ppi)
=
\frac{1}{T}\sum_{t=1}^T
\partial_{\ttheta}\rhotPS[\bbeta][\ppi][\ttheta][]\,\overline Y_t,
\qquad
\partial_{\ttheta} D_{\ttheta}(\ppi)
=
\frac{1}{T}\sum_{t=1}^T
\partial_{\ttheta}\rhotPS[\bbeta][\ppi][\ttheta][].
\]
Assumption~\ref{assump: ps_model} implies
\[
\partial_{\ttheta}\rhotPS[\bbeta][\ppi][\ttheta][]
=
-\rhotPS[\bbeta][\ppi][\ttheta][]\,\psi(W_t,\lHt[-1];\ttheta)
= -\rhotPS[\bbeta][\ppi][\ttheta][]\,\psi_t(\ttheta),
\]
where we write $\psi_t(\ttheta)=\psi(W_t,\lHt[-1];\ttheta)$ for brevity.

By Assumption~\ref{assump:Policy_est1}(a), $|\overline Y_t|\le\delta_Y$ and
$|V^{\ppi}|\le\delta_Y$, and hence $|\overline Y|\le\delta_Y$.
Furthermore, by \Cref{assump: factor_ps} and the form of the additive IPW
weights,
\[
0\ \le\ \rhotPS[\bbeta][\ppi][\ttheta][]\ \le\ \prod_{m=1}^M C_{\beta_m}(n,\eta)
\quad\text{for all $t$, all $\ppi$, and all $\ttheta$ near $\ttheta^\ast$}.
\]

We now bound $\partial_{\ttheta}\hatVPS[][\bbeta][\ppi][\ttheta][]$.
Recall
\[
\hatVPS[][\bbeta][\ppi][\ttheta][]
=
N_{\ttheta}(\ppi)
+
E_{\ttheta}(\ppi),
\]
with
\[
E_{\ttheta}(\ppi)
=
\Big\{\lambda(D_{\ttheta}(\ppi))\frac{N_{\ttheta}(\ppi)}{D_{\ttheta}(\ppi)}
+\big(1-\lambda(D_{\ttheta}(\ppi))\big)\overline Y\Big\}(1-D_{\ttheta}(\ppi)).
\]
Differentiating,
\[
\partial_{\ttheta}\hatVPS[][\bbeta][\ppi][\ttheta][]
=
\partial_{\ttheta}N_{\ttheta}(\ppi)
+
\partial_{\ttheta}E_{\ttheta}(\ppi).
\]
The derivative $\partial_{\ttheta}E_{\ttheta}(\ppi)$ is a linear combination of
$\partial_{\ttheta}N_{\ttheta}(\ppi)$ and $\partial_{\ttheta}D_{\ttheta}(\ppi)$,
with coefficients depending smoothly on $D_{\ttheta}(\ppi)$, $N_{\ttheta}(\ppi)$,
$\overline Y$, and $\hat\lambda$ and its derivative. Since $\hat\lambda$ is a
bounded smooth function in $\ttheta$ and $N_{\ttheta}(\ppi)$, $D_{\ttheta}(\ppi)$ and
$\overline Y$ are uniformly bounded (using the bounds on $\rhotPS$ and 
$|\overline Y_t|$), there exists a constant $C>0$ such that for all $\ttheta$
in a neighborhood of $\ttheta^\ast$,
\[
\big\|\partial_{\ttheta}\hatVPS[][\bbeta][\ppi][\ttheta][]\big\|
\ \le\ 
C\Big(
\big\|\partial_{\ttheta}N_{\ttheta}(\ppi)\big\|
+
\big\|\partial_{\ttheta}D_{\ttheta}(\ppi)\big\|
\Big).
\]

By Cauchy–Schwarz inequality, for any fixed $\ppi$ and $\ttheta$,
\begin{align*}
\Big\|\partial_{\ttheta}N_{\ttheta}(\ppi)\Big\|
&=
\Big\|\frac{1}{T}\sum_{t=1}^T \rhotPS[\bbeta][\ppi][\ttheta][]\,\psi_t(\ttheta)\,\overline Y_t\Big\|
\ \le\
\frac{\delta_Y}{T}\sum_{t=1}^T \rhotPS[\bbeta][\ppi][\ttheta][]\,\|\psi_t(\ttheta)\|,\\[0.5em]
\Big\|\partial_{\ttheta}D_{\ttheta}(\ppi)\Big\|
&=
\Big\|\frac{1}{T}\sum_{t=1}^T \rhotPS[\bbeta][\ppi][\ttheta][]\,\psi_t(\ttheta)\Big\|
\ \le\
\frac{1}{T}\sum_{t=1}^T \rhotPS[\bbeta][\ppi][\ttheta][]\,\|\psi_t(\ttheta)\|.
\end{align*}
Using Cauchy–Schwarz again,
\[
\frac{1}{T}\sum_{t=1}^T \rhotPS[\bbeta][\ppi][\ttheta][]\,\|\psi_t(\ttheta)\|
\le
\Big(\frac{1}{T}\sum_{t=1}^T (\rhotPS[\bbeta][\ppi][\ttheta][])^2\Big)^{1/2}
\Big(\frac{1}{T}\sum_{t=1}^T \|\psi_t(\ttheta)\|^2\Big)^{1/2}.
\]
As in the proof of Theorem~\ref{thm:regret}, the first factor is bounded
uniformly in $\ppi$ and $\ttheta$ near $\ttheta^\ast$:
\[
\sup_{\ppi,\ \ttheta\ \text{near}\ \ttheta^\ast}
\frac{1}{T}\sum_{t=1}^T (\rhotPS[\bbeta][\ppi][\ttheta][])^2
\ \lesssim\
\prod_{m=1}^M C_{\beta_m}(n,\eta),
\]
hence its square root is $O_p(1)$ uniformly. By Assumption~\ref{assump: ps_model}
and a law of large numbers for the score array, we also have
\[
\sup_{\ttheta\ \text{near}\ \ttheta^\ast}
\frac{1}{T}\sum_{t=1}^T \|\psi_t(\ttheta)\|^2
=O_p(1).
\]
Combining these bounds gives
\[
\sup_{\ppi,\ \ttheta\ \text{near}\ \ttheta^\ast}
\frac{1}{T}\sum_{t=1}^T \rhotPS[\bbeta][\ppi][\ttheta][]\,\|\psi_t(\ttheta)\|
=O_p(1),
\]
and therefore
\[
\sup_{\ppi,\ \ttheta\ \text{near}\ \ttheta^\ast}
\Big(\big\|\partial_{\ttheta}N_{\ttheta}(\ppi)\big\|
+
\big\|\partial_{\ttheta}D_{\ttheta}(\ppi)\big\|\Big)
=O_p(1).
\]
By the previous display relating $\partial_{\ttheta}\hatVPS$ to
$\partial_{\ttheta}N_{\ttheta}$ and $\partial_{\ttheta}D_{\ttheta}$, we obtain
\[
\sup_{\ppi,\ \ttheta\ \text{near}\ \ttheta^\ast}
\big\|\partial_{\ttheta}\hatVPS[][\bbeta][\ppi][\ttheta][]\big\|
=O_p(1).
\]
In particular, this holds at $\ttheta=\bar\ttheta$, so
\[
\sup_{\ppi}\big\|\partial_{\ttheta}\hatVPS[][\bbeta][\ppi][\bar\ttheta][]\big\|
=O_p(1).
\]

\medskip\noindent
\textbf{Step 3: Combine with the rate of $\hat\ttheta$.}
Returning to the mean-value bound from Step~1 and using Lemma~\ref{lem:ps_rate},
we have
\[
\sup_{\ppi}\big|\hatVPS[][\bbeta][\ppi][\hat\ttheta][]-\hatVPS[][\bbeta][\ppi][\ttheta^\ast][]\big|
\le
\Big(\sup_{\ppi}\big\|\partial_{\ttheta}\hatVPS[][\bbeta][\ppi][\bar\ttheta][]\big\|\Big)
\|\hat\ttheta-\ttheta^\ast\|
=O_p(1)\cdot O_p(T^{-1/2})
=O_p(T^{-1/2}).
\]

\medskip\noindent
\textbf{Step 4: Conclude the regret rate.}
We have shown that both
\[
\sup_{\ppi}\big|\hatVPS[][\bbeta][\ppi][\ttheta^\ast][]-V^{\ppi}\big|
=O_p(T^{-1/2})
\quad\text{and}\quad
\sup_{\ppi}\big|\hatVPS[][\bbeta][\ppi][\hat\ttheta][]-\hatVPS[][\bbeta][\ppi][\ttheta^\ast][]\big|
=O_p(T^{-1/2}).
\]
Therefore
\[
\sup_{\ppi}\big|\hatVPS[][\bbeta][\ppi][\hat\ttheta][]-V^{\ppi}\big|
=O_p(T^{-1/2}),
\]
and the regret satisfies
\[
R(\hat\ppi_{\hat\ttheta})
\le 2\sup_{\ppi\in\Pi^M}\big|\hatVPS[][\bbeta][\ppi][\hat\ttheta][]-V^{\ppi}\big|
=O_p(T^{-1/2}).
\]
This completes the proof of Theorem~\ref{thm: regret_estimate_ps}.
\end{proof}

\section{Alternative assumptions and estimators}
\label{sec:alternative_assumptions}

In this section, we present two extensions to our methodology. First, we show how the policy value estimator can be refined if we have a priori knowledge about the true network structure governing interference. Second, we discuss how the unit-level independent treatment assignment assumption can be relaxed to accommodate block-level dependence. We present the results for a single-period policy $\ppi$, noting that these results naturally extend to multi-period policies $\ppi^{(M)}$ through the sequential multiplicative structure established in \cref{sec:estimation}.

\subsection{Restricting the outcome model via network information}
\label{subsec:network_restriction}

In \Cref{assump: additive_model}, we allowed for all possible linear and interaction terms up to order $\beta$, meaning that any unit could potentially interact with any other unit in driving potential outcomes. However, researchers often have prior structural information on the form of interference. Here, we discuss how that prior knowledge can be incorporated by restricting the model, thereby reducing the number of nuisance parameters and potentially improving estimation efficiency.

We formalize this by first allowing interactions up to order $\beta$ in the potential outcome model for each unit, and then specifying which interactions are included at each order for that specific unit. Specifically, define $\mathcal{N}_{k,i}$ as a prespecified unit-specific index set of $k$-tuples that specifies exactly which $k$-way interactions are included in the outcome model for unit $i$, for $k = 1, \dots, \beta$. To avoid duplicate representations of the same interaction, we take the indices within each tuple to be ordered. We also assume that these interaction sets are hierarchical, meaning that whenever a higher-order interaction is included, all of its lower-order component interactions are also included. 

Depending on the application and the prior knowledge, these interaction sets can be defined in different ways. For example, one might choose to allow for interactions of treatments among units based on their proximity on a network. Specifically, they might define $\mathcal{N}_{1,i}$ to include unit $i$ and its direct neighbors, then define $\mathcal{N}_{2,i}$ to contain only pairs consisting of unit $i$ and one of its direct neighbors, and restrict $\mathcal{N}_{3,i}$ to contain triplets consisting of unit $i$ and any two of its direct neighbors, and so on. In the model of \Cref{assump: additive_model_multi}, these sets correspond to the collection of all possible $k$-tuples for every unit $i$, so all interactions up to order $\beta$ are included.

Using the sets $\mathcal{N}_{1,i}, \dots, \mathcal{N}_{\beta,i}$ to specify which interactions are included for each unit, we formalize the following restricted semiparametric model:

\begin{assumption}[Network-Informed Additive Outcome Model]
\label{assump: network_model}
    For $1\le t\le T$, $1\le i\le n$, and $\bm w_t\in\{0,1\}^n$, the potential outcomes satisfy
    \begin{equation}
    \label{eq: network_model}
      Y_{ti}(\lWt[-1],\bm w_t)
      =
      \bm g_{ti}\left(\lHt[-1]\right)^\top
      \phi_i^{inf}(\bm w_t),
    \end{equation}
    where $\phi_i^{inf}(\bm w_t)$ is a unit-specific, network-restricted treatment vector at time $t$ that only includes the specific main effects and interaction terms defined by the sets $\mathcal{N}_{1,i}, \dots, \mathcal{N}_{\beta,i}$,
    \[
    \phi_i^{inf}(\bm w_t)
    =
    \left(
    1,
    \{w_{tj}\}_{j\in\mathcal{N}_{1,i}},
    \{w_{tj_1}w_{tj_2}\}_{(j_1,j_2)\in\mathcal{N}_{2,i}},
    \dots,
    \{w_{tj_1}\dots w_{tj_\beta}\}_{(j_1,\dots,j_\beta)\in\mathcal{N}_{\beta,i}}
    \right)^\top.
    \]
\end{assumption}

Under this assumption, we propose a network-restricted policy value estimator,
\[
\widetilde{V}^{inf}(\ppi)
=
\frac{1}{T}
\sum_{t=1}^T
\widetilde{V}_t^{inf}(\ppi),
\]
where the single-period estimator is defined as
\[
\widetilde{V}_t^{inf}(\ppi)
=
\frac{1}{n}
\sum_{i=1}^n
\rho_{ti}^{inf}(\ppi)Y_{ti}.
\]
The corresponding restricted weight for each unit $i$ is given by
\begin{align}
\label{eq: netIPW_weight}
\begin{split}
 \rho_{ti}^{inf}(\ppi)
 = 1
 &+ \sum_{j\in\mathcal{N}_{1,i}}
 \left(\rho_{tj}(\ppi)-1\right) \\
 &+ \sum_{(j_1,j_2)\in\mathcal{N}_{2,i}}
 \left(\rho_{tj_1}(\ppi)-1\right)
 \left(\rho_{tj_2}(\ppi)-1\right)
 + \dots \\
 &+ \sum_{(j_1,\dots,j_\beta)\in\mathcal{N}_{\beta,i}}
 \left(\rho_{tj_1}(\ppi)-1\right)
 \dots
 \left(\rho_{tj_\beta}(\ppi)-1\right),
\end{split}
\end{align}
where
\[
\rho_{tj}(\ppi)
=
\frac{
\mathds{1}\left\{
W_{tj}
=
\pi_j\left(\lHt[-1]\right)
\right\}
}{
e_{tj}(W_{tj})
}.
\]
We next show that $\widetilde{V}_t^{inf}(\ppi)$ is unbiased for $V_t^{\ppi}$ under the modified assumption.

\begin{theorem}[Conditional unbiasedness of the network-restricted estimator]
\label{thm: network_unbiased}
    Under Assumptions \ref{assump: unconfoundedness}, \ref{assump: factor_ps}, and \ref{assump: network_model},
    \[
    \E\left[
    \widetilde{V}_t^{inf}(\ppi)
    \mid
    \lHt[-1]
    \right]
    =
    V_t^{\ppi}\left(\lHt[-1]\right).
    \]
\end{theorem}

\begin{proof}
The proof consists of two steps. First, we establish an algebraic equivalence between the unit-specific restricted additive weights $\rho_{ti}^{inf}(\ppi)$ defined in \cref{eq: netIPW_weight} and a linear projection matrix form. Second, we evaluate the conditional expectation using this matrix form.

Recall that
\[
e_{tj}(1)
=
\P\left(
W_{tj}=1
\mid
\lHt[-1]
\right).
\]
We apply a linear change of basis from the raw treatment indicators to the centered variables
\[
C_{tj}
=
W_{tj}-e_{tj}(1).
\]
Let $\bm Z_{ti}^{inf}(\bm w_t)$ be the unit-specific network-restricted basis vector using these centered variables, constructed identically to $\phi_i^{inf}(\bm w_t)$:
\[
\bm Z_{ti}^{inf}(\bm w_t)
=
\left(
1,
\{w_{tj}-e_{tj}(1)\}_{j\in\mathcal{N}_{1,i}},
\left\{
(w_{tj_1}-e_{tj_1}(1))
(w_{tj_2}-e_{tj_2}(1))
\right\}_{(j_1,j_2)\in\mathcal{N}_{2,i}},
\dots
\right)^\top.
\]
Because the interaction sets are hierarchical, $\phi_i^{inf}(\bm w_t)$ and $\bm Z_{ti}^{inf}(\bm w_t)$ span the same linear space. Therefore, after ordering the terms by interaction order, there exists an invertible lower-triangular matrix $\bm L_{ti}$ such that
\[
\phi_i^{inf}(\bm w_t)
=
\bm L_{ti}
\bm Z_{ti}^{inf}(\bm w_t).
\]
Using this property, the projection matrix term for unit $i$ can be rewritten independently of $\bm L_{ti}$:
\begin{align*}
    &\phi_i^{inf}\left(\ppi\left(\lHt[-1]\right)\right)^\top
    \E\left[
    \phi_i^{inf}(\bm W_t)
    \phi_i^{inf}(\bm W_t)^\top
    \mid
    \lHt[-1]
    \right]^{-1}
    \phi_i^{inf}(\bm W_t) \\
    &\qquad =
    \bm Z_{ti}^{inf}\left(\ppi\left(\lHt[-1]\right)\right)^\top
    \E\left[
    \bm Z_{ti}^{inf}(\bm W_t)
    \bm Z_{ti}^{inf}(\bm W_t)^\top
    \mid
    \lHt[-1]
    \right]^{-1}
    \bm Z_{ti}^{inf}(\bm W_t).
\end{align*}

Under \Cref{assump: factor_ps}, treatments are conditionally independent. Thus,
\[
\E\left[
C_{tj}
\mid
\lHt[-1]
\right]
=
0.
\]
Moreover, the conditional expectation of the cross-product of any two distinct centered interaction terms is zero. In particular, for any two distinct interaction tuples, there is at least one centered treatment indicator that appears only once in their cross-product, and its conditional expectation is zero. Therefore, the expected outer product matrix
\[
\bm D_{ti}
=
\E\left[
\bm Z_{ti}^{inf}(\bm W_t)
\bm Z_{ti}^{inf}(\bm W_t)^\top
\mid
\lHt[-1]
\right]
\]
is diagonal.

The diagonal entry corresponding to a $k$-way interaction tuple $(j_1,\dots,j_k)\in\mathcal{N}_{k,i}$ is
\[
\prod_{m=1}^k
e_{tj_m}(1)
\left\{
1-e_{tj_m}(1)
\right\}.
\]
Since the inverse of a diagonal matrix is obtained by taking the reciprocal of each diagonal element, expanding the inner product
\[
\bm Z_{ti}^{inf}\left(\ppi\left(\lHt[-1]\right)\right)^\top
\bm D_{ti}^{-1}
\bm Z_{ti}^{inf}(\bm W_t)
\]
yields, writing $\pi_j=\pi_j(\lHt[-1])$ to simplify notation,
\begin{align}
\label{eq:centered_expansion}
    1
    &+
    \sum_{j\in\mathcal{N}_{1,i}}
    \frac{
    (\pi_j-e_{tj}(1))
    (W_{tj}-e_{tj}(1))
    }{
    e_{tj}(1)(1-e_{tj}(1))
    }
    \nonumber \\
    &+
    \sum_{(j_1,j_2)\in\mathcal{N}_{2,i}}
    \frac{
    (\pi_{j_1}-e_{tj_1}(1))
    (W_{tj_1}-e_{tj_1}(1))
    }{
    e_{tj_1}(1)(1-e_{tj_1}(1))
    }
    \frac{
    (\pi_{j_2}-e_{tj_2}(1))
    (W_{tj_2}-e_{tj_2}(1))
    }{
    e_{tj_2}(1)(1-e_{tj_2}(1))
    }
    +\dots.
\end{align}
For any binary treatment $W_{tj}\in\{0,1\}$ and policy decision $\pi_j\in\{0,1\}$, which is fixed conditional on $\lHt[-1]$, we have the algebraic identity
\[
\frac{
(\pi_j-e_{tj}(1))
(W_{tj}-e_{tj}(1))
}{
e_{tj}(1)(1-e_{tj}(1))
}
=
\frac{
\mathds{1}(W_{tj}=\pi_j)
}{
e_{tj}(W_{tj})
}
-1
=
\rho_{tj}(\ppi)-1.
\]
Substituting this identity into \cref{eq:centered_expansion} exactly recovers the unit-specific network-restricted weight $\rho_{ti}^{inf}(\ppi)$. Therefore, the proposed estimator $\widetilde{V}^{inf}_t(\ppi)$ can be equivalently written in matrix projection form:
\begin{equation}
\label{eq:network_projection_estimator}
\widetilde{V}^{inf}_t(\ppi)
=
\frac{1}{n}
\sum_{i=1}^n
\phi_i^{inf}\left(\ppi\left(\lHt[-1]\right)\right)^\top
\E\left[
\phi_i^{inf}(\bm W_t)
\phi_i^{inf}(\bm W_t)^\top
\mid
\lHt[-1]
\right]^{-1}
\phi_i^{inf}(\bm W_t)Y_{ti}.
\end{equation}

Finally, we evaluate the conditional expectation of this form. By consistency and \Cref{assump: network_model},
\[
Y_{ti}
=
Y_{ti}(\lWt[-1],\bm W_t)
=
\bm g_{ti}\left(\lHt[-1]\right)^\top
\phi_i^{inf}(\bm W_t).
\]
Thus,
\begin{align*}
    &\E\left[
    \widetilde V_t^{inf}(\ppi)
    \mid
    \lHt[-1]
    \right] \\
    &=
    \frac{1}{n}
    \sum_{i=1}^n
    \phi_i^{inf}\left(\ppi\left(\lHt[-1]\right)\right)^\top
    \E\left[
    \phi_i^{inf}(\bm W_t)
    \phi_i^{inf}(\bm W_t)^\top
    \mid
    \lHt[-1]
    \right]^{-1} \\
    &\qquad\qquad \times
    \E\left[
    \phi_i^{inf}(\bm W_t)
    \phi_i^{inf}(\bm W_t)^\top
    \bm g_{ti}\left(\lHt[-1]\right)
    \mid
    \lHt[-1]
    \right] \\
    &=
    \frac{1}{n}
    \sum_{i=1}^n
    \bm g_{ti}\left(\lHt[-1]\right)^\top
    \phi_i^{inf}\left(\ppi\left(\lHt[-1]\right)\right) \\
    &=
    \frac{1}{n}
    \sum_{i=1}^n
    Y_{ti}^{\ppi}\left(\lHt[-1]\right)
    =
    V_t^{\ppi}\left(\lHt[-1]\right).
\end{align*}
This completes the proof.
\end{proof}

We can further normalize these weights and construct the mixing estimator using a procedure similar to that described in \cref{subsec:single}.

\subsection{Relaxing the conditional independence assumption of the treatment assignment}
\label{subsec:block_dependence}

\Cref{assump: factor_ps} requires treatments to be assigned independently across units conditional on the observed history of past treatments, covariates and outcomes. 
This assumption may be violated in settings where spatial dependence influences joint treatment assignments. 

Here, we relax \Cref{assump: factor_ps}. We allow for units to be partitioned into groups such that treatment assignments can be dependent within a group but independent across groups.
Specifically, assume that the $n$ units can be partitioned into $B$ disjoint blocks, denoted as $\mathcal{B}_{1}, \dots, \mathcal{B}_{B}$. Let $\bm W_{t,\mathcal{B}_b}$ be the treatment vector for units in block $b$. 

\begin{assumption}[Factored block-level propensity score]\label{assump: block_ps} 
At time period $t$, the propensity score satisfies 
$$e_t(\bm w_t) = \prod_{b=1}^B e_{t,\mathcal{B}_b}(\bm w_{t,\mathcal{B}_b}),$$
where $e_{t,\mathcal{B}_b}(\bm w_{t,\mathcal{B}_b}) = \P(\bm W_{t,\mathcal{B}_b} = \bm w_{t,\mathcal{B}_b} \mid \lHt[-1])$. In addition, we assume strict joint positivity for each block: there exists $\eta > 0$ such that $e_{t,\mathcal{B}_b}(\bm w_{t,\mathcal{B}_b}) > \eta$ for all possible treatment combinations $\bm w_{t,\mathcal{B}_b} \in \{0,1\}^{|\mathcal{B}_b|}$.
\end{assumption}

Clearly, by setting $B = n$ such that each unit is its own block, \Cref{assump: block_ps} encompasses the factored propensity score \Cref{assump: factor_ps}.

We restrict our focus to the additive outcome model without interactions ($\beta=1$). Let $e_{t,\mathcal{B}_b} = \E[\bm W_{t,\mathcal{B}_b} \mid \lHt[-1]]$ be the vector of marginal propensity scores, and define the block covariance matrix $\bm D_{tb} = \mathrm{Var}(\bm W_{t,\mathcal{B}_b} \mid \lHt[-1])$. Note that \cref{assump: block_ps} ensures $\bm D_{tb}$ is non-singular. 

Under this relaxed assumption, we define the block-adjusted estimator $\widetilde{V}^{block}(\ppi) = \frac{1}{T} \sum_{t=1}^T \widetilde{V}_t^{block}(\ppi)$, where $\widetilde{V}_t^{block}(\ppi) = \frac{1}{n} \sum_{i=1}^n \rho_t^{block}(\ppi)Y_{ti}$ and the corresponding weights are:
\begin{equation}\label{eq: blockIPW_weight}
 \rho_t^{block}(\ppi) = 1 + \sum_{b=1}^B \big(\ppi_{\mathcal{B}_b}\left(\lHt[-1]\right) - e_{t,\mathcal{B}_b}\big)^\top \bm D_{tb}^{-1} \big(\bm W_{t,\mathcal{B}_b} - e_{t,\mathcal{B}_b}\big).
\end{equation}

\begin{theorem}[Conditional unbiasedness of the block-adjusted estimator]\label{thm: block_unbiased}
    Under Assumptions \ref{assump: unconfoundedness}, \ref{assump: additive_model} (with $\beta=1$), and \ref{assump: block_ps}, $\E\left[\widetilde{V}_t^{block}(\ppi) \mid \lHt[-1]\right] = V_t^{\ppi}\left(\lHt[-1]\right)$.
\end{theorem}

\begin{proof}
As in \cref{thm: network_unbiased}, we use a centered basis to evaluate the matrix inverse. Let the centered design vector for the entire cluster be $\bm Z^{block}_t = (1, \bm C_{t,\mathcal{B}_1}^\top, \dots, \bm C_{t,\mathcal{B}_B}^\top)^\top$, where $\bm C_{t,\mathcal{B}_b} = \bm W_{t,\mathcal{B}_b} - e_{t,\mathcal{B}_b}$.

By \Cref{assump: block_ps}, treatments are independent across blocks. Therefore, $\E[\bm C_{t,\mathcal{B}_b} \bm C_{t,\mathcal{B}_{b'}}^\top \mid \lHt[-1]] = \bm 0$ for all $b \neq b'$. Consequently, the expected outer product matrix $\E\left[\bm Z^{block}_t (\bm Z^{block}_t)^\top \mid \lHt[-1]\right]$ is the block-diagonal matrix:
$$
\E\left[\bm Z^{block}_t (\bm Z^{block}_t)^\top \mid \lHt[-1]\right] = 
\begin{bmatrix}
1 & \mathbf{0} & \dots & \mathbf{0} \\
\mathbf{0} & \bm D_{t1} & \dots & \mathbf{0} \\
\vdots & \vdots & \ddots & \vdots \\
\mathbf{0} & \mathbf{0} & \dots & \bm D_{tB}
\end{bmatrix}
$$
Because the matrix is block-diagonal, its inverse is simply the block-by-block inverse of the submatrices $\bm D_{tb}^{-1}$.

Since the expected outer product matrix is block-diagonal, its inverse decomposes perfectly across the independent blocks. When we evaluate the bilinear projection form by multiplying the partitioned vectors $\bm Z^{block}_t(\ppi)$ and $\bm Z^{block}_t(\bm W_t)$ against this block-diagonal inverse, the cross-block terms vanish, yielding the exact scalar expansion:
\begin{align*}
    &Z^{block}_t(\ppi)^\top \E\left[Z^{block}_t(\bm W_t) Z^{block}_t(\bm W_t)^\top \mid \lHt[-1]\right]^{-1} Z^{block}_t(\bm W_t) \\
    &= 1 \cdot 1 \cdot 1 + \sum_{b=1}^B \big(\ppi_{\mathcal{B}_b}\left(\lHt[-1]\right) - e_{t,\mathcal{B}_b}\big)^\top \bm D_{tb}^{-1} \bm C_{t,\mathcal{B}_b} \\
    &= 1 + \sum_{b=1}^B \big(\ppi_{\mathcal{B}_b}\left(\lHt[-1]\right) - e_{t,\mathcal{B}_b}\big)^\top \bm D_{tb}^{-1} \big(\bm W_{t,\mathcal{B}_b} - e_{t,\mathcal{B}_b}\big).
\end{align*}
This expansion matches the definition of the block-adjusted weight $\rho_t^{block}(\ppi)$ given in \cref{eq: blockIPW_weight}. Therefore, the estimator can be equivalently expressed in the matrix projection form. By substituting the outcome model $Y_{ti} = \bm g_{ti}\left(\lHt[-1]\right)^\top Z^{block}_t(\bm W_t)$ into this form and taking the conditional expectation, and following the exact same steps established in the proof of \cref{thm: network_unbiased}, we recover $V_t^{\ppi}\left(\lHt[-1]\right)$, which completes the proof of conditional unbiasedness.
\end{proof}

Based on the conditional unbiasedness, all subsequent theoretical results derived for the original estimator extend directly to $\widetilde{V}^{block}(\ppi)$ under \Cref{assump: block_ps}. The corresponding arguments follow by similar proofs as in \cref{a: Asymptotic_property}, and are therefore omitted.

Theoretically, the block-adjusted approach can accommodate higher-order interactions ($\beta > 1$) by expanding the block basis vector. However, doing so requires estimating higher-order joint propensity scores and inverting a fully dense covariance matrix whose dimension grows exponentially, rendering higher-order models practically infeasible.

\section{Additional simulation results}
\label{a:sec:additional_sim}

In this section, we present additional simulation results that complement \Cref{sec:simulation}. In \Cref{a:subsec:additional_policy_eval}, we provide policy evaluation results. In \Cref{a:subsec:additional_policy_learning}, we present additional policy learning results for $n=10$ and $n=100$. In \Cref{a:subsec:L_sensitivity}, we summarize the sensitivity analysis for the number of sampled candidate policies $L$ used in the bootstrap test.

\subsection{Results for policy evaluation}
\label{a:subsec:additional_policy_eval}

To evaluate the proposed estimator for policy evaluation, we consider two policies over two time periods: (i) the true optimal policy in the linear class for the $c = 0$ scenario, and (ii) the treat-nobody policy, which assigns no treatments in all scenarios. We compare six estimators: the proposed mixing estimator ($\hatVPS$), the H\'ajek-type estimator ($\hatVPS[H,]$), and the shifted estimator ($\hatVPS[S,]$), each implemented under either a linear additive model ($\bbeta = (1,1)$) or a quadratic additive model ($\bbeta = (2,2)$) for both time periods. We estimate the propensity scores based on the correctly specified propensity score model. 

\Cref{fig:mix_shift_M2} compares the finite-sample performance of the proposed mixing estimator $\hatVPS$ with that of the shifted estimator $\hatVPS[S,]$ for policy value estimation. \Cref{fig:policy_mix_shift_M2} reports estimated policy values based on different estimators. The results illustrate a clear bias--variance tradeoff in selecting the degree of interaction $\bbeta$ in the assumed outcome model. Estimators based on outcome models with two-way interactions exhibit higher variability, reflected in longer error bars, while those assuming a linear additive structure tend to be more stable. However, low-order outcome models may be misspecified, leading to biased estimation. This bias is particularly evident when the true outcome model contains strong interaction effects ($c=-10$), where estimators relying on no interaction terms (addIPW(mix) and addIPW(S)) exhibit noticeable bias in estimating Policy1, despite being nearly unbiased when $c=0$ or $c=-2$. Across all scenarios, the mixing estimator generally outperforms its shifted counterpart. The advantage is modest when assuming linear additive models, but becomes more pronounced under quadratic additive models, where additive IPW weights are more variable, especially with small $T$. When the number of time periods is very large ($T=1000$), the performance of the two estimators becomes nearly indistinguishable.

\Cref{fig:RMSE_mix_shift_M2} presents the RMSE of each estimator across the different scenarios. As $T$ increases, RMSE decreases for all methods, reflecting improved estimation accuracy with larger sample sizes. Estimators based on quadratic additive models consistently have higher RMSE than those using linear additive models, due to increased variability from the more flexible model specification. Under linear additive outcome assumptions, the mixing and shifted estimators perform similarly, with the mixing estimator having a slightly smaller RMSE in most settings. When using quadratic additive models, the mixing estimator yields smaller RMSE than the shifted estimator, particularly for small $T$, while the gap narrows as $T$ increases.

A comparison between the mixing estimator and the H\'ajek-type estimator is shown in \cref{fig:mix_hajek_M2}. Although the H\'ajek-type estimator performs well in most runs, it occasionally produces extreme estimates, resulting in unstable finite-sample behavior. In contrast, the mixing estimator remains substantially more stable across simulations.

We additionally report the policy evaluation results for $n=10$ and $n=100$ in \Cref{fig:mix_shift_M2_n10,fig:mix_hajek_M2_n10,fig:mix_shift_M2_n100,fig:mix_hajek_M2_n100}. Overall, the results are qualitatively similar across different values of $n$. The mixing estimator remains stable across scenarios and generally improves upon the shifted and H\'ajek-type estimators, especially when the quadratic additive model is used and the weights are more variable.

\begin{figure}[!t]
\centering
\subfloat[Estimated policy value]{
\includegraphics[width = \textwidth,trim = 0 0 0 15, clip]{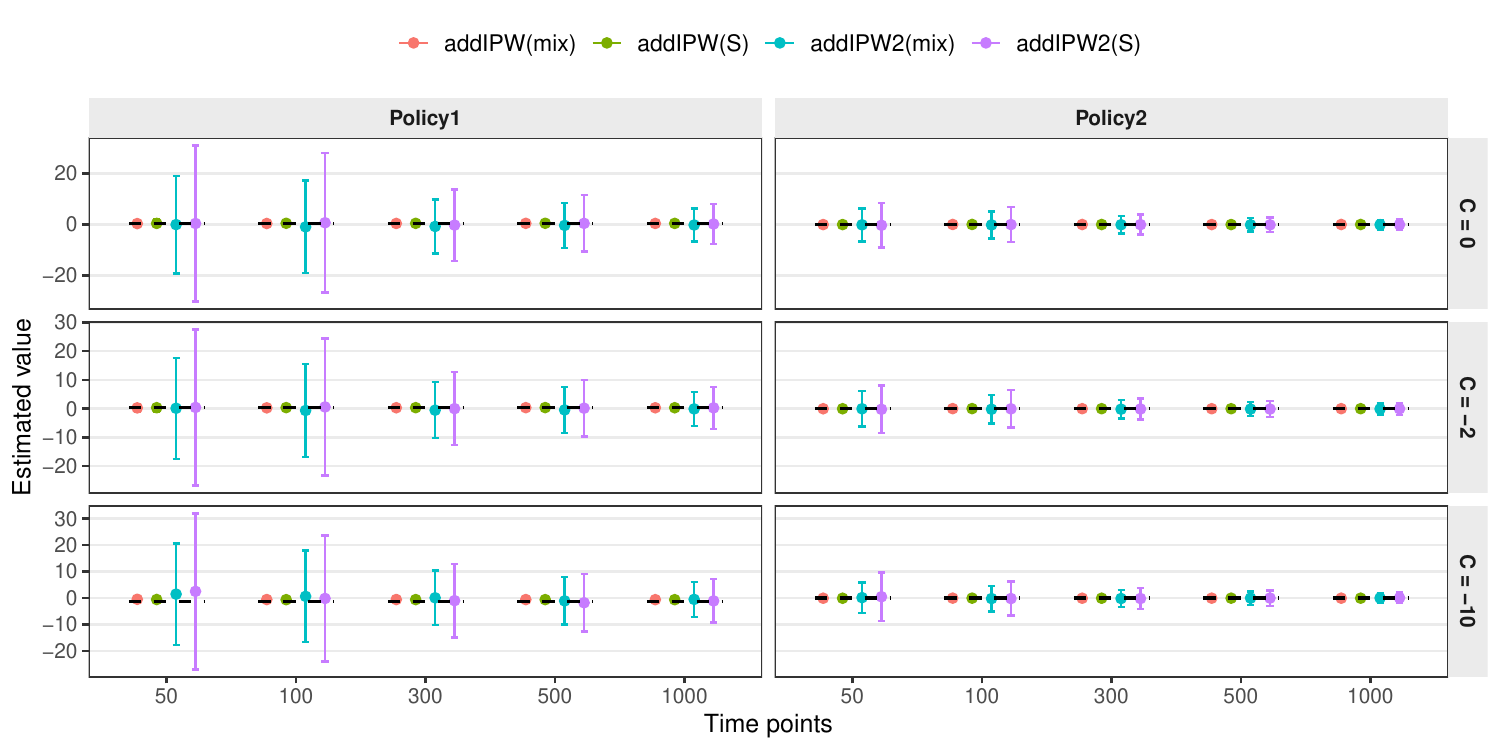}
\label{fig:policy_mix_shift_M2}
}  \\
\subfloat[RMSE]{
\includegraphics[width = \textwidth,trim = 0 0 0 35, clip]{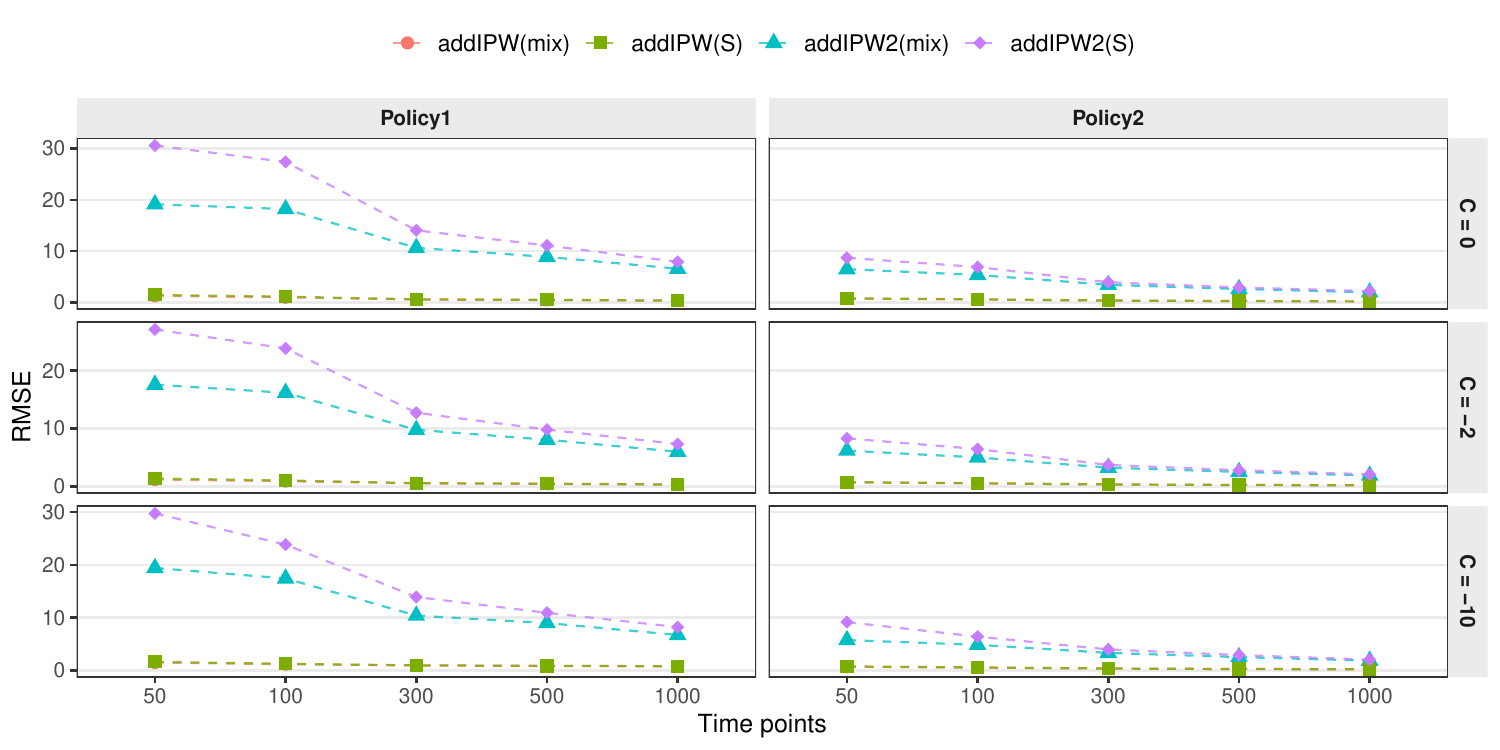}
\label{fig:RMSE_mix_shift_M2}
}
\caption{Performance of policy evaluation using the mixing estimator with $\bbeta=\bm 1$ (red) and $\bbeta=\bm 2$ (blue), and the shifted estimator with $\bbeta=\bm 1$ (green) and $\bbeta=\bm 2$ (purple), for $n=50$ across scenarios with different model specifications (rows) and two different policies (columns). Policy1 is the true optimal policy in the $c = 0$ scenario, and Policy2 is the treat-nobody policy. Panel~(a) reports the estimated policy values. The dots represent the average estimate across simulations, and the lines show one standard deviation above and below the mean. The true policy value is indicated by the dashed black line. Panel~(b) reports the RMSE of each estimator.}
\label{fig:mix_shift_M2}
\end{figure}

\begin{figure}[!t]
\centering
\subfloat[Estimated policy value]{
\includegraphics[width = \textwidth,trim = 0 0 0 15, clip]{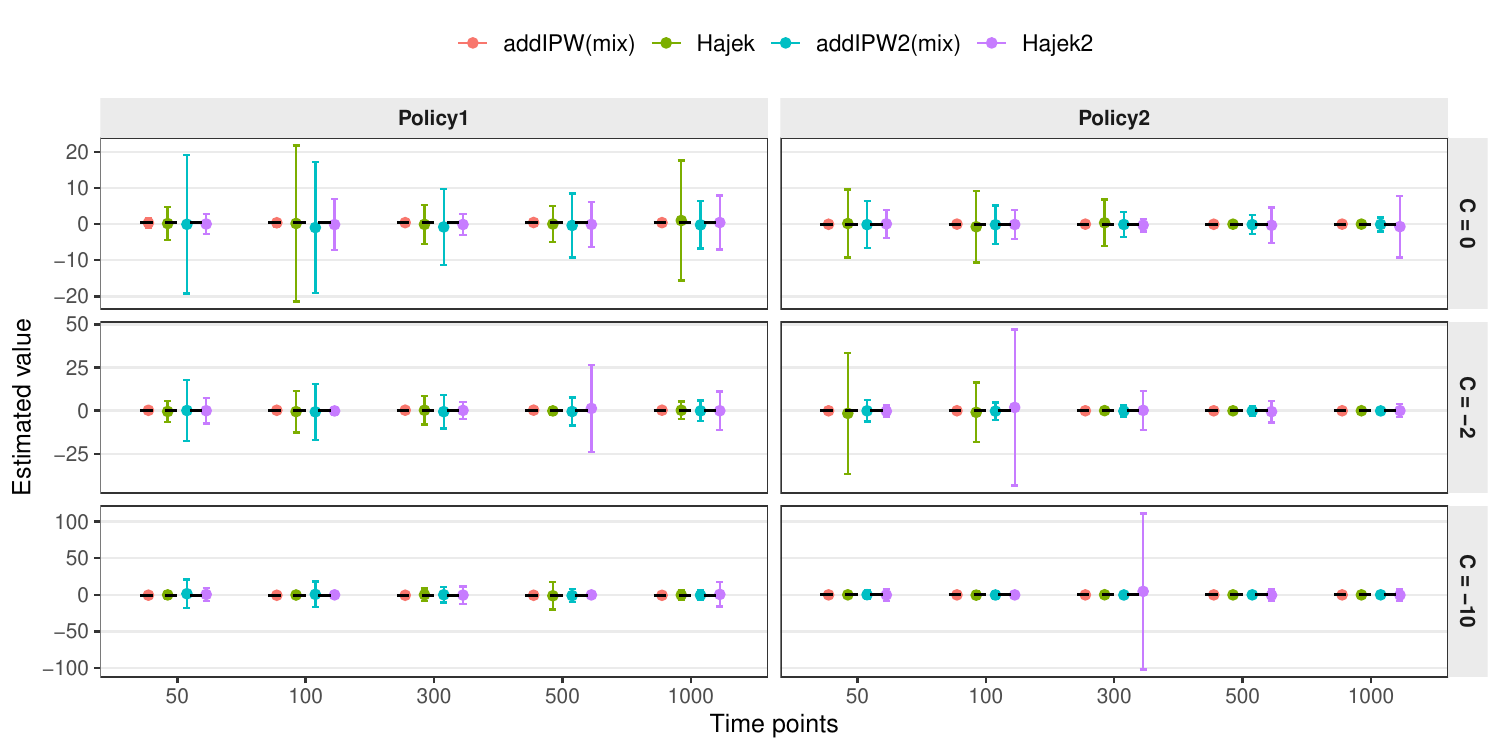}
\label{fig:policy_mix_hajek_M2}
}  \\
\subfloat[RMSE]{
\includegraphics[width = \textwidth,trim = 0 0 0 35, clip]{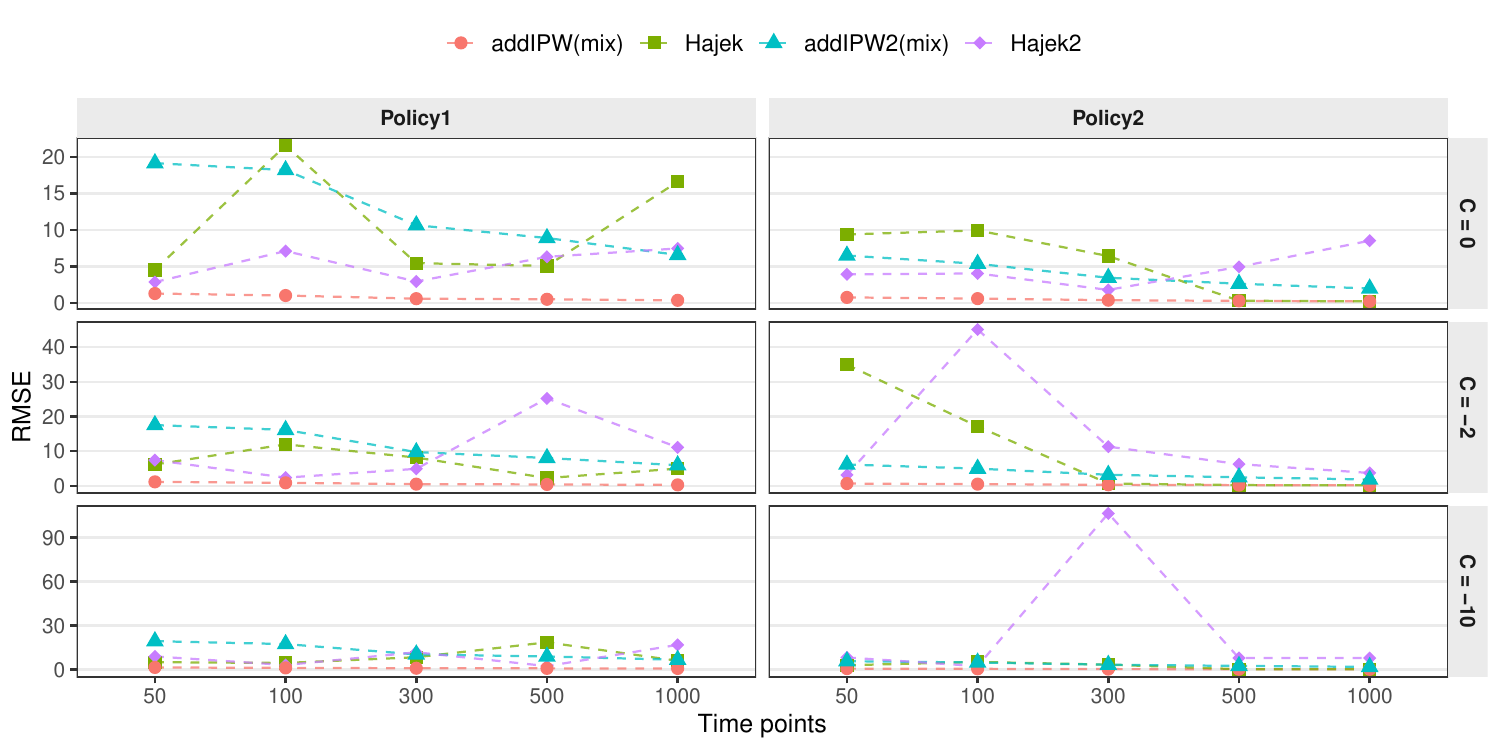}
\label{fig:RMSE_mix_hajek_M2}
}
\caption{Performance of policy evaluation using the mixing estimator with $\bbeta=\bm 1$ (red) and $\bbeta=\bm 2$ (blue), and the H\'ajek-type estimator with $\bbeta=\bm 1$ (green) and $\bbeta=\bm 2$ (purple), for $n=50$ across scenarios with different model specifications (rows) and two different policies (columns). Policy1 is the true optimal policy in the $c = 0$ scenario, and Policy2 is the treat-nobody policy. Panel~(a) reports the estimated policy values. The dots represent the average estimate across simulations, and the lines show one standard deviation above and below the mean. The true policy value is indicated by the dashed black line. Panel~(b) reports the RMSE of each estimator.}
\label{fig:mix_hajek_M2}
\end{figure}

\begin{figure}[!t]
\centering
\subfloat[Estimated policy value]{
\includegraphics[width = \textwidth,trim = 0 0 0 15, clip]{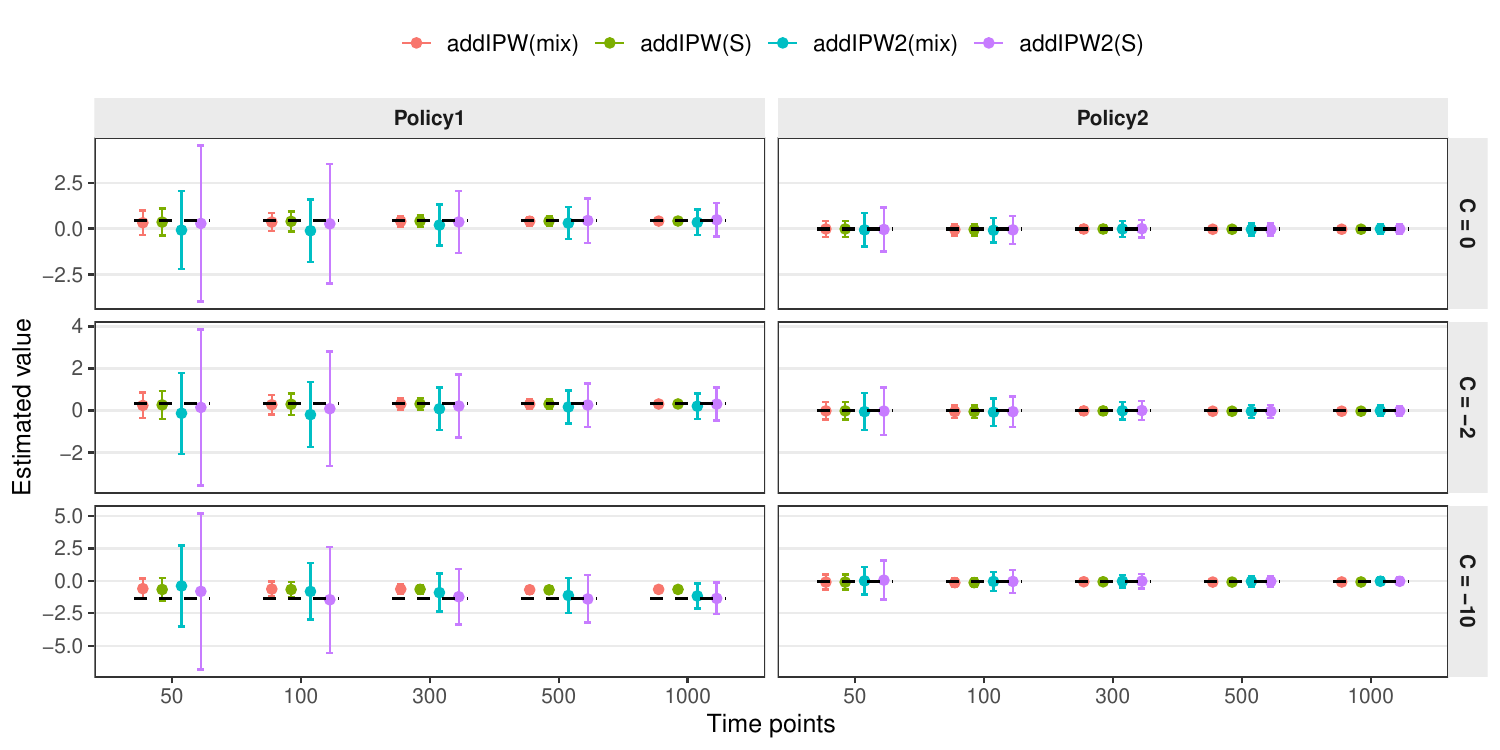}
\label{fig:policy_mix_shift_M2_n10}
}  \\
\subfloat[RMSE]{
\includegraphics[width = \textwidth,trim = 0 0 0 35, clip]{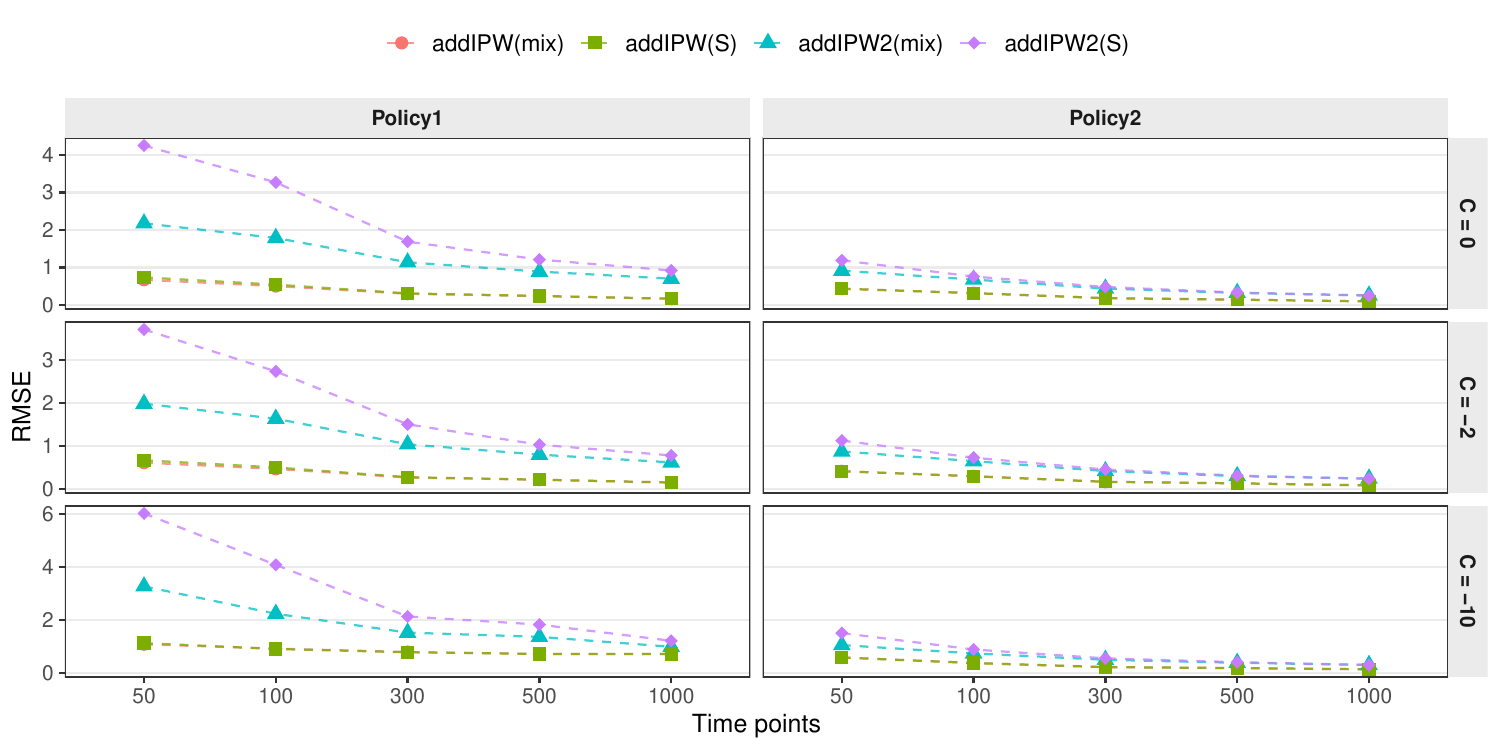}
\label{fig:RMSE_mix_shift_M2_n10}
}
\caption{Performance of policy evaluation using the mixing estimator and the shifted estimator for $n=10$. The plotting convention is the same as in \Cref{fig:mix_shift_M2}.}
\label{fig:mix_shift_M2_n10}
\end{figure}

\begin{figure}[!t]
\centering
\subfloat[Estimated policy value]{
\includegraphics[width = \textwidth,trim = 0 0 0 15, clip]{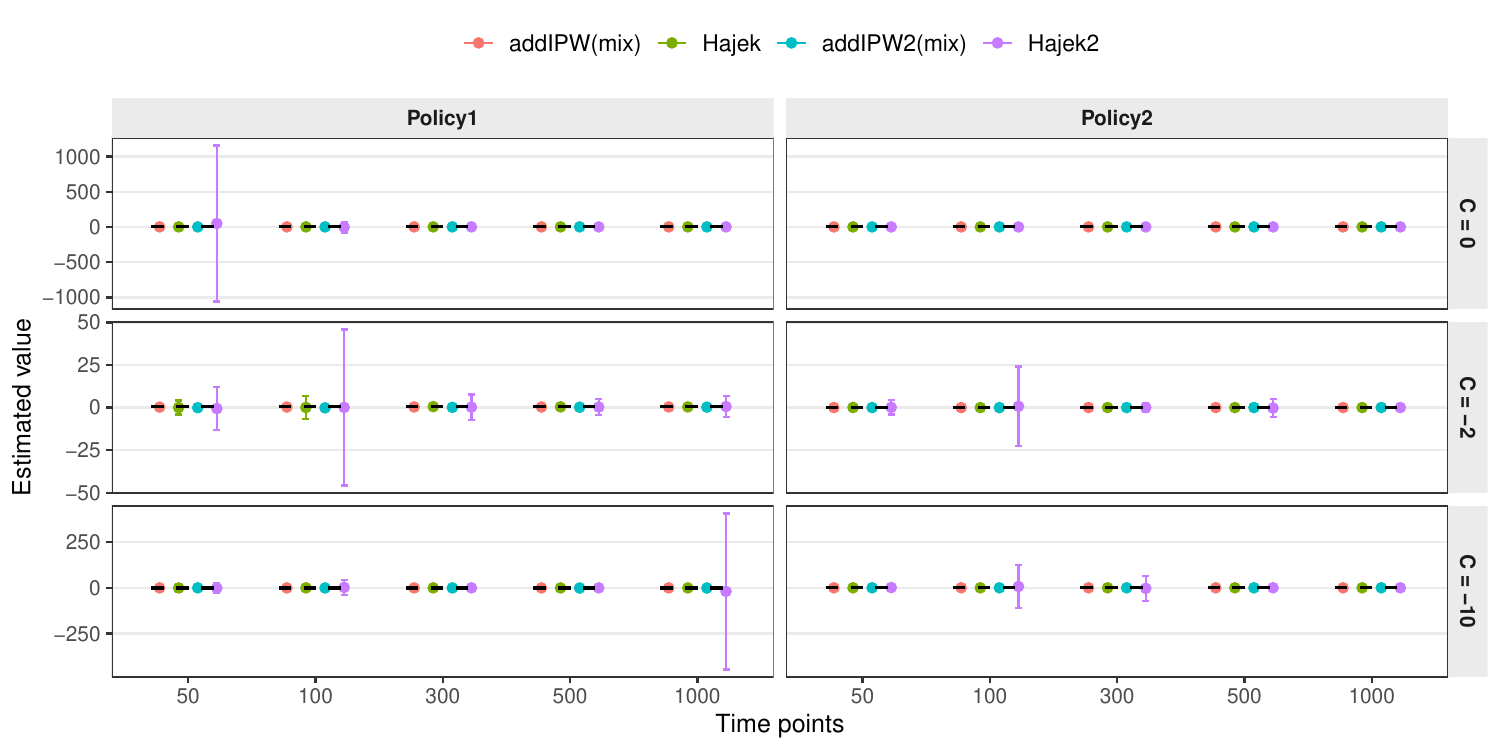}
\label{fig:policy_mix_hajek_M2_n10}
}  \\
\subfloat[RMSE]{
\includegraphics[width = \textwidth,trim = 0 0 0 35, clip]{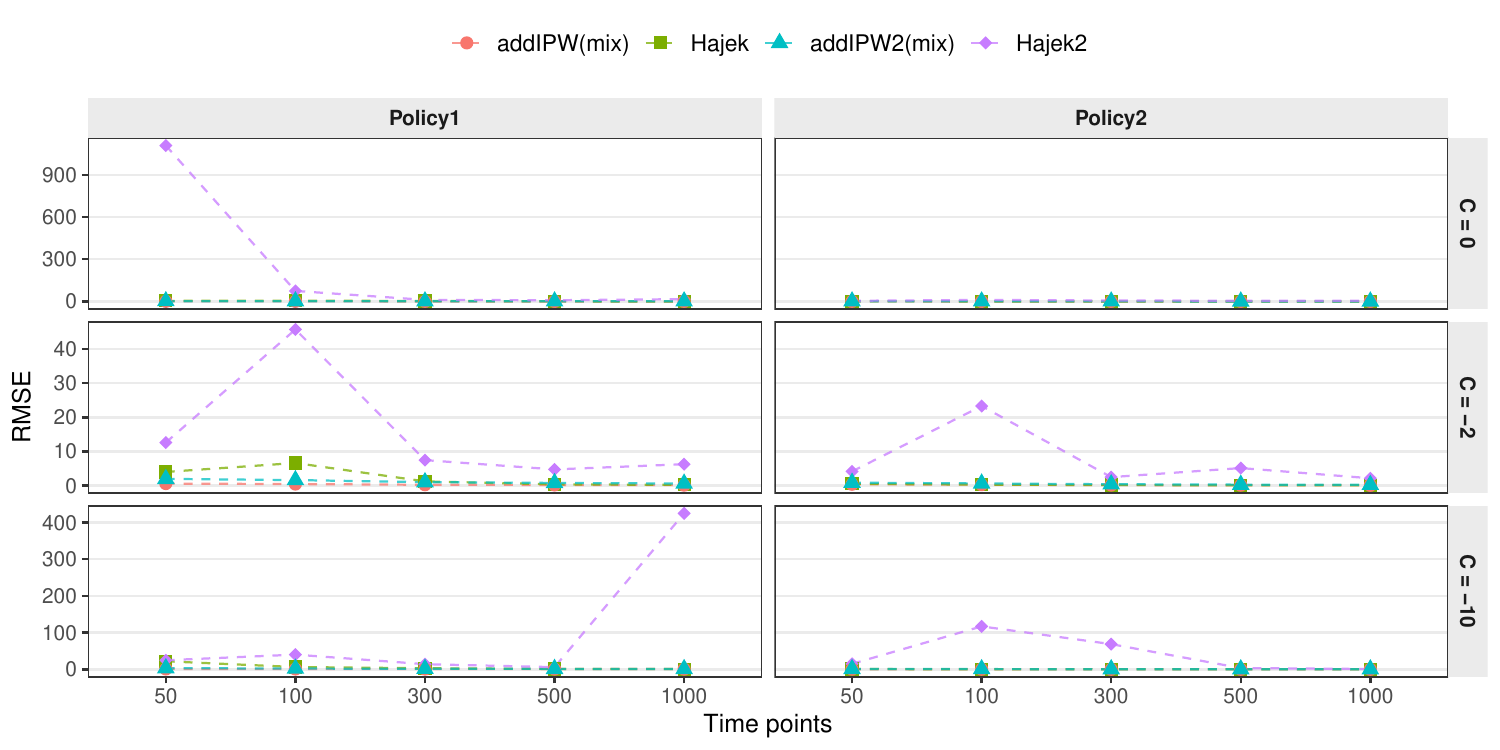}
\label{fig:RMSE_mix_hajek_M2_n10}
}
\caption{Performance of policy evaluation using the mixing estimator and the H\'ajek-type estimator for $n=10$. The plotting convention is the same as in \Cref{fig:mix_hajek_M2}.}
\label{fig:mix_hajek_M2_n10}
\end{figure}

\begin{figure}[!t]
\centering
\subfloat[Estimated policy value]{
\includegraphics[width = \textwidth,trim = 0 0 0 15, clip]{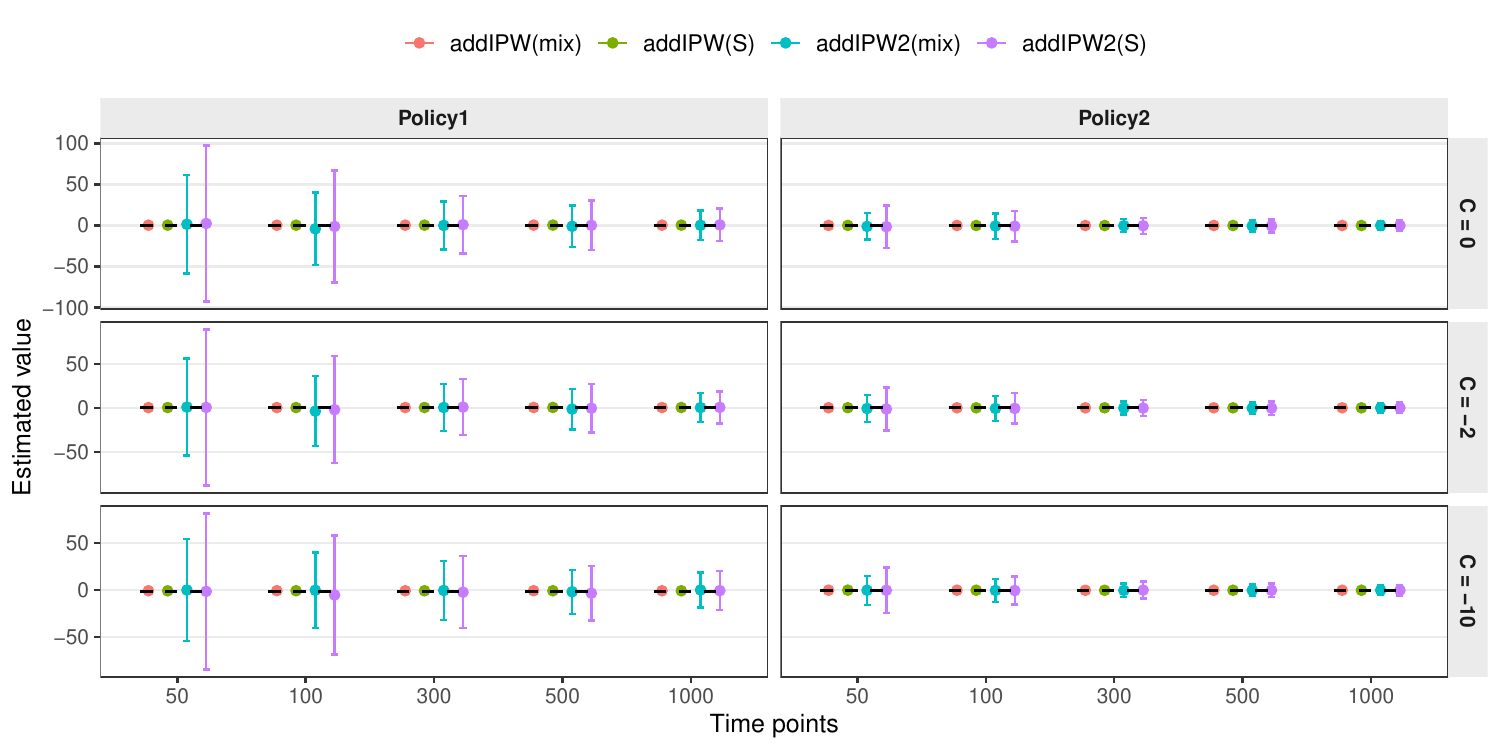}
\label{fig:policy_mix_shift_M2_n100}
}  \\
\subfloat[RMSE]{
\includegraphics[width = \textwidth,trim = 0 0 0 35, clip]{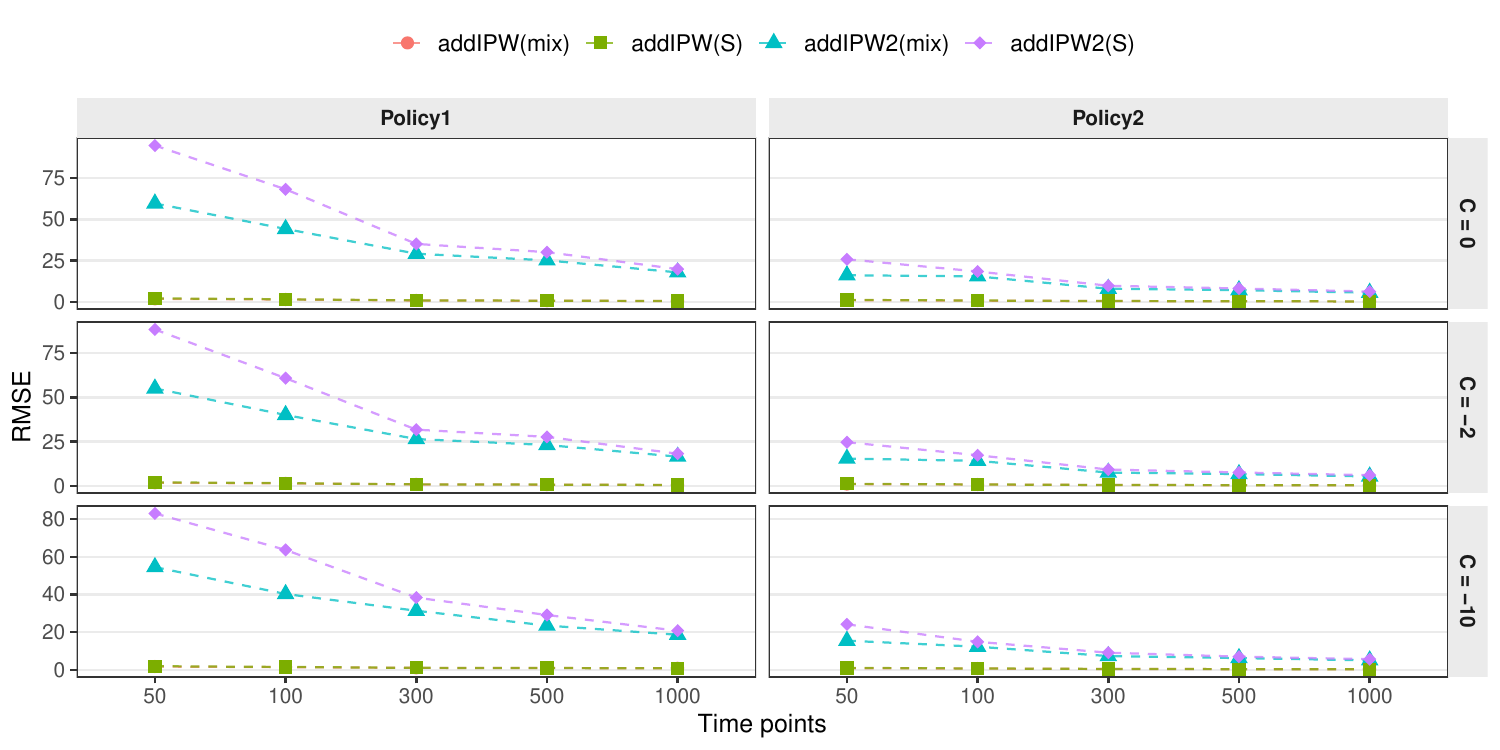}
\label{fig:RMSE_mix_shift_M2_n100}
}
\caption{Performance of policy evaluation using the mixing estimator and the shifted estimator for $n=100$. The plotting convention is the same as in \Cref{fig:mix_shift_M2}.}
\label{fig:mix_shift_M2_n100}
\end{figure}

\begin{figure}[!t]
\centering
\subfloat[Estimated policy value]{
\includegraphics[width = \textwidth,trim = 0 0 0 15, clip]{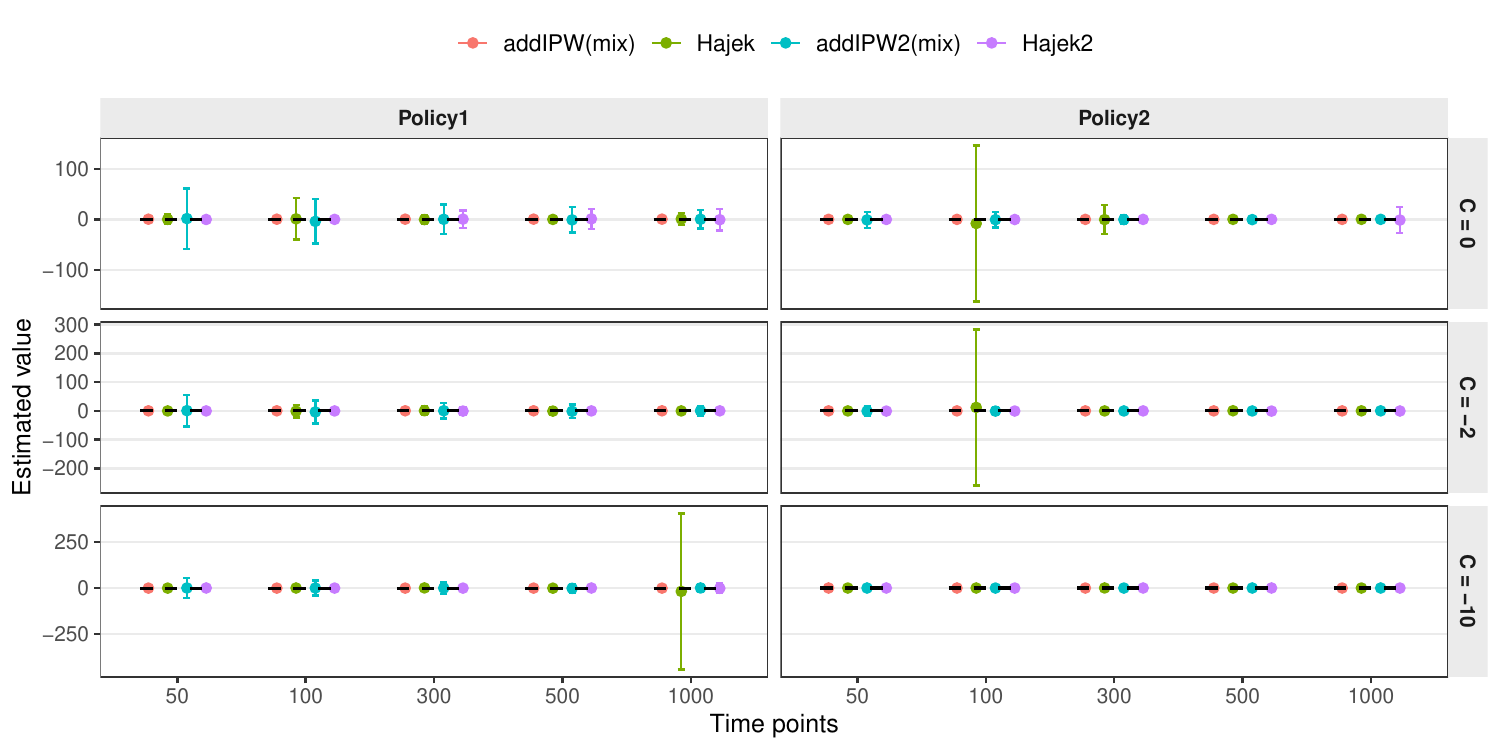}
\label{fig:policy_mix_hajek_M2_n100}
}  \\
\subfloat[RMSE]{
\includegraphics[width = \textwidth,trim = 0 0 0 35, clip]{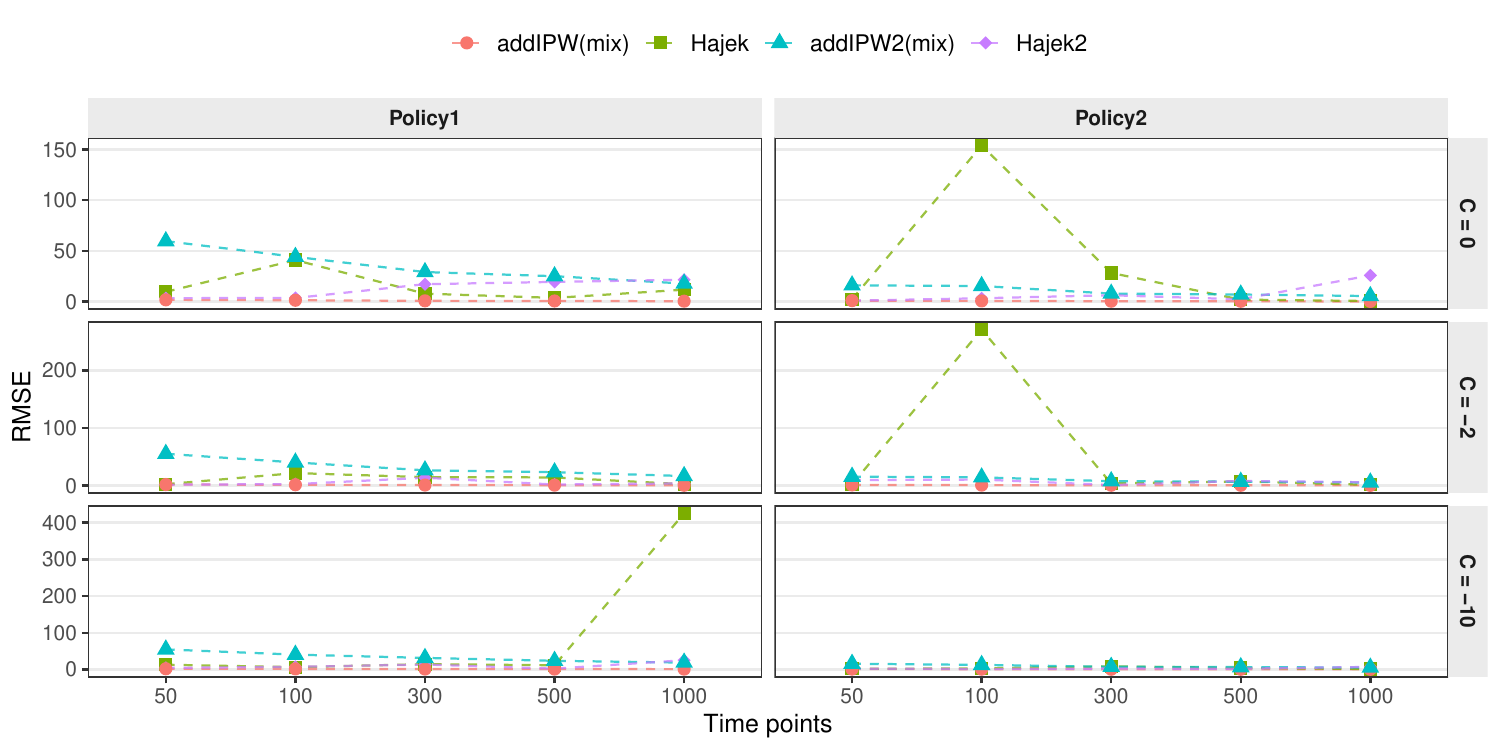}
\label{fig:RMSE_mix_hajek_M2_n100}
}
\caption{Performance of policy evaluation using the mixing estimator and the H\'ajek-type estimator for $n=100$. The plotting convention is the same as in \Cref{fig:mix_hajek_M2}.}
\label{fig:mix_hajek_M2_n100}
\end{figure}

\subsection{Additional results for policy learning}
\label{a:subsec:additional_policy_learning}

Although the H\'ajek estimator based on standard IPW weights is generally not feasible in policy evaluation due to the degeneracy of the average IPW weights toward zero as the number of units $n$ increases, it can still be implemented for policy learning using the stochastic policy approximation introduced in \cref{subsec: find_optimal_policy}. \Cref{fig:sim_res_M1M2_Hajek_compare_estPS} provides the corresponding comparison between the mixing estimators based on additive IPW weights and the standard H\'ajek estimator for $n=50$. The mixing estimators consistently outperform the standard H\'ajek estimator.

We also report policy learning results for $n=10$ and $n=100$ in \Cref{fig:sim_res_M1M2_mix_estPS_n10,fig:sim_res_M1M2_mix_estPS_n100}. The qualitative patterns are similar to those in the main text. When $c=0$ or $c=-2$, the estimators based on lower-order outcome models generally perform well because of their smaller variance. When $c=-10$, estimators allowing second-order interactions become more favorable as $T$ increases. The test-based estimator approaches the best-performing estimator more quickly for $n=10$, where the test has higher power. For $n=100$, the improvement is slower in the strong-interaction setting, reflecting the lower power of the test when the additive IPW weights are more variable.

\begin{figure}[!t]
    \centering
    \includegraphics[width=\linewidth]{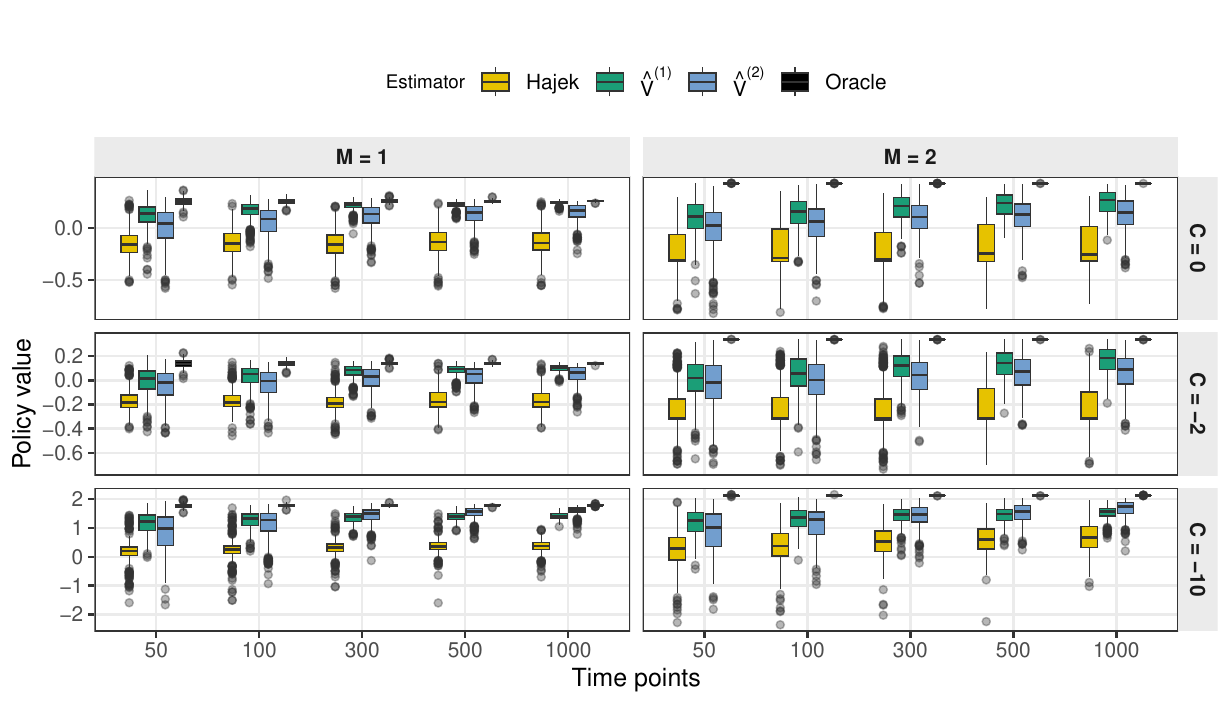}
    \caption{Boxplots of the policy value $V^{\hat\ppi_{\hat\ttheta}}$ under the learned policies for $n=50$, comparing the standard H\'ajek estimator with the mixing estimators based on additive IPW weights. We consider estimators assuming a linear additive outcome model ($\hatVPS[][(\bm 1)]$, green), a quadratic additive outcome model ($\hatVPS[][(\bm 2)]$, blue) for all time periods, the standard H\'ajek estimator (yellow), and the oracle policy (black).}
    \label{fig:sim_res_M1M2_Hajek_compare_estPS}
\end{figure}

\begin{figure}[!t]
    \centering
    \includegraphics[width=\linewidth]{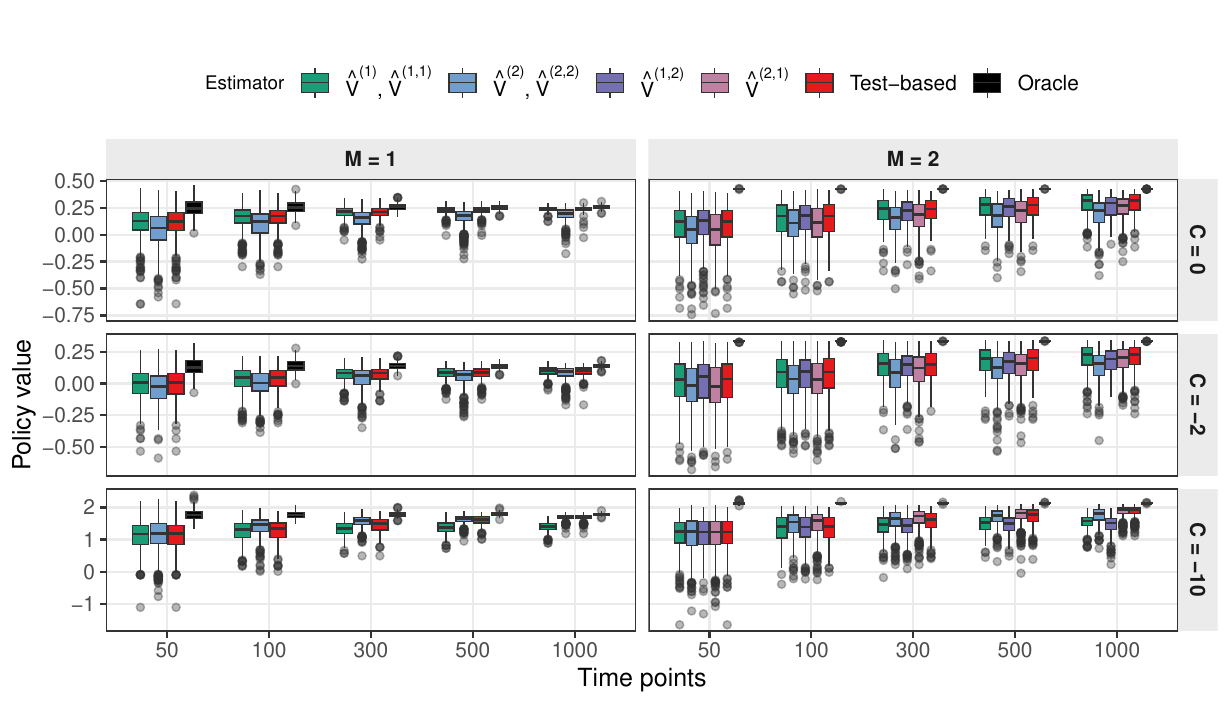}
    \caption{Boxplots of the policy value $V^{\hat\ppi_{\hat\ttheta}}$ under the learned policies for $n=10$, based on different estimators, across different model specifications (rows) and policies of different lengths (columns). The plotting convention is the same as in \Cref{fig:sim_res_M1M2_mix_estPS}.}
    \label{fig:sim_res_M1M2_mix_estPS_n10}
\end{figure}

\begin{figure}[!t]
    \centering
    \includegraphics[width=\linewidth]{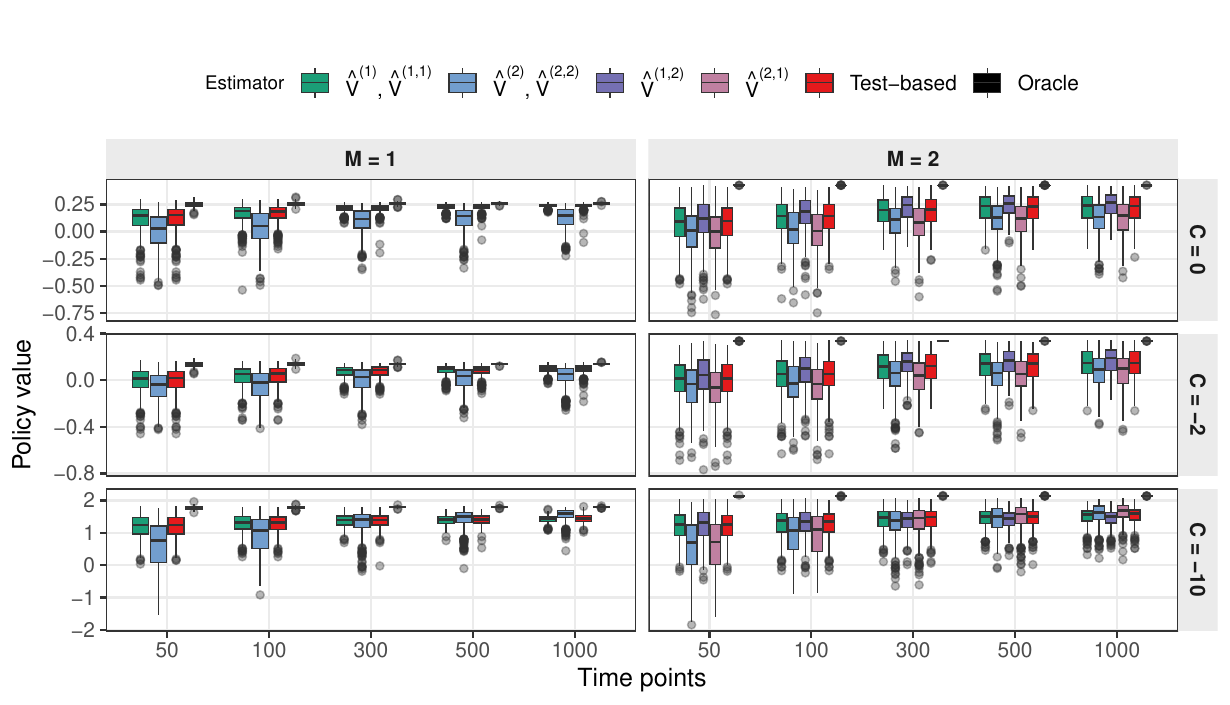}
    \caption{Boxplots of the policy value $V^{\hat\ppi_{\hat\ttheta}}$ under the learned policies for $n=100$, based on different estimators, across different model specifications (rows) and policies of different lengths (columns). The plotting convention is the same as in \Cref{fig:sim_res_M1M2_mix_estPS}.}
    \label{fig:sim_res_M1M2_mix_estPS_n100}
\end{figure}

\begin{figure}[!t]
    \centering
    \includegraphics[width=\linewidth]{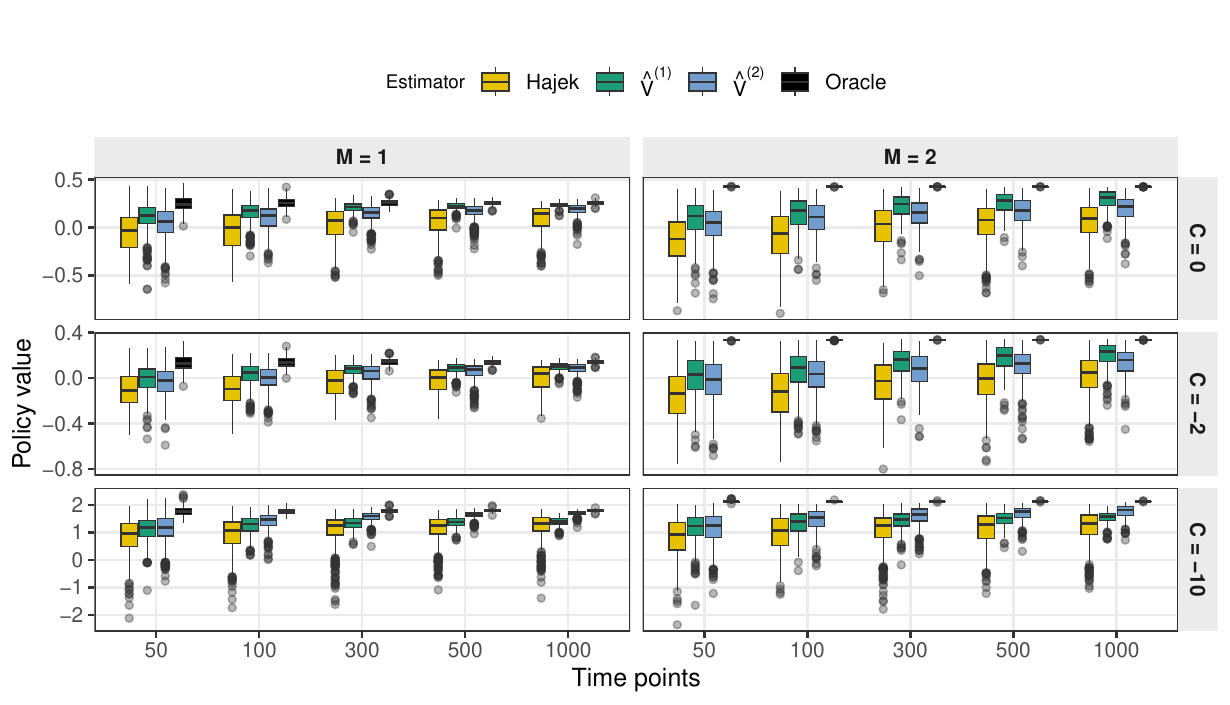}
    \caption{Comparison between the standard H\'ajek estimator and the mixing estimators based on additive IPW weights for policy learning with $n=10$. The plotting convention is the same as in \Cref{fig:sim_res_M1M2_Hajek_compare_estPS}.}
    \label{fig:sim_res_M1M2_Hajek_compare_estPS_n10}
\end{figure}

\begin{figure}[!t]
    \centering
    \includegraphics[width=\linewidth]{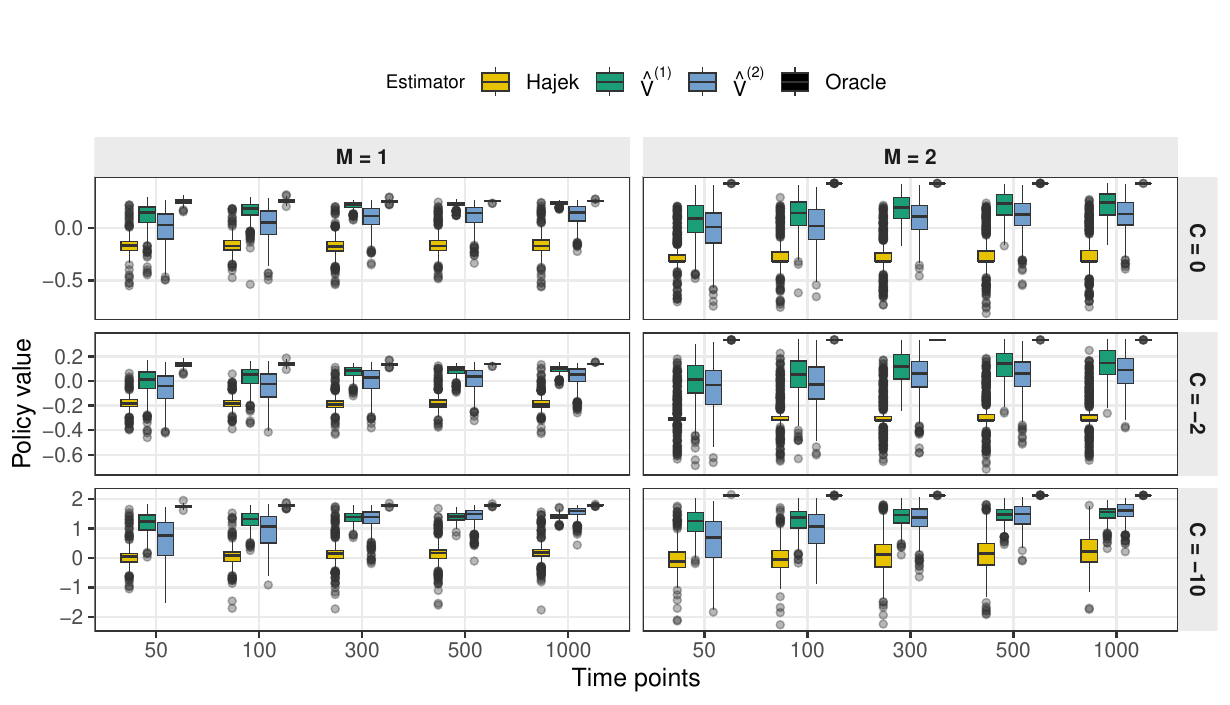}
    \caption{Comparison between the standard H\'ajek estimator and the mixing estimators based on additive IPW weights for policy learning with $n=100$. The plotting convention is the same as in \Cref{fig:sim_res_M1M2_Hajek_compare_estPS}.}
    \label{fig:sim_res_M1M2_Hajek_compare_estPS_n100}
\end{figure}

\subsection{Sensitivity analysis for the number of candidate policies}
\label{a:subsec:L_sensitivity}

The proposed bootstrap test is implemented by sampling $L$ candidate policies from the policy class. The value of $L$ should be large enough so that the sampled policies provide adequate coverage of the policy class, but the simulation results suggest that the performance of the test is not very sensitive to the exact value of $L$ once $L$ is reasonably large.

In \Cref{tab:L_sensitivity}, we report the rejection rate of the bootstrap test for $L\in\{500,1000,2000\}$. The rejection rate is defined as the proportion of simulations in which at least one of the sequential tests rejects the null hypothesis at level $\alpha=0.1$. Across all scenarios, the rejection rates are very similar for the three choices of $L$. For example, when there is no interaction ($c=0$), the rejection rates remain below the chosen $\alpha$ value, and range from approximately 0.004 to 0.073 across all values of $n$, $T$, and $L$. Under moderate interactions ($c=-2$), the rejection rates are also relatively low, with the largest values around 0.11 for $n=10$ and $T=1000$. Under strong interactions ($c=-10$), the rejection rate increases with $T$, especially when $n$ is smaller: for $n=10$, it increases from about 0.04 at $T=50$ to nearly 1 at $T=1000$, while for $n=50$ it increases from about 0.002 at $T=50$ to about 0.78 at $T=1000$. For $n=100$, the rejection rate is lower, reaching about 0.33 at $T=1000$. These results show that increasing $L$ from 500 to 2000 has little effect on the rejection rate, while the power of the test is more strongly affected by the strength of the interaction, the number of time periods, and the variability induced by larger $n$.

\begin{table}[p]
\centering
\caption{Sensitivity analysis for the number of sampled candidate policies $L$. Entries report the rejection rate of the bootstrap test, defined as the proportion of simulations in which at least one sequential test rejects the null hypothesis at level $\alpha=0.1$.}
\label{tab:L_sensitivity}
\fontsize{9}{11}\selectfont
\begin{tabular}[t]{cccrrr}
\toprule
Scenario & $n$ & $T$ & $L=500$ & $L=1000$ & $L=2000$ \\
\midrule
\multirow{15}{*}{$c=0$}
& \multirow{5}{*}{10}
& 50   & 0.028 & 0.028 & 0.030 \\
& & 100  & 0.050 & 0.050 & 0.050 \\
& & 300  & 0.066 & 0.042 & 0.058 \\
& & 500  & 0.058 & 0.073 & 0.066 \\
& & 1000 & 0.066 & 0.073 & 0.072 \\
& \multirow{5}{*}{50}
& 50   & 0.006 & 0.006 & 0.004 \\
& & 100  & 0.026 & 0.030 & 0.030 \\
& & 300  & 0.018 & 0.016 & 0.022 \\
& & 500  & 0.042 & 0.040 & 0.040 \\
& & 1000 & 0.050 & 0.055 & 0.056 \\
& \multirow{5}{*}{100}
& 50   & 0.006 & 0.005 & 0.004 \\
& & 100  & 0.006 & 0.007 & 0.008 \\
& & 300  & 0.025 & 0.028 & 0.023 \\
& & 500  & 0.023 & 0.026 & 0.021 \\
& & 1000 & 0.029 & 0.026 & 0.035 \\
\midrule
\multirow{15}{*}{$c=-2$}
& \multirow{5}{*}{10}
& 50   & 0.022 & 0.020 & 0.018 \\
& & 100  & 0.058 & 0.059 & 0.050 \\
& & 300  & 0.094 & 0.090 & 0.082 \\
& & 500  & 0.084 & 0.092 & 0.084 \\
& & 1000 & 0.114 & 0.106 & 0.106 \\
& \multirow{5}{*}{50}
& 50   & 0.018 & 0.018 & 0.018 \\
& & 100  & 0.032 & 0.030 & 0.028 \\
& & 300  & 0.024 & 0.024 & 0.024 \\
& & 500  & 0.040 & 0.036 & 0.040 \\
& & 1000 & 0.062 & 0.068 & 0.054 \\
& \multirow{5}{*}{100}
& 50   & 0.004 & 0.005 & 0.006 \\
& & 100  & 0.010 & 0.007 & 0.008 \\
& & 300  & 0.033 & 0.033 & 0.031 \\
& & 500  & 0.027 & 0.030 & 0.027 \\
& & 1000 & 0.037 & 0.028 & 0.039 \\
\midrule
\multirow{15}{*}{$c=-10$}
& \multirow{5}{*}{10}
& 50   & 0.042 & 0.045 & 0.044 \\
& & 100  & 0.186 & 0.188 & 0.168 \\
& & 300  & 0.728 & 0.664 & 0.712 \\
& & 500  & 0.932 & 0.933 & 0.932 \\
& & 1000 & 0.998 & 1.000 & 1.000 \\
& \multirow{5}{*}{50}
& 50   & 0.002 & 0.002 & 0.002 \\
& & 100  & 0.012 & 0.012 & 0.014 \\
& & 300  & 0.110 & 0.111 & 0.114 \\
& & 500  & 0.330 & 0.298 & 0.320 \\
& & 1000 & 0.774 & 0.789 & 0.784 \\
& \multirow{5}{*}{100}
& 50   & 0.002 & 0.000 & 0.000 \\
& & 100  & 0.008 & 0.005 & 0.008 \\
& & 300  & 0.037 & 0.042 & 0.041 \\
& & 500  & 0.095 & 0.091 & 0.098 \\
& & 1000 & 0.336 & 0.337 & 0.330 \\
\bottomrule
\end{tabular}
\end{table}

\section{Additional application results}\label{a:sec:additional_app}

\subsection{Covariates balance}\label{a: subsec:balance}

To assess covariate balance after estimation of the propensity score model, we computed the absolute standardized differences  for each  variable.  A widely used rule of thumb considers covariate balance acceptable when the ASD is below 0.1.

Figure \ref{fig:ASD} plots the ASDs for all covariates. Most covariates fall below the 0.1 criterion, indicating that the propensity score model achieves satisfactory balance. However, three variables—total amount of aid received in the previous week, distance to the main route, and distance to river—slightly exceed the 0.1 threshold (by a small margin), suggesting minor imbalance.

Although these deviations indicate that perfect balance was not achieved for these specific covariates, their magnitude remains modest. In practice, slight exceedances of the 0.1 threshold are not uncommon and generally do not compromise causal inference, especially when the overall balance profile is strong. Furthermore, given that these imbalances are small and limited to only a few covariates, we expect them to have minimal impact on our estimated treatment effects.

\begin{figure}
    \centering
    \includegraphics[width=0.8\linewidth]{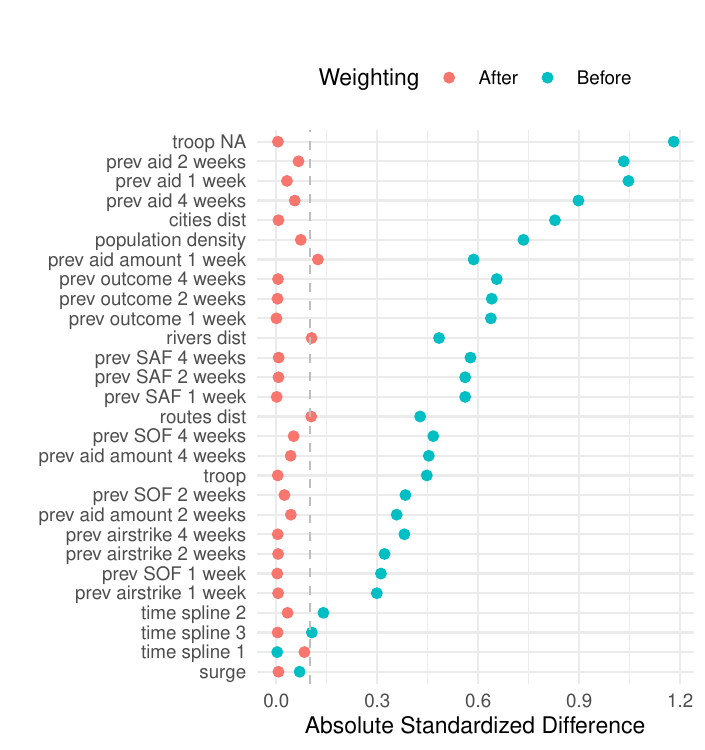}
    \caption{Absolute standardized differences for covariates in the propensity score model. Most covariates have ASDs below the conventional threshold of 0.1. Three variables: total amount of aid received in the previous week, distance to the main route, and distance to river slightly exceed this threshold, indicating minor imbalance}
    \label{fig:ASD}
\end{figure}

\subsection{Learned policy coefficients for the IED analysis}\label{a: subsec:coeff}
For the IED analysis, the covariates are
\[
\bm X^{P}_{ti} = \left(1,W_{t-1,i},\text{Airstrike}_{t-1,i}, \text{IED}_{t-1,i}, \text{Troop}_{t-1,i},\text{Troop (miss.)}_{t-1,i}\right),
\]
and we use a similar specification for the SAF analysis reported in \cref{a:subsec:SAF}. 
\Cref{tab:policy_coefficients_ied} reports the learned policy coefficients across the four policy periods when $M=4$. The table also reports the expected share of treated districts under the learned policy. In the policy learning step, the absolute value of each coefficient is constrained to be no larger than 100. The observed aid allocation treats about 41.4\% of districts on average across weeks. The learned policy treats a smaller share of districts, ranging from 23.7\% to 27.9\% across the four periods.

The coefficient for the log of the number of IED attacks in the previous week is equal to -100 in all four periods. This is the most stable pattern in the table. The coefficient for U.S. troop density is positive throughout. In contrast, the coefficient for previous aid allocation is negative in all four periods, although its magnitude becomes much smaller in the third and fourth periods. The coefficient for the log of the number of airstrikes in the previous week is positive in the first period and negative in the later periods. The missing troop density indicator is negative throughout, with relatively small changes across periods.

Overall, the learned policy for the IED analysis treats fewer districts than the observed allocation while still reducing the estimated policy value. The coefficients suggest that the policy places the strongest weight on recent IED activity, together with troop density, airstrikes, previous aid allocation, and missing troop density information.

\begin{table}[!h]

\caption{Learned policy coefficients across four time periods ($M=4$), with the expected share of treated districts and coefficients for the policy covariates in the IED analysis. Here, W denotes previous aid allocation, IED denotes the number of IED attacks in the previous week, Troop denotes U.S. troop density, Troop (miss.) denotes an indicator for missing troop density data, and Airstrikes denotes the number of airstrikes in the previous week.}
\label{tab:policy_coefficients_ied}
\centering
 \resizebox{\textwidth}{!}{%
\begin{tabular}[t]{llrrrrrr}
\toprule
Lag & \% treated & Intercept & W & Troop & Airstrike (log) & Troop (miss.) & IED (log)\\
\midrule
1 & 27.9 & -17.693 & -8.784 & 32.113 & 1.086 & -4.868 & -100\\
2 & 26.7 & -17.475 & -8.689 & 31.714 & -11.773 & -3.141 & -100\\
3 & 25.2 & -29.922 & -0.234 & 36.354 & -12.314 & -3.188 & -100\\
4 & 23.7 & -29.524 & -0.631 & 36.178 & -12.051 & -2.578 & -100\\
\bottomrule
\end{tabular}
}
\end{table}

\subsection{SAF results}\label{a:subsec:SAF}

This section reports the SAF analysis using the simpler policy class without the population terms. It complements the IED results in Section \ref{subsec:app_res}. The policy covariates include previous aid allocation, U.S. troop density, the log of the number of airstrikes in the previous week, the missing troop density indicator, and the log of the number of SAF attacks in the previous week.

\Cref{fig:SAF_results} shows the spatial results for $M=1,2,3,4$. The learned policy assigns more aid to Anbar across all four horizons. The largest decreases in aid are mainly concentrated in Salah al-Din and Diyala. Ninewa and Basrah also appear among the governorates with larger reductions for some policy horizons. This pattern is fairly stable across $M$, although the set of governorates with the largest decreases changes slightly as the policy horizon increases.

The outcome results show that the largest reductions in expected SAF attacks are concentrated in Baghdad and Diyala. For $M=4$, Baghdad has the most pronounced reduction, about 57 fewer SAF attacks, and Diyala also shows a large reduction of about 20. Salah al-Din also has a visible reduction for $M=4$. The largest positive changes are much smaller than the main reductions, and their confidence intervals include zero.

\Cref{tab:policy_coefficients_saf} reports the learned policy coefficients across the four policy periods when $M=4$. The expected share of treated districts ranges from 13.7\% to 33.8\%. The coefficient for previous aid allocation is positive in all four periods, which suggests persistence in aid assignment. The coefficient for the log of the number of SAF attacks in the previous week is negative throughout, with especially large negative values in the middle periods. The coefficient for U.S. troop density is positive across all periods. The coefficient for the log of the number of airstrikes in the previous week is negative throughout.

The coefficient for the missing troop density indicator changes sign across periods. It is negative in the first and fourth periods, but positive in the second and third periods. Overall, the simpler SAF analysis gives results that are consistent with the main spatial pattern. The learned policy shifts aid toward Anbar, reduces aid in several governorates with high observed allocation, and is associated with lower expected SAF attacks, especially in Baghdad and Diyala.

\begin{figure}[!t]
\centering
\includegraphics[width = \textwidth,trim = 0 0 0 0, clip]{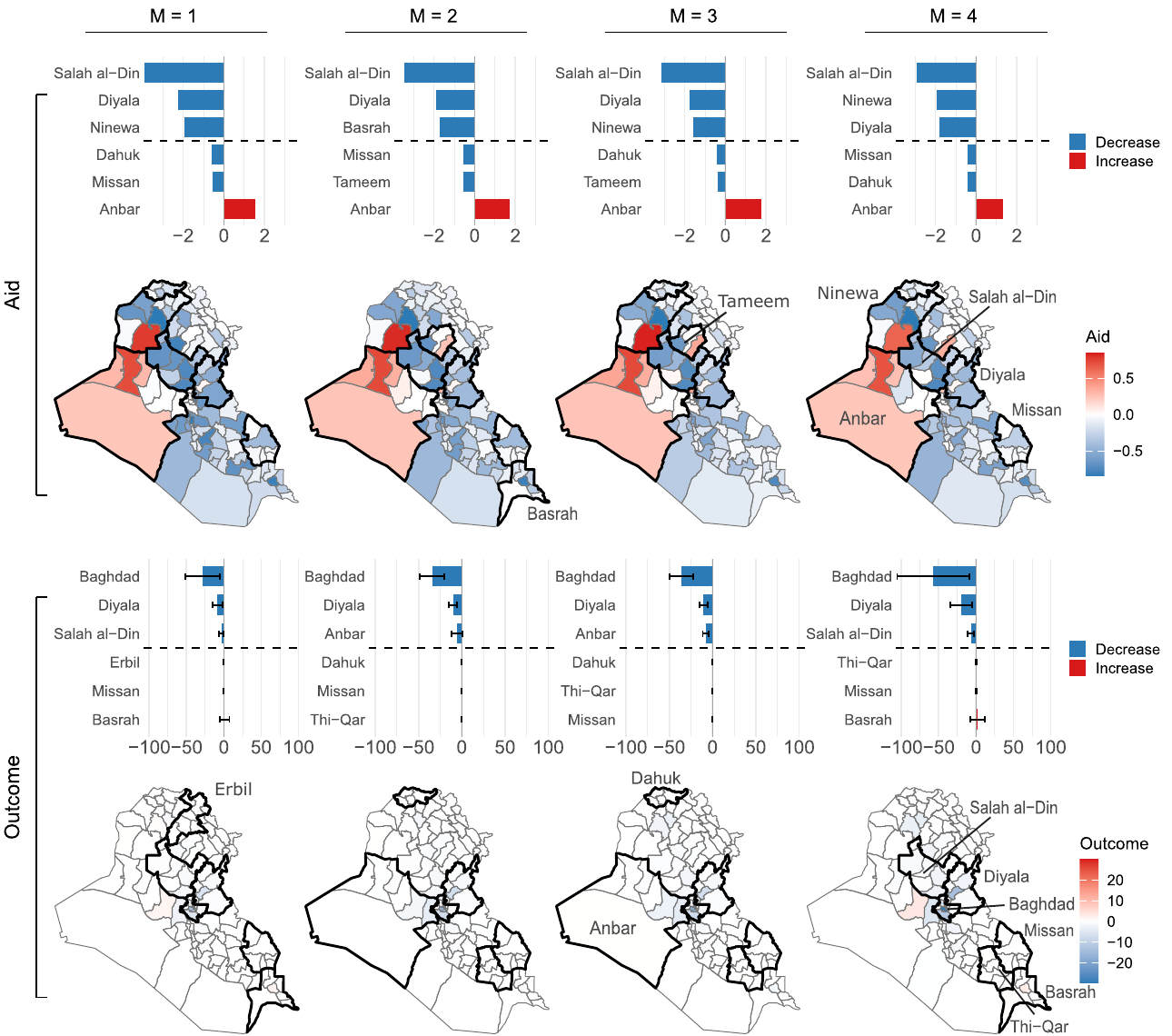}
\caption{Results for the analysis of SAF attacks. The four columns correspond to policies learned over horizons of $M=1,2,3,4$ weeks. The first two rows display the difference in average aid assignment between the learned policy and the observed allocation, while the last two rows show the difference in the expected number of IED attacks relative to observed outcomes. The bar plots highlight the three governorates with the largest decreases and the three with the smallest (or most negative) decreases, aggregated at the governorate level. For the outcome differences, the 95\% confidence intervals are also shown.} 
\label{fig:SAF_results}
\end{figure}

\begin{table}[!h]

\caption{Learned policy coefficients across four time periods ($M=4$), with the expected share of treated districts and coefficients for the policy covariates in the SAF analysis. Here, W denotes previous aid allocation, Troop denotes U.S. troop density, Airstrike (log) denotes the log-transformed number of airstrikes in the previous week, Troop (miss.) denotes an indicator for missing troop density data, and SAF (log) denotes the log transformed number of SAF attacks in the previous week.}
\label{tab:policy_coefficients_saf}
\centering
 \resizebox{\textwidth}{!}{%
\begin{tabular}[t]{llrrrrrr}
\toprule
Lag & \% treated & Intercept & W & Troop & Airstrike (log) & Troop (miss.) & SAF (log)\\
\midrule
1 & 20.1 & -100.000 & 59.290 & 21.416 & -89.150 & -1.309 & -93.363\\
2 & 26.4 & -76.192 & 42.489 & 40.300 & -50.268 & 7.901 & -100.000\\
3 & 33.8 & -86.411 & 53.934 & 38.322 & -83.395 & 43.805 & -100.000\\
4 & 13.7 & -100.000 & 54.601 & 9.744 & -51.225 & -26.464 & -81.042\\
\bottomrule
\end{tabular}
}
\end{table}

\subsection{Analysis with additional policy covariates}\label{a:subsec:morecovariates}

In this section, we repeat the application analysis with a larger policy class. The added covariates are log population density and the interaction term of log population density and the log of the number of attacks in the previous week. The interaction term uses IED attacks for the IED analysis and SAF attacks for the SAF analysis. The interaction term of the previous outcome and the missing troop density indicator is not included.

Table \ref{tab:policy_res_wpop} reports the estimated policy values for $M=1,2,3,4$. In the IED analysis, the estimated value is 1.64 for $M=1$, compared with the observed value of 2.215. For $M=4$, the estimated value is 6.17, compared with the observed value of 8.860. The 95\% confidence interval is [5.66, 6.68]. In the SAF analysis, the estimated value is 1.72 for $M=1$, compared with the observed value of 2.230. For $M=4$, the estimated value is 4.99, compared with the observed value of 8.919. The 95\% confidence interval is [1.21, 8.76]. Across both outcomes and all four values of $M$, the estimated policy values are lower than the observed values.

Figures \ref{fig:IED_results_wpop} and \ref{fig:SAF_results_wpop} show the spatial patterns from this larger policy class. For the IED analysis, the learned policy assigns more aid to Anbar across all four horizons. The increase ranges from 1.81 to 1.87. The largest reductions in aid appear in Salah al-Din, Basrah, and Babylon. For $M=4$, the changes are -3.36 in Salah al-Din, -1.67 in Basrah, and -1.24 in Babylon. The outcome results show the largest reductions in Baghdad, Salah al-Din, and Diyala. For $M=4$, the estimated changes are -28.5 in Baghdad, -15.8 in Salah al-Din, and -6.59 in Diyala. The largest positive change is much smaller, with Missan increasing by 0.24.

For the SAF analysis, Anbar also receives more aid under the learned policy. The increase is 1.55 for $M=1$ and 1.33 for $M=4$. Aid is reduced in Salah al-Din for all four horizons, with changes from -3.89 to -3.04. Diyala and Basrah also receive less aid for most horizons, and Ninewa appears among the largest reductions when $M=4$. The outcome results show the largest reductions in Baghdad and Diyala. For $M=4$, the estimated change is -56.6 in Baghdad and -19.2 in Diyala. Ninewa also decreases by -6.79, although its confidence interval includes zero. The largest positive changes are small relative to these reductions. Basrah increases by 2.04 for $M=4$, with a confidence interval that includes zero.

Tables \ref{tab:policy_coefficients_ied_wpop} and \ref{tab:policy_coefficients_saf_wpop} report the learned policy coefficients across the four time periods when $M=4$. In the IED analysis, the expected treated share decreases from 27.9\% in the first period to 23.7\% in the fourth period. The coefficient for the log of IED attacks in the previous week is at or very close to -100 in all four periods. The coefficient for U.S. troop density is positive throughout. The coefficient for log population density is positive in the first two periods, 3.516 and 3.189, and negative in the last two periods, -1.090 and -2.289. The interaction term of the log of IED attacks in the previous week and log population density is negative in all four periods, ranging from -18.347 to -12.965.

For the SAF analysis, the expected treated share is 19.8\%, 26.4\%, 28.5\%, and 13.7\% across the four periods. The coefficient for previous aid allocation is positive in all four periods, with values between 76.024 and 88.199. The coefficient for the log of SAF attacks in the previous week is -100.000 in the first three periods and -60.913 in the fourth period. The coefficient for U.S. troop density is positive in all four periods. The coefficient for the log of airstrikes in the previous week is negative throughout. The coefficients for log population density are also negative in all four periods. The interaction term of the log of SAF attacks in the previous week and log population density is negative in all four periods. Overall, adding the population terms changes some coefficient patterns, especially in the IED analysis, but the estimated policy values and spatial summaries still show lower expected outcomes under the learned policies than under the observed allocation.

\begin{table}[!h]
\centering
\caption{Estimated policy values ($\hat V$) and corresponding 95\% confidence
intervals for the learned optimal policies in the IED and SAF analyses
incorporating additional policy covariates, for $M=1,2,3,4$. Observed
denotes the average number of attacks across all time periods and spatial
units.}
\label{tab:policy_res_wpop}

\fontsize{10.5}{12}\selectfont
\setlength{\tabcolsep}{5pt}

\begin{tabular}[t]{crrrrrrrr}
\toprule
\multicolumn{1}{c}{ } & \multicolumn{4}{c}{IED} & \multicolumn{4}{c}{SAF} \\
\cmidrule(l{3pt}r{3pt}){2-5} \cmidrule(l{3pt}r{3pt}){6-9}
$M$ & $\hat V$ & Lower & Upper & Observed & $\hat V$ & Lower & Upper & Observed\\
\midrule
1 & 1.64 & 1.38 & 1.89 & 2.215 & 1.72 & 1.23 & 2.21 & 2.230\\
2 & 3.40 & 3.05 & 3.74 & 4.430 & 3.15 & 2.53 & 3.76 & 4.460\\
3 & 4.82 & 4.39 & 5.26 & 6.645 & 4.42 & 3.55 & 5.30 & 6.689\\
4 & 6.17 & 5.66 & 6.68 & 8.860 & 4.99 & 1.21 & 8.76 & 8.919\\
\bottomrule
\end{tabular}
\end{table}

\begin{figure}[!t]
\centering
\includegraphics[width = \textwidth,trim = 0 0 0 0, clip]{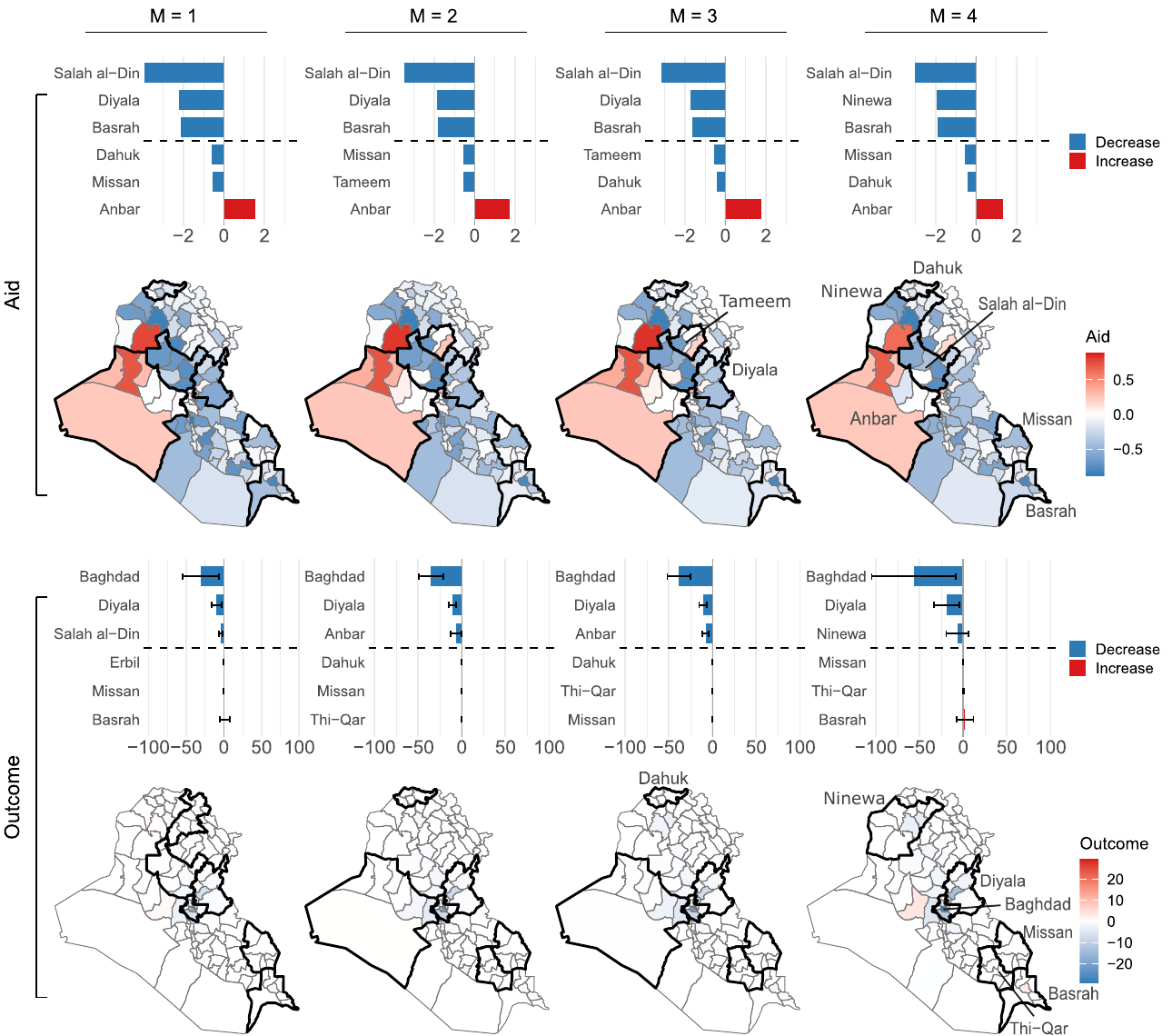}
\caption{Results for the analysis of SAF attacks which more policy covariates. The four columns correspond to policies learned over horizons of $M=1,2,3,4$ weeks. The first two rows display the difference in average aid assignment between the learned policy and the observed allocation, while the last two rows show the difference in the expected number of IED attacks relative to observed outcomes. The bar plots highlight the three governorates with the largest decreases and the three with the smallest (or most negative) decreases, aggregated at the governorate level. For the outcome differences, the 95\% confidence intervals are also shown.} 
\label{fig:SAF_results_wpop}
\end{figure}

\begin{figure}[!t]
\centering
\includegraphics[width = \textwidth,trim = 0 0 0 0, clip]{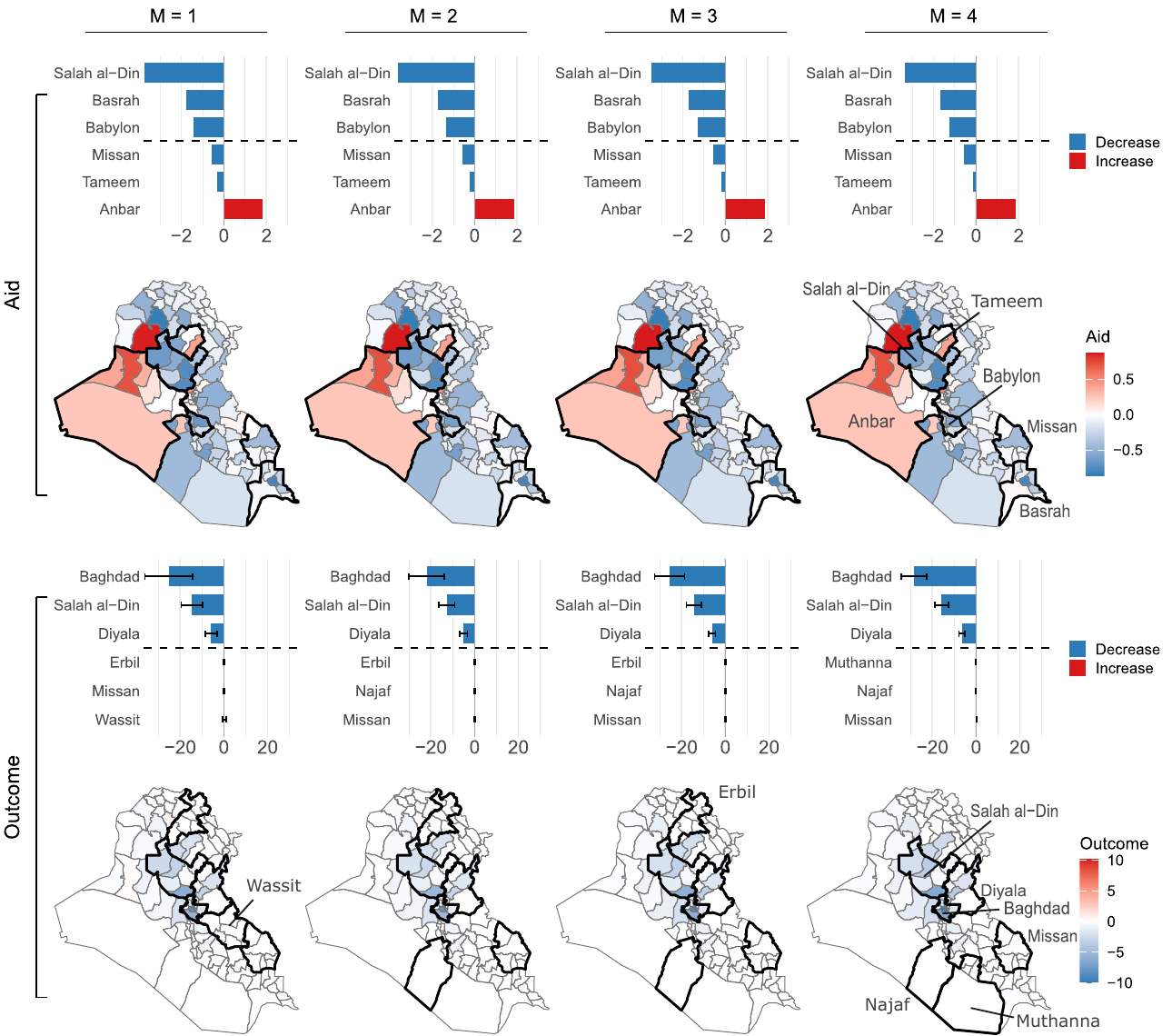}
\caption{Results for the analysis of IED attacks which includes more policy covariates. The four columns correspond to policies learned over horizons of $M=1,2,3,4$ weeks. The first two rows display the difference in average aid assignment between the learned policy and the observed allocation, while the last two rows show the difference in the expected number of SAF attacks relative to observed outcomes. The bar plots highlight the three governorates with the largest decreases and the three with the smallest (or most negative) decreases, aggregated at the governorate level. For the outcome differences, the 95\% confidence intervals are also shown.} 
\label{fig:IED_results_wpop}
\end{figure}

\begin{table}[!h]
  \centering
  \footnotesize
  \setlength{\tabcolsep}{3pt}
  \caption{Learned policy coefficients across four time periods ($M=4$), with the expected share of treated districts and coefficients for the expanded policy covariates in the IED analysis. Here, W denotes previous aid allocation, Troop denotes U.S. troop density, Airstrike (log) denotes the log of the number of airstrikes in the previous week, Troop (miss.) denotes an indicator for missing troop density, IED (log) denotes the log of the number of IED attacks in the previous week, Pop denotes log population density, and IED (log) x Pop denotes the interaction between IED attacks in the previous week and log population density.}
  \label{tab:policy_coefficients_ied_wpop}
  \resizebox{\textwidth}{!}{%
\begin{tabular}[t]{llrrrrrrrr}
\toprule
Lag & \% treated & Intercept & W & Troop & Airstrike (log) & Troop (miss.) & IED (log) & Pop & IED (log) x Pop\\
\midrule
1 & 27.9 & -100.000 & -21.944 & 100.000 & 0.302 & -4.869 & -100.000 & 3.516 & -18.190\\
2 & 26.9 & -96.204 & -21.716 & 100.000 & -35.922 & -3.152 & -99.951 & 3.189 & -18.347\\
3 & 25.2 & -68.701 & -0.753 & 100.000 & -28.935 & -3.188 & -100.000 & -1.090 & -15.904\\
4 & 23.7 & -42.543 & -0.672 & 87.234 & -25.672 & -2.578 & -100.000 & -2.289 & -12.965\\
\bottomrule
    \end{tabular}%
  }
\end{table}

\begin{table}[!h]
\footnotesize
  \setlength{\tabcolsep}{3pt}
\caption{Learned policy coefficients across four time periods ($M=4$), with the expected share of treated districts and coefficients for the expanded policy covariates in the SAF analysis. Here, W denotes previous aid allocation, Troop denotes U.S. troop density, Airstrike (log) denotes the log-transformed number of airstrikes in the previous week, Troop (miss.) denotes an indicator for missing troop density data, SAF (log) denotes the log transformed number of SAF attacks in the previous week, Pop denotes log population density, and SAF (log) x Pop denotes the interaction between previous SAF attacks and log population density.}
\label{tab:policy_coefficients_saf_wpop}
\centering
\resizebox{\textwidth}{!}{%
\begin{tabular}[t]{llrrrrrrrr}
\toprule
Lag & \% treated & Intercept & W & Troop & Airstrike (log) & Troop (miss.) & SAF (log) & Pop & SAF (log) x Pop\\
\midrule
1 & 19.8 & -49.747 & 80.167 & 28.548 & -100.000 & -1.309 & -100.000 & -7.499 & -2.049\\
2 & 26.4 & -99.743 & 76.024 & 67.532 & -85.105 & 7.901 & -100.000 & -2.530 & -5.944\\
3 & 28.5 & -10.718 & 88.083 & 56.104 & -100.000 & 60.841 & -100.000 & -11.101 & -5.537\\
4 & 13.7 & -100.000 & 88.199 & 17.038 & -81.397 & -26.464 & -60.913 & -6.149 & -7.057\\
\bottomrule
\end{tabular}
}
\end{table}

\end{document}